\documentclass[reqno,10pt,a4paper,dvips]{amsart}

\usepackage{amssymb}
\usepackage{mathptmx}
\usepackage{cite}
\usepackage{psfrag}
\usepackage{array}
\usepackage{setspace}
\usepackage{geometry}
\usepackage{color}
\usepackage[dvips]{graphicx}
\usepackage{amscd}
\usepackage{enumitem}

\numberwithin{equation}{section}
\newcolumntype{C}{>{$}c<{$}} 
\newcolumntype{R}{>{$}r<{$}} 
\allowdisplaybreaks

\newcommand{\bC} {\mathbb{C}}  
\newcommand{\bR} {\mathbb{R}}  
\newcommand{\bZ} {\mathbb{Z}}  
\newcommand{\bQ} {\mathbb{Q}}  
\newcommand{\sH} {\mathcal{H}} %
\newcommand{\sJ} {\mathcal{J}} %
\newcommand{\sS} {\mathcal{S}} %
\newcommand{\sM} {\mathcal{M}} %
\newcommand{\sN} {\mathcal{N}} %
\newcommand{\sK} {\mathcal{K}} %
\DeclareMathOperator{\Imag}{Im}
\DeclareMathOperator{\Kern}{Ker}

\DeclareMathOperator{\dmn}{dim}

\newcommand{\vir} {\mathfrak{vir}}

\newcommand{\uea} {\mathcal{U}} %
\newcommand{\Verma} {\mathcal{V}}

\newcommand{\Stagg} {\mathcal{S}}

\newcommand{\isom} {\cong}

\newcommand{\lft} {\mathrm{L}}
\newcommand{\rgt} {\mathrm{R}}
\newcommand{\Hlft} {\sH^\lft}
\newcommand{\Hrgt} {\sH^\rgt}

\newcommand{\minsubmod} {\sM}

\newcommand{\data} {\Omega}
\newcommand{\gauge} {G}

\newcommand{\grp}[1]{\mathsf{#1}}
\newcommand{\Ext}[4]{\grp{Ext}^{#1}_{#2} \left( #3 , #4 \right)}

\newcommand{\func}[2]{#1 \left( #2 \right)}
\newcommand{\tfunc}[2]{#1 \bigl( #2 \bigr)}

\newcommand{\brac}[1]{\left( #1 \right)}
\newcommand{\sqbrac}[1]{\left[ #1 \right]}
\newcommand{\set}[1]{\left\{ #1 \right\}}
\newcommand{\tset}[1]{\bigl\{ #1 \bigr\}}
\newcommand{\st}{\mspace{5mu} : \mspace{5mu}}

\newcommand{\abs}[1]{\left| #1 \right|}

\newcommand{\inner}[2]{\left\langle #1 , #2 \right\rangle}
\newcommand{\tinner}[2]{\bigl\langle #1 , #2 \bigr\rangle}

\newcommand{\ZZ}{\mathbb{Z}}
\newcommand{\NN}{\mathbb{N}}
\newcommand{\QQ}{\mathbb{Q}}

\newcommand{\CC}{\mathbb{C}}

\newcommand{\wun}{\mathbf{1}}

\newcommand{\ii}{\mathfrak{i}}

\newcommand{\eps}{\varepsilon}

\newcommand{\comm}[2]{\bigl[ #1 , #2 \bigr]}

\newcommand{\VerMod}{\mathcal{V}}
\newcommand{\IrrMod}{\mathcal{L}}
\newcommand{\StagMod}{\mathcal{S}}
\newcommand{\QuotMod}{\mathcal{Q}}

\newcommand{\LMod}{\Hlft}
\newcommand{\RMod}{\Hrgt}
\newcommand{\LDim}{h^{\lft}}
\newcommand{\RDim}{h^{\rgt}}
\newcommand{\Lhws}{x^{\lft}}
\newcommand{\Rhws}{x^{\rgt}}

\newcommand{\partnum}[1]{\func{p}{#1}}
\newcommand{\tpartnum}[1]{\tfunc{p}{#1}}

\newcommand{\dses}[5]{0 \longrightarrow #1 \overset{#2}{\longrightarrow} #3 \overset{#4}{\longrightarrow} #5 \longrightarrow 0}

\newcommand{\eqnref}[1]{Equation~(\ref{#1})}
\newcommand{\eqnDref}[2]{Equations~(\ref{#1}) and (\ref{#2})}

\newcommand{\secref}[1]{Section~\ref{#1}}
\newcommand{\secDref}[2]{Sections~\ref{#1} and \ref{#2}}

\newcommand{\figref}[1]{Figure~\ref{#1}}

\newcommand{\thmref}[1]{Theorem~\ref{#1}}
\newcommand{\thmDref}[2]{Theorems~\ref{#1} and \ref{#2}}
\newcommand{\lemref}[1]{Lemma~\ref{#1}}
\newcommand{\propref}[1]{Proposition~\ref{#1}}
\newcommand{\propDref}[2]{Propositions~\ref{#1} and \ref{#2}}
\newcommand{\propTref}[3]{Propositions~\ref{#1}, \ref{#2} and \ref{#3}}
\newcommand{\corref}[1]{Corollary~\ref{#1}}
\newcommand{\exref}[1]{Example~\ref{#1}}
\newcommand{\exDref}[2]{Examples~\ref{#1} and \ref{#2}}
\newcommand{\exTref}[3]{Examples~\ref{#1}, \ref{#2} and \ref{#3}}

\newcommand{\cft}{conformal field theory}
\newcommand{\cfts}{conformal field theories}
\newcommand{\uealg}{universal enveloping algebra}
\newcommand{\uealgs}{universal enveloping algebras}
\newcommand{\lcft}{logarithmic conformal field theory}
\newcommand{\lcfts}{logarithmic conformal field theories}
\newcommand{\hws}{highest weight vector}
\newcommand{\hwss}{highest weight vectors}
\newcommand{\hwm}{highest weight module}
\newcommand{\hwms}{highest weight modules}
\newcommand{\sv}{singular vector}
\newcommand{\svs}{singular vectors}

\newcommand{\PBW}{Poincar\'{e}-Birkhoff-Witt}

\newcommand{\nbrBeta}{\mathbf{b}}
\newcommand{\nbrGamma}{\mathbf{g}}
\newcommand{\Rnull}{\mathbf{n}}
\newcommand{\rnk}{r}
\newcommand{\rnkO}{\rho}
\newcommand{\rnkR}{\overline{\rnkO}}
\newcommand{\ghwm}{\mathcal{K}}
\newcommand{\ghw}{h}
\newcommand{\ghwv}{\tilde{x}}
\newcommand{\gell}{l}

\newcommand{\gellm}{l^-}
\newcommand{\gpsi}{\Psi}

\newcommand{\gX}{\chi}
\newcommand{\gXoth}{\chi^-}

\newcommand{\OthStagMod}{\mathcal{T}}
\newcommand{\LOd}{\mathsf{L}_0^d}
\newcommand{\LOn}{\mathsf{L}_0^n}

\newcommand{\betabar}{\overline{\beta}}

\newcommand{\msmbar}{\overline{\minsubmod}}
\newcommand{\ellbar}{\overline{\ell}}
\newcommand{\omegabar}{\overline{\omega}}

\DeclareMathOperator{\rank}{rank}

\DeclareMathOperator*{\dispotimes}{\otimes}

\theoremstyle{plain}
\newtheorem{theorem}{Theorem}[section]
\newtheorem{lemma}[theorem]{Lemma}
\newtheorem{proposition}[theorem]{Proposition}
\newtheorem{corollary}[theorem]{Corollary}

\newtheorem{example}{Example}

\makeatletter
\newcommand\@dotsep{4.5}
\def\@tocline#1#2#3#4#5#6#7{\relax
  \ifnum #1>\c@tocdepth 
  \else
    \par \addpenalty\@secpenalty\addvspace{#2}%
    \begingroup \hyphenpenalty\@M
    \@ifempty{#4}{%
      \@tempdima\csname r@tocindent\number#1\endcsname\relax
    }{%
      \@tempdima#4\relax
    }%
    \parindent\z@ \leftskip#3\relax
    \advance\leftskip\@tempdima\relax
    \rightskip\@pnumwidth plus1em \parfillskip-\@pnumwidth
    #5\leavevmode\hskip-\@tempdima #6\relax
    \leaders\hbox{$\m@th
      \mkern \@dotsep mu\hbox{.}\mkern \@dotsep mu$}\hfill
    \hbox to\@pnumwidth{\@tocpagenum{#7}}\par
    \nobreak
    \endgroup
  \fi}
\makeatother

\begin{document}

\title[Staggered Virasoro Modules]{On Staggered Indecomposable Virasoro Modules}

\author[K Kyt\"ol\"a]{Kalle Kyt\"ol\"a}

\address[Kalle Kyt\"ol\"a]{
Section de Math\'ematiques, Universit\'e de Gen\`eve \\
2-4 Rue du Li\`evre \\
CP 64, 1211 Gen\`eve 4, Switzerland
}

\email{kalle.kytola@math.unige.ch}

\author[D Ridout]{David Ridout}

\address[David Ridout]{
Theory Group, DESY \\
Notkestra\ss{}e 85 \\
D-22603, Hamburg, Germany
}

\email{dridout@mail.desy.de}

\thanks{\today \\ This work was supported by the Swiss National Science Foundation, the European Research Council and the Marie Curie Excellence Grant MEXT-CT-2006-042695.}

\begin{abstract}
In this article, certain indecomposable Virasoro modules are studied.  Specifically, the Virasoro mode $L_0$ is assumed to be non-diagonalisable,
possessing Jordan blocks of rank two.  Moreover, the module is further assumed to have a highest weight submodule, the ``left module'', and that the quotient by this submodule yields another highest weight module, the ``right module''.  Such modules, which have been called \emph{staggered}, have appeared repeatedly in the logarithmic conformal field theory literature, but their theory has not been explored in full generality.  Here, such a theory is developed for the Virasoro algebra using rather elementary techniques.  The focus centres on two different but related questions typically encountered in practical studies:  How can one identify a given staggered module, and how can one demonstrate the existence of a proposed staggered module.

Given just the values of the highest weights of the left and right modules, themselves subject to simple necessary conditions, invariants are defined which together with the knowledge of the left and right modules uniquely identify a staggered module.  The possible values of these invariants form a vector space of dimension zero, one or two, and the structures of the left and right modules limit the isomorphism classes of the corresponding staggered modules to an affine subspace (possibly empty).  The number of invariants and affine restrictions is purely determined by the structures of the left and right modules.  Moreover, in order to facilitate applications, the expressions for the invariants and restrictions are given by formulae as explicit as possible (they generally rely on expressions for Virasoro singular vectors).  Finally, the text is liberally peppered throughout with examples illustrating the general concepts.  These have been carefully chosen for their physical relevance or for the novel features they exhibit.
\end{abstract}

\maketitle

\onehalfspacing

\tableofcontents

\section{Introduction} \label{sec: intro}

The successes of \cft{}, in particular its applications to condensed matter physics, depended crucially on the theory of \hwms{} of the Virasoro algebra. Such a theory became available in the early eighties, ultimately due to the work of Kac \cite{KacCon79} and Feigin and Fuchs \cite{FeiSke82}.  The corresponding \cfts{}, the minimal models of \cite{BelInf84}, are constructed from a certain finite collection of irreducible highest weight Virasoro modules and rightly enjoy their position as some of the simplest and most useful of \cfts{}.

In spite of this, the past fifteen years have witnessed the construction, in varying degrees, of a different kind of \cft{} \cite{RozQua92,GurLog93}.  These theories are constructed from certain indecomposable, rather than irreducible, modules and are collectively known as \lcfts{}.  Despite a promising beginning, logarithmic theories quickly attained a reputation for being esoteric and technical.  Some impressive examples were constructed, but the field suffered from a perceived lack of concrete applications.  To be sure, there were many attempts to use logarithmic theories to explain discrepancies in models of the fractional quantum Hall effect, abelian sandpiles, D-brane recoil and more (see \cite{FloBit03} for references to these), but none of these attempts really left an enduring mark upon their intended field.  Nevertheless, condensed matter physicists remained interested in these theories for the simple reason that the standard minimal model description of many of their favourite models was known to be incomplete or even entirely missing.

Recently, there has been something of a resurgence in the study of
\lcfts{}, with the aim of clarifying applications to condensed
matter physics and developing the mathematical properties of
logarithmic theories so as to more closely mirror those of
standard theories.  One can isolate several different approaches
including free field methods and connections to quantum group
theory \cite{FjeLog02,FeiLog06}, lattice model constructions
\cite{PeaLog06,ReaAss07} and construction through explicit
fusion \cite{EbeVir06,RidPer07}.  All of these involve exploring the
new features of a theory built from indecomposable but reducible modules.  
Intriguingly, recent developments in random conformally 
invariant fractals, Schramm-Loewner evolutions in particular
\cite{SchSca00}, have started to bridge the gap between the
field-theoretic and probabilistic approaches to the statistical
models of condensed matter theory (see \cite{LawCon05,BauGro06,CarSLE05} for reviews).  In particular, the kernel of the infinitesimal generator of the
Schramm-Loewner evolution, which consists of local martingales of the
stochastic growth process that builds the fractal curve, carries
a representation of the Virasoro algebra \cite{BauSLE03,BauCon04,KytVir07},
and it has recently been observed that in certain cases this representation becomes indecomposable, of the type found in \lcft{} \cite{KytFro08}.  This has led to renewed proposals for some sort of SLE-LCFT correspondence \cite{RidLog07,SaiGeo09,SimTwi08}.

Advances such as these have necessitated a better understanding of
the representation theory of the Virasoro algebra beyond \hwms{}. 
In the corresponding logarithmic theories, the Virasoro element $L_0$
acts non-diagonalisably, manifestly demonstrating that more general
classes of modules are required.  One such class consists of the
so-called \emph{staggered} modules and it is these which we will
study in what follows.  More precisely, we will consider indecomposable
Virasoro modules on which $L_0$ acts non-diagonalisably and which
generalise \hwms{} by having a submodule isomorphic to a \hwm{}
such that the quotient by this submodule is again isomorphic to a
\hwm{}.  We refer to the submodule and its quotient as the left
and right module, respectively (the naturality of this nomenclature
will become evident in \secref{secStag}).  Roughly speaking, these
staggered modules can be visualised as two \hwms{} which have been
``glued'' together by a non-diagonalisable action of $L_0$.  Such
staggered modules were first constructed for the Virasoro algebra
in \cite{GabInd96}.

We mention that staggered Virasoro modules corresponding to gluing
more than two \hwms{} together have certainly been considered in the
literature \cite{EbeVir06,RidPer08}, but we shall not do so here. 
Similarly, one could try to develop staggered module theories for
other algebras which arise naturally in \lcfts{}. 
We will leave such studies for future work, noting only that we
expect that the results we are reporting will provide a very useful
guide to the eventual form of these generalisations.  Here, we restrict
ourselves to the simplest case, treating it in as elementary a way as
possible.  We hope that the resulting clarity will allow the reader to
easily apply our results, and to build upon them.  Our belief is that
this simple case will be a correct and important step towards a more
complete representation theory applicable to general \lcfts{}.

No introduction to these representation-theoretic aspects of \lcft{} could be complete without mentioning the seminal contributions of Rohsiepe.  These appeared thirteen years ago as a preprint \cite{RohRed96}, which to the best of our knowledge was never published, and a dissertation in German \cite{RohNic96}.  As far as we are aware, these are the only works which try to systematically
develop a representation theory for the Virasoro algebra, keeping in mind
applications to \lcft{} (specifically the so-called $LM(1,q)$ theories of
\cite{GabInd96}).  Indeed, it was Rohsiepe who first introduced the term
``staggered module'', though in a setting rather more general than we use
it here.  These references contain crucial insights on how to start
building the theory, and treat explicitly a particular subcase of the
formalism we construct.  We clearly owe a lot to the ideas and results contained therein.

On the other hand, Rohsiepe's formulation of the problem in
\cite{RohRed96} is somewhat different to our own, which in our opinion
has made applying his results a little bit inconvenient.  Moreover, an
unfortunate choice of wording in several of his statements,
as well as in the introduction and conclusions, can lead the casual reader
to conclude that the results have been proven in a generality significantly
exceeding the actuality.  Finally, the article seems to contain several
inaccuracies and logical gaps which we believe deserve correction and
filling (respectively).  We depart somewhat from the notation and terminology
of \cite{RohRed96} when we feel that it is important for clarity.

We have organised our article as follows. \secref{secNotation}
introduces the necessary basics --- the Virasoro algebra, some
generalities about its representations and most importantly the
result of Feigin and Fuchs describing the structure of \hwms{}. 
This section also serves to introduce the notation and conventions
that we shall employ throughout.  In \secref{secStag}, we then precisely
define our staggered modules and state the question which we are trying
to answer.  Here again, we fix notation and conventions.  The rest of
the section is devoted to observing some simple but important
consequences of our definitions.  In particular, we derive some
basic necessary conditions that must be satisfied by a staggered
module, and show how to determine when two staggered modules are
isomorphic.  This gives us a kind of uniqueness result.

\secref{sec: construction} then marks the beginning of our study of the
far more subtle question of existence.  Here, we prove an existence
result by explicitly constructing staggered modules, noting that we succeed precisely when a certain condition is satisfied.  This condition
is not yet in a particularly amenable form, but it does allow us to
deduce two useful results which answer the existence question for
certain staggered modules provided that the answer has been found
for certain other staggered modules.  These results are crucial to
the development that follows.  In particular, we conclude that if
a staggered module exists, then the module obtained by replacing
its right module by a Verma module (with the same highest weight) also exists.

We then digress briefly to set up and prove a technical result, the
Projection Lemma, which will be used later to reduce the enormous
number of staggered module possibilities to the consideration of a
finite number of cases.  This is the subject of \secref{sec: projections}. 
We then turn in \secref{sec: right Verma} to the existence question in the
case when the right module is a Verma module, knowing that this case is
the least restrictive.  Our goal is to reduce the not-so-amenable
condition for existence which we derived in \secref{sec: construction}
to a problem in linear algebra.  This is an admittedly lengthy exercise,
with four separate cases of varying difficulty to be considered (thanks
to the Projection Lemma).  The result is nevertheless a problem that we
can solve, and its solution yields a complete classification of
staggered modules whose right module is Verma.  This is completed in
\secref{sec: dimensions}.  We then consider in \secref{sec: invariants}
how to distinguish different staggered modules within the space of
isomorphism classes, when their left and right modules are the same.  
This is achieved by introducing invariants of the staggered module
structure and proving that they completely parametrise this space.

Having solved the case when the right module is Verma, we attack the
general case in \secref{sec: general case}.  We first characterise
when one can pass from Verma to general right modules in terms of
\svs{} of staggered modules.  This characterisation is then combined
with the Projection Lemma to deduce the classification of staggered
modules in all but a finite number of cases.  Unhappily, our methods
do not allow us to completely settle the outstanding cases, but we
outline what we expect in \secref{secCritRank} based on theoretical
arguments and studying an extensive collection of examples.  Finally,
we present our results in \secref{sec: conclusions} in a self-contained
summary.  Throughout, we attempt to illustrate the formalism that we
are developing with relevant examples, many of which have a physical
motivation and are based on explicit constructions in \lcft{} or
Schramm-Loewner evolution.

\section{Notation, Conventions and Background} \label{secNotation}

Our interest lies in the indecomposable modules of the Virasoro algebra, $\vir$.  These are modules which cannot be written as a direct sum of two (non-trivial) submodules, and therefore generalise the concept of irreducibility.  The Virasoro algebra is the infinite-dimensional (complex) Lie algebra spanned by modes $L_n$ ($n \in \ZZ$) and $C$, which satisfy
\begin{equation} \label{eqnDefVir}
\comm{L_m}{L_n} = \brac{m-n} L_{m+n} +  \delta_{m+n,0} \frac{m^3-m}{12} C \qquad \text{and} \qquad \comm{L_m}{C} = 0.
\end{equation}
The mode $C$ is clearly central, and in fact spans the centre of $\vir$.  We will assume from the outset that $C$ can be diagonalised on the modules we consider (this is certainly true of the modules which have been studied by physicists).  Its eigenvalue $c$ on an indecomposable module is then well-defined, and is called the central charge of that module.  We will always assume that the central charge is real.  Note that under the adjoint action, $\vir$ is itself an indecomposable $\vir$-module with central charge $c=0$.

In applications, the central charges of the relevant indecomposable modules usually all coincide.  It therefore makes sense to speak of \emph{the} central charge of a theory.  To compare different theories, it is convenient to parametrise the central charge, and a common parametrisation is the following:
\begin{equation} \label{eqnParC1}
c = 13 - 6 \brac{t + t^{-1}}.
\end{equation}
This is clearly symmetric under $t \leftrightarrow t^{-1}$.  For $c \leqslant 1$, we may take $t \geqslant 1$.  For $c \geqslant 25$, we may take $t \leqslant -1$.  When $1 < c < 25$, $t$ must be taken complex.  Many physical applications correspond to $t$ rational, so we may write $t = q / p$ with $\gcd \set{p,q} = 1$.  In this case, the above parametrisation becomes
\begin{equation} \label{eqnParC2}
c = 1 - \frac{6 \brac{p-q}^2}{pq}.
\end{equation}

The Virasoro algebra is moreover graded by the eigenvalue of $L_0$ under the
adjoint action.  Note however that this action on $L_n$ gives $-n L_n$ ---
the index and the grade are opposite one another.  This is a consequence of
the factor $\brac{m-n}$ on the right hand side of \eqnref{eqnDefVir}. 
Changing this to $\brac{n-m}$ by replacing $L_n$ by $-L_n$ would alleviate
this problem, and in fact this is often done in the mathematical literature.
However, we shall put up with this minor annoyance as it is this definition
which is used, almost universally, by the physics community.

The Virasoro algebra admits a triangular decomposition into subalgebras,
\begin{equation}
\vir = \vir^- \oplus \vir^0 \oplus \vir^+,
\end{equation}
in which $\vir^{\pm}$ is spanned by the modes $L_n$ with $n$ positive or
negative (as appropriate) and $\vir^0$ is spanned by $L_0$ and $C$.  We
note that $\vir^+$ is generated as a Lie subalgebra by the modes $L_1$
and $L_2$ (and similarly for $\vir^-$).  This follows recursively from the
fact that commuting $L_1$ with $L_n$ gives a non-zero multiple of
$L_{n+1}$, for $n \geqslant 2$.  The corresponding Borel subalgebras will
be denoted by $\vir^{\leqslant 0} = \vir^- \oplus \vir^0$ and $\vir^{\geqslant 0} = \vir^0 \oplus \vir^+$.  We mention that this triangular decomposition respects the standard anti-involution of the Virasoro algebra which is given by
\begin{equation} \label{eqnDefAdj}
L_n^{\dagger} = L_{-n} \qquad \text{and} \qquad C^{\dagger} = C,
\end{equation}
extended linearly to the whole algebra.  We shall often refer to this as the adjoint.\footnote{For applications to field theory, one would normally extend \emph{antilinearly}, hence the appellation ``adjoint''.  However, this distinction is largely irrelevant to the theory we are developing here.}

We will frequently find it more convenient to work within the \uealg{}
of the Virasoro algebra.  As we are assuming that $C$ always acts as
$c \, \wun$ on representations, we find it convenient to make this
identification from the outset.  In other words, we quotient the
\uealg{} of $\vir$ by the ideal generated by $C - c \, \wun$.  We
denote this quotient by $\uea$ and will henceforth abuse terminology
by referring to it as \emph{the} \uealg{} of $\vir$.  Similarly, the
\uealgs{} of $\vir^-$, $\vir^+$, $\vir^{\leqslant}$ and $\vir^{\geqslant}$
will be denoted by $\uea^-$, $\uea^+$, $\uea^{\leqslant}$ and
$\uea^{\geqslant}$, respectively.  The latter two are also to be understood
as quotients in which $C$ and $c \, \wun$ are identified.

The \uealg{} is a $\vir$-module under left-multiplication.  Moreover, it is also an $L_0$-graded $\vir$-module with central charge $0$ under the (induced) adjoint action, and it is convenient to have a notation for the homogeneous
subspaces.  We let $\uea_n$ denote the elements $U \in \uea$ for which
\begin{equation}
L_0 U - U L_0 = n U.
\end{equation}
Note that \eqnref{eqnDefVir} forces $L_n \in \uea_{-n}$.  We moreover remark
that the adjoint (\ref{eqnDefAdj}) extends to an adjoint on $\uea$ in the
obvious fashion:  $(L_{n_1} \cdots L_{n_k})^\dagger = L_{-n_k} \cdots L_{-n_1}$.

The most important fact about \uealgs{} is the \PBW{} Theorem which states,
for $\vir$, that the set
\begin{equation*}
\set{\cdots L_{-m}^{a_{-m}} \cdots L_{-1}^{a_{-1}} L_0^{a_0} L_1^{a_1} \cdots L_n^{a_n} \cdots \st a_i \in \NN \text{ with only finitely many } a_i \neq 0}
\end{equation*}
constitutes a basis of $\uea$.  Similar results are valid for $\uea^-$,
$\uea^+$, $\uea^{\leqslant}$ and $\uea^{\geqslant}$ (a proof valid for quite
general \uealgs{} may be found in \cite{MooLie95}).  Two
simple but useful consequences of this are that $\uea$ and its variants
have no zero-divisors and that
\begin{equation} \label{eqnDimVer}
\dim \uea^{\pm}_{\mp n} = \partnum{n},
\end{equation}
where $\tpartnum{n}$ denotes the number of partitions of $n \in \NN$.

As we have a triangular decomposition, we can define \hwss{} and Verma modules.  A \hws{} for $\vir$ is an eigenvector of $\vir^0$ which is annihilated by $\vir^+$.  To construct a Verma module, we begin with a vector $v$.  We make the space $\CC v$ into a $\vir^{\geqslant 0}$-module (hence a $\uea^{\geqslant 0}$-module) by requiring that $v$ is an eigenvector of $\vir^0$ which is annihilated by $\vir^+$ ($v$ is then a \hws{} for $\vir^{\geqslant 0}$).  Finally, the Verma module is then the $\vir$-module
\begin{equation*}
\CC v \dispotimes_{\uea^{\geqslant 0}} \uea,
\end{equation*}
in which the Virasoro action on the second factor is just by left
multiplication.  This is an example of the induced module construction.
Roughly speaking, it just amounts to letting $\vir^-$ act freely on the
\hws{} $v$.  In particular, we may identify this Verma module with
$\uea^- v$.

It follows that Verma modules are completely characterised by their
central charge $c$ and the eigenvalue $h$ of $L_0$ on their \hws{}.
 We will therefore denote a Verma module by $\VerMod_{h,c}$ (though
we will frequently omit the $c$-dependence when this is clear from
the context).  Its \hws{} will be similarly denoted by $v_{h,c}$
(so $\VerMod_{h,c} = \uea^- v_{h,c}$).  The \PBW{} Theorem for
$\uea^-$ then implies that $\VerMod_{h,c}$ has the following basis:
\begin{equation*}
\set{L_{-n_1} L_{-n_2} \cdots L_{-n_k} v_{h,c} \st k \geqslant 0
\text{ and } n_1 \geqslant n_2 \geqslant \cdots \geqslant n_k \geqslant 1}.
\end{equation*}
$L_0$ is thus diagonalisable on $\VerMod_{h,c}$, so $\VerMod_{h,c}$ may be
graded by the $L_0$-eigenvalues relative to that of the \hws{}.  These
eigenvalues are called the conformal dimensions of the corresponding
eigenstates.  The homogeneous subspaces 
$\brac{\VerMod_{h,c}}_n = \Kern \bigl( L_0-h-n \bigr)$
are finite-dimensional and in fact,
\begin{equation}
\dim \brac{\VerMod_{h,c}}_n = \partnum{n},
\end{equation}
by \eqnref{eqnDimVer}.  Finally, each $\VerMod_{h,c}$ admits a unique
symmetric bilinear form $\tinner{\cdot}{\cdot}_{\VerMod_{h,c}}$, contravariant with respect to the adjoint (\ref{eqnDefAdj}), $\tinner{u'}{U u} = \tinner{U^\dagger u'}{u}$,
and normalised by $\tinner{v_{h,c}}{v_{h,c}} = 1$ (we will usually neglect to specify the module with a subscript index when this causes no confusion).  This is referred to as the Shapovalov form of $\VerMod_{h,c}$.\footnote{In applications to field theory, where the adjoint (\ref{eqnDefAdj}) is extended antilinearly to all of $\vir$, this would define a hermitian form.  Physicists often refer to this form as the Shapovalov form as well.}  We will also refer to it as the scalar product.  Note that distinct homogeneous subspaces are orthogonal with respect to this form.

A useful alternative construction of the Verma module $\VerMod_{h,c}$ is to
instead regard it as the quotient of $\uea$ (regarded now as a $\vir$-module
under left-multiplication) by the left-ideal (left-submodule) $\mathcal{I}$
generated by $L_0 - h \; \wun$, $L_1$ and $L_2$ (recall that $L_1$ and $L_2$
generate $\vir^+$ hence $\uea^+$).  It is easy to check that the
equivalence class of the unit $\sqbrac{\wun}$ is a \hws{} of
$\VerMod_{h,c}$ with the correct conformal dimension and
central charge.  We will frequently use the consequence
that any element $U \in \uea$ which annihilates the
\hws{} of $\VerMod_{h,c}$ must belong to $\mathcal{I}$:
If $U v_{h,c}=0$, then
\begin{equation} \label{eq: uea null in Verma}
U = U_0 \brac{L_0 - h \; \wun} + U_1 L_1 + U_2 L_2 \qquad \text{for some $U_0 , U_1 , U_2 \in \uea$.}
\end{equation}

As Verma modules are cyclic (generated by acting upon a single vector),
they are necessarily indecomposable.  However, they need not be irreducible.
If the Verma module $\VerMod_{h,c}$ is reducible then
it can be shown that there exists another $L_0$-eigenvector, not proportional
to $v_{h,c}$, which is annihilated by $\vir^+$.
Such vectors are known as \emph{\svs{}}.  If
there is a singular vector $w \in \VerMod_{h,c}$ at grade $n$, then it
generates a submodule isomorphic to $\VerMod_{h+n,c}$.  Conversely, every
submodule of a Verma module is generated by \svs{}.  Any quotient of a
Verma module by a proper submodule is said to be a \hwm{}.  It follows
that such a quotient also has a cyclic \hws{} (in fact, this is the usual
definition of a \hwm{}) with the same conformal dimension and central
charge as that of the Verma module.  Moreover, it inherits the obvious
$L_0$-grading.  Finally, factoring out the maximal proper submodule gives
an irreducible \hwm{}, which we will denote by $\IrrMod_{h,c}$ (or
$\IrrMod_h$ when $c$ is contextually clear).

We pause here to mention that in the physics literature, the term ``singular
vector'' is often used to emphasise that the \hws{} in question is not the
one from which the entire \hwm{} is generated (that is, it is not the cyclic
\hws{}).  This is rather inconvenient from a mathematical point of view, but
is natural because of the following calculation:  If
$w \in \brac{\VerMod_{h,c}}_n$ is a \sv{} (with $n>0$), then for all
$w' = U v_{h,c} \in \brac{\VerMod_{h,c}}_n$ (so $U \in \uea^-_n$), hence for
all $w' \in \VerMod_{h,c}$,
\begin{equation}
\inner{w}{w'} = \inner{w}{U v_{h,c}} = \inner{U^{\dagger} w}{v_{h,c}} = 0,
\end{equation}
as $U^{\dagger} \in \uea^+_{-n}$ with $n>0$.  We will however follow the
definition used in mathematics in which a \sv{} is precisely a \hws{},
qualifying those which are not generating as \emph{proper}.  We will also
frequently express a singular vector in the form $w = X v_{h,c}$,
$X \in \uea^-$, in which case we will also refer to $X$ as being
singular.\footnote{Admittedly, $X \in \uea^-_{m}$ being singular
only makes sense when $h$ (and $c$) is specified.  What this
concretely means is that for $n=1,2$ there are
$X^{(n)}_0, X^{(n)}_1, X^{(n)}_2 \in \uea$
such that $L_n X = X^{(n)}_0 (L_0-h) + X^{(n)}_1 L_1 + X^{(n)}_2 L_2$
(compare with \eqnref{eq: uea null in Verma}).  The value of $h$ should
nevertheless always be clear from the context, so we trust that this 
terminology will not lead to any confusion.}

Let us further define a descendant of a \sv{} $w$ to be an element of
$\uea^- w$.  
The above calculation then states that proper \svs{} and their descendants
have vanishing scalar product with all of $\VerMod_{h,c}$, including 
themselves.\footnote{This is in fact the origin of the term ``singular'' 
in this context --- it refers to the fact that the matrix representing the
Shapovalov form at grade $n$ has determinant zero.}  It now follows that 
the maximal proper submodule of $\VerMod_{h,c}$ is precisely the subspace
of vectors which are orthogonal to $\VerMod_{h,c}$.  The Shapovalov form
$\tinner{\cdot}{\cdot}_{\VerMod_{h,c}}$ therefore descends to a well-defined
symmetric bilinear form $\tinner{\cdot}{\cdot}_{\mathcal{K}}$ on any \hwm{}
$\mathcal{K}$ (also called the Shapovalov form).  It is non-degenerate if
and only if $\mathcal{K}$ is irreducible.

Through a cleverly arranged computation \cite{AstStr97}, it is not hard to
show the following facts:  In a Virasoro Verma module, there can only exist
one \sv{} $w = X v_{h,c}$, up to constant multipliers, at any given grade
$n$ (that is, with $X \in \uea^-_n$).  Moreover, the coefficient of
$L_{-1}^n$ when $X$ is written in the \PBW-basis is never zero.  If this
coefficient is unity, we will say that $X$ is \emph{normalised}, and by
association, that $w$ is also normalised.  This particular
normalisation is convenient because it does not depend on whether we choose
to represent $X$ as a sum of monomials ordered in our standard \PBW{} manner
or with respect to some other ordering.  
We note explicitly that $v_{h,c}$ is a normalised \sv{}.  This normalisation
also extends readily to cover general \hwms{}:  A (non-zero) \sv{} of such
a module will be said to be normalised if it is the projection of a
normalised \sv{} of the corresponding Verma module.

A far more difficult, but nevertheless fundamental, result in Virasoro algebra
representation theory concerns the explicit evaluation of the determinant of
the Shapovalov form, restricted to $\brac{\VerMod_{h,c}}_n$.  The vanishing
of this determinant indicates the existence of proper \svs{} (and their
descendants), so understanding the submodule structure of \hwms{}
reduces, to a large extent, to finding the zeroes of the
\emph{Kac determinant formula},
\begin{equation} \label{eqnKac}
\det \ \bigl. \inner{\cdot}{\cdot} \bigr\rvert_{\brac{\VerMod_{h,c}}_n} = \alpha_n \prod_{\substack{r,s \in \ZZ_+ \\ rs \leqslant n}} \brac{h - h_{r,s}}^{\partnum{n-rs}}.
\end{equation}
Here, $\alpha_n$ is a non-zero constant independent of $h$ and $c$, and the $h_{r,s}$ vary with $c$ according to
\begin{equation} \label{eqnKacDims}
h_{r,s} = \frac{r^2 - 1}{4} t - \frac{rs - 1}{2} + \frac{s^2 - 1}{4} t^{-1} = \frac{\brac{ps - qr}^2 - \brac{p-q}^2}{4pq},
\end{equation}
when $c$ is parametrised as in \eqnDref{eqnParC1}{eqnParC2} (respectively).
This determinant vanishes when $h = h_{r,s}$ for some $r,s \in \ZZ_+$ with
$rs \leqslant n$.  Given such an $h = h_{r,s}$ then, it
can be shown that there exists a (proper) singular vector at grade $rs$.

The Kac determinant formula was conjectured by Kac in \cite{KacCon79} and
proven by Feigin and Fuchs in \cite{FeiSke82}.  Reasonably accessible
treatments may be found in \cite{KacBom88,ItzSta89}. 
Feigin and Fuchs then used this formula to find all the homomorphisms 
between Verma modules, effectively determining the \sv{} structure of any
Verma module \cite{FeiVer84}.  It turns out to be convenient to distinguish
four different types of structures which we illustrate in \figref{figVerma}.
We will refer to these as ``point'', ``link'', ``chain'' or ``braid'' type
Verma modules (hopefully this notation is self-explanatory).  These
correspond to the cases $\mathrm{I}$, $\mathrm{II}_0$ and $\mathrm{II_-}$
(point), $\mathrm{II_+}$ (link), $\mathrm{III_{\pm}^0}$ and
$\mathrm{III_{\pm}^{00}}$ (chain), and $\mathrm{III_{\pm}}$ (braid), in the
notation of Feigin and Fuchs.  We will also say that more general \hwms{}
are of the above types, defined through inheriting their type from the
corresponding Verma module.

{
\psfrag{c<1}[][]{$t>0$}
\psfrag{c>25}[][]{$t<0$}
\psfrag{P}[][]{Point}
\psfrag{L}[][]{Link}
\psfrag{C}[][]{Chain}
\psfrag{B}[][]{Braid}
\begin{figure}
\begin{center}
\includegraphics[width=10cm]{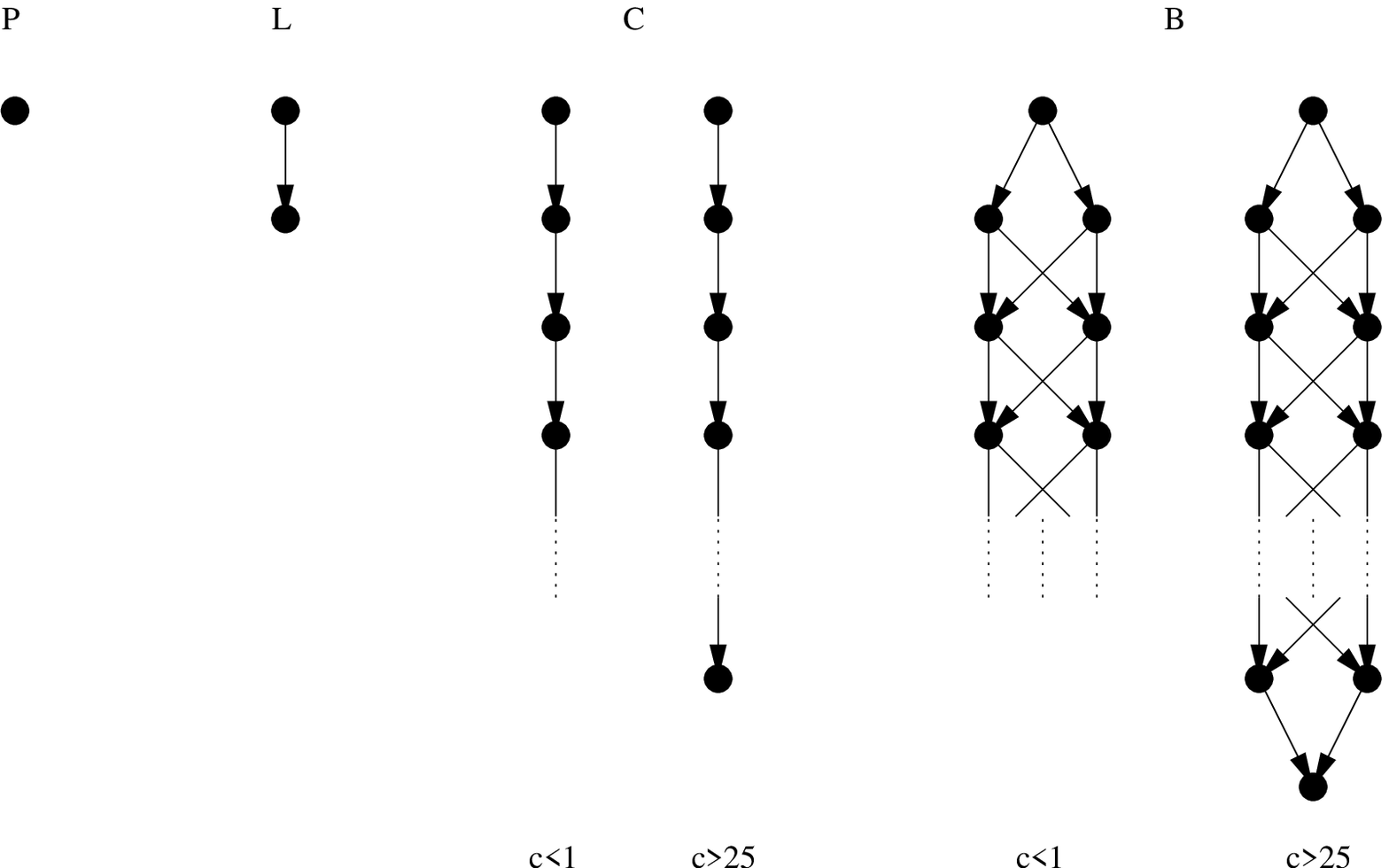}
\caption{The \sv{} structure, marked by black circles, of Virasoro
Verma modules.  Arrows from one vector to another indicate that the
latter is a descendant of the former and not vice-versa.  Point and
link-type Verma modules occur for all central charges.  Chain and
braid-type modules occur only when $t$ is rational and non-zero.
Note that $t > 0$ corresponds to $c \leqslant 1$ and $t < 0$
corresponds to $c \geqslant 25$.} \label{figVerma}
\end{center}
\end{figure}
}

We take this opportunity to describe when each of these cases occurs
(see \cite{FeiVer84} for further details) and to
introduce some useful notation for each.  Recall that each $h_{r,s}$ depends
on $t$ and that $t$ parametrises the central charge via \eqnref{eqnParC1}.
\begin{description}
\item[Point] If $t$ and $h$ are such that $h \neq h_{r,s}$ for every
$r,s \in \bZ_+$, then $\Verma_h$ is irreducible and there are no highest
weight vectors besides the multiples of the cyclic highest weight vector $v_h$.
\item[Link] Suppose that $t \notin \bQ$ (recall that $t$ may be complex) and that there exist $r,s \in \bZ_+$ (unique since $t$ is not rational) such that
$h = h_{r,s}$.  Then $\Verma_h$ possesses a \sv{} at grade $rs$ which generates the maximal proper submodule of $\Verma_h$.  This maximal proper submodule, itself isomorphic to $\Verma_{h+rs}$, is then of point type, so there are no other non-trivial \svs{}.  We denote the normalised \sv{} at grade $rs$ by
$X_1 v_h$ ($X_1 \in \uea^-_{rs}$ is therefore also normalised) and for
compatibility with the chain case, we will denote the grade of this \sv{}
by $\ell_1 = rs$.
\item[Chain] Suppose that $t = q/p$ with $p \in \ZZ_+$ and
$q \in \ZZ \setminus \set{0}$ relatively prime, and that $h = h_{r,s}$
for some $r,s \in \bZ_+$ with $p \mid r$ or $q \mid s$.  Then, choosing
$r$ and $s$ such that $h = h_{r,s}$ and $rs>0$ is minimal, $\Verma_h$ has a
\sv{} at grade $rs$ which generates the maximal proper submodule, itself
isomorphic to $\Verma_{h+rs}$.  In contrast to the link case, this maximal
proper submodule is also of chain type, except in the degenerate case where $t<0$, $r \leqslant p$ and $s \leqslant \abs{q}$, in which case it is of point type.  Thus, we iteratively find a sequence of \svs{} as in \figref{figVerma}.  This sequence is infinite if $t$ is positive and finite if $t$ is negative (terminating with a degenerate case).  We write the normalised \svs{} of $\Verma_h$ as $v_h = X_0 v_h, X_1 v_h, X_2 v_h, \ldots$, and denote their respective grades by $0 = \ell_0 < \ell_1 < \ell_2 < \cdots$ (so $X_k \in \uea^-_{\ell_k}$).
\item[Braid] Suppose that $t = q/p$ with $p \in \ZZ_+$ and
$q \in \ZZ \setminus \set{0}$ relatively prime, and that $h = h_{r,s}$ for
some $r,s \in \bZ_+$ with $p \nmid r$ and $q \nmid s$. Choose $r$,
$s$, $r'$ and $s'$ such that $h = h_{r,s} = h_{r',s'}$, $rs>0$ is minimal and
$r's'>rs$ is minimal but for $rs$ (such $r',s'$ always exist except in certain
degenerate cases which we will describe below). Then $\Verma_h$ has two \svs{}, $X^-_1 v_h$ and $X^+_1 v_h$ at grades $h + rs$
and $h + r's'$ respectively. Together they generate the maximal proper
submodule (not a highest weight module in this case).  The Verma modules
generated by these two \svs{} (separately) are again of braid type
(except in the degenerate cases),
and their intersection is the maximal proper submodule of either.
One therefore finds a double sequence of \svs{} in this case, as
illustrated in \figref{figVerma}.  As in the chain case, these sequences are infinite if $t$ is positive and finite if $t$ is negative.

The degenerate cases referred to above occur when $t<0$, $r < p$ and $s < \abs{q}$.  Then, there are no labels $r',s'$ to be found, the maximal proper submodule is generated by a single \sv{}, and is in fact of point type.
In the non-degenerate cases, we write the normalised \svs{} of $\Verma_h$ as
$v_h = X_0^+ v_h, X_1^- v_h, X_1^+ v_h, X_2^- v_h, X_2^+ v_h, \ldots$, denoting their respective grades by $0 = \ell_0^+ < \ell_1^- < \ell_1^+ < \ell_2^- < \ell_2^+ < \cdots$ (so $X_k^{\pm} \in \uea^-_{\ell_k^{\pm}}$).  
When $t < 0$, the double sequence of \svs{} terminates because of the above degenerate cases, so for some $k$, there is no $X^+_k$ and the singular
vector $X_k^- v_{h}$ generates an irreducible Verma module.
\end{description}
Note that when it comes to the submodule structure, the link case
is identical to the degenerate cases of both the chain and
braid cases.  However, we emphasise that chain and braid type modules only
exist when $t$ is rational.  With this proviso in mind, we can (and often
will) treat the link case as a subcase of the chain case.

Suppose that for a (normalised) \sv{} $w = X v_{h,c}$, we can factor $X \in \uea^-$ non-trivially as $X' X''$ where $X'' v_{h,c}$ is again (normalised and) singular.  We will then say that $w$ (and $X$) is composite.  Otherwise, $w$ (and $X$) is said to be prime.  A composite \sv{} is then just one which is a proper descendant of another (proper) \sv{}.  We can generalise this by further factoring $X$ as $X^{\brac{1}} X^{\brac{2}} \cdots X^{\brac{\rnkO}}$, where $X^{\brac{i}} X^{\brac{i+1}} \cdots X^{\brac{\rnkO}} v_{h,c}$ is (normalised and) singular for all $i$.  Such factorisations will not be unique, but when they cannot be further refined, we will say that each $X^{\brac{i}}$ is prime.  Such prime factorisations need not be unique either when the Verma module is of braid type, but it is easy to check from the above classification that for these factorisations the number of factors $\rnkO$ is constant.  We will refer to $\rnkO$ as the \emph{rank} of the \sv{} $w = X v_{h,c}$.  Rank-$1$ \svs{} are therefore prime, and we may regard the cyclic \hws{} as the (unique) rank-$0$ \sv{}.  In our depiction of Verma modules (\figref{figVerma}), the \sv{} rank corresponds to the vertical axis (pointing down).

\section{Staggered Modules} \label{secStag}

The central objects of our study are the so-called \emph{staggered modules}
of Rohsiepe \cite{RohRed96}.  The simplest non-trivial case, which is all
that will concern us, is the following:  A staggered module $\StagMod$ is
an indecomposable $\vir$-module for which we have a short exact sequence
\begin{equation} \label{eqnDefStag}
\dses{\LMod}{\iota}{\StagMod}{\pi}{\RMod},
\end{equation}
in which it is understood that $\LMod$ and $\RMod$ are \hwms{}, $\iota$
and $\pi$ are module homomorphisms, and $L_0$ is \emph{not} diagonalisable
on $\StagMod$, possessing instead Jordan cells of rank at most $2$.  When
we refer to a module as being staggered, we have these restrictions in mind.
In particular, our staggered modules are extensions of one \hwm{} by another.
As we are assuming that indecomposable modules such as $\StagMod$ have a well-defined central charge, those of $\LMod$ and $\RMod$ must coincide.  More generally, one could consider indecomposable modules constructed from more than two \hwms{}, and with higher-rank Jordan cells for $L_0$, but we shall not do
so here.

We call $\LMod$ and $\RMod$ the \emph{left} and \emph{right} modules
(of $\StagMod$), and denote their \hwss{} by $\Lhws$ and $\Rhws$, with
(real) conformal dimensions $\LDim$ and $\RDim$, respectively.  $\LMod$
is then a submodule of $\StagMod$ (we will frequently forget to distinguish
between $\LMod$ and $\func{\iota}{\LMod}$), whereas $\RMod$ is not (in
general).  We remark that Rohsiepe uses similar nomenclature in this case,
defining ``lower'' and ``upper modules'' such that the latter is the
quotient of the staggered module by the former.  However, we stress
that these do not in general coincide with our left and right modules.
In particular, Rohsiepe defines his lower module to be
the submodule of all $L_0$-eigenvectors, which need not be a \hwm{}
(a concrete illustration of this will be given in
\exref{ex:-2,11,13} and the remark following it --- the general phenomenon
will be discussed after \propref{prop: singular vectors in staggered modules}).

Our question is the following:
\begin{quote}
\textbf{Given two highest weight modules $\Hlft$ and $\Hrgt$, can we
classify the (isomorphism classes of) staggered modules $\sS$
corresponding to the short exact sequence (\ref{eqnDefStag})?}
\end{quote}
Abstractly, if we dropped the requirement that $L_0$ has non-trivial Jordan
cells, then we would be asking for a computation of
$\Ext{1}{\uea}{\RMod}{\LMod}$ in an appropriate category \cite{CarHom56},
a difficult task.  As we shall see however, requiring non-diagonalisability
leads to a reasonably tractable problem for which we do not need the
abstract machinery of homological algebra.

An answer to our question will be given in the following sections.
For convenience, we summarise our results in \secref{sec: conclusions}
(\thmref{thmTheAnswer}).  This section is largely self-contained, and
so may be read independently of most of what follows.  However, we
suggest that an appreciation of the r\^{o}le of the beta-invariants
(\secDref{secStag}{sec: invariants}) represents a minimal prerequisite
for this result.

As staggered modules necessarily have vectors which are not
$L_0$-eigenvectors, we cannot grade the module by the eigenvalue of $L_0$ 
relative to that of some  reference vector.  However, $L_0$ can still be
put in Jordan normal form, so we may decompose it into commuting
diagonalisable and nilpotent operators:  $L_0 = \LOd + \LOn$.  A staggered
module may then be consistently graded by the eigenvalues of its vectors
under $\LOd$, relative to the minimal eigenvalue of $\LOd$.  We will refer
to $\LOd$-eigenvalues as conformal dimensions, even when the corresponding
eigenvector is not an $L_0$-eigenvector.
Note that the maps $L_m$ are still consistent with
this more general grading --- one easily checks that $L_m \in \uea_{-m}$ maps
the $\LOd$ eigenspace of eigenvalue $h$ to that of eigenvalue
$h-m$.

A submodule of a (graded) Virasoro module can be assigned a grading in at least two distinct ways.  First, it can inherit the grading from its parent, so that homogeneous states have the same grade in both modules.  The inclusion map is then a graded homomorphism.  Second, a grading may be defined as the conformal dimension of the states relative to the minimal conformal dimension of the submodule.  Both have their uses, but unless otherwise specified, we will always assume that a submodule inherits its grading from its parent.

We introduce some more notation.  Let $x = \func{\iota}{\Lhws}$ denote the
\hws{} of the submodule $\func{\iota}{\LMod} \subset \StagMod$ and choose 
an $\LOd$-eigenvector $y$ in the preimage
$\func{\pi^{-1}}{\Rhws} \subset \StagMod$. The vector
$x$ is then an eigenvector of $L_0$ whilst $y$ is not (if it were, its
descendants would also be, hence $L_0$ would be diagonalisable on
$\StagMod$).  Their conformal dimensions are $\LDim$ and $\RDim$
respectively.  We now define the auxiliary vectors
\begin{equation} \label{eqnDefOmega}
\omega_0 = \brac{L_0 - \RDim} y, \qquad \omega_1 = L_1 y \qquad \text{and}
\qquad \omega_2 = L_2 y.
\end{equation}
Since $L_1$ and $L_2$ generate $\uea^+$, $\omega_1$ and $\omega_2$ determine
the action of $\uea^+$ on $y$.

\begin{proposition} \label{propSVw0}
$\omega_0, \omega_1, \omega_2 \in \LMod$ and $\omega_0$ is a non-zero
\sv{} of $\LMod \subset \StagMod$.
\end{proposition}
\begin{proof}
Since $L_0 - \RDim$, $L_1$ and $L_2$ annihilate
$\Rhws = \func{\pi}{y} \in \RMod = \StagMod / \LMod$, their action on $y$
must yield elements of $\LMod$.  If $\omega_0$ vanished then $y$ would be
an eigenvector of $L_0$, hence $\omega_0 \neq 0$. Moreover,
\begin{equation}
L_n \omega_0 = L_n \brac{L_0 - \RDim} y = \brac{L_0 - \RDim + n} L_n y,
\end{equation}
hence $L_n \omega_0 = 0$ for all $n>0$, as $y$ has $\LOd$-eigenvalue
$h^\rgt$, so $L_n y \in \Hlft$ has $L_0$-eigenvalue $h^\rgt-n$.
\end{proof}

Define $\ell = \RDim - \LDim$.  It follows that $\ell$ is then the grade of
the singular vector $\omega_0$ and its Jordan partner $y$ in the staggered
module $\StagMod$.  The grades of $\omega_1$ and $\omega_2$ are therefore
$\ell - 1$ and $\ell - 2$, respectively.  One immediate
consequence is that $\ell$ is a non-negative integer.
Exact sequences (\ref{eqnDefStag}) with $\ell < 0$ certainly exist, but cannot
describe staggered modules.\footnote{Apart from the obvious direct sums
$\sH^\lft \oplus \sH^\rgt$,
reducible Verma modules form a simple class of examples of this type.}
When $\ell = 0$, we must have $\omega_0 = x$ up
to a non-zero multiplicative constant.  When $\ell > 0$, $\LMod$ has a
proper \sv{}, hence the Kac determinant formula (\ref{eqnKac}) has a zero.
We thereby obtain our first necessary conditions on the existence of
staggered modules.

\begin{corollary}
A staggered module cannot exist unless $\ell \in \NN$.  Moreover, if
$\ell > 0$, then $\LDim = h_{r,s}$ for some $r,s \in \ZZ_+$ (where
$h_{r,s}$ is given in \eqnref{eqnKacDims}).
\end{corollary}

We will assume from here on that $\omega_0 = X x$, where $X \in \uea^-_{\ell}$
is normalised (and singular).
Since $y$ is related to the normalised \sv{} $\omega_0$ by
\eqnref{eqnDefOmega}, this also serves to normalise $y$
(equivalently, we rescale $\pi$).  However, there is
still some residual freedom in the choice of $y$.  Indeed, $y$ was only
chosen to be an $\LOd$-eigenvector in $\func{\pi^{-1}}{\Rhws}$, so we
are still free to make the redefinitions
\begin{equation} \label{eqnGT}
y \longrightarrow y + u \qquad \text{for any $u \in \LMod_{\ell}$,}
\end{equation}
without affecting the defining property (or normalisation) of $y$.
Following \cite{RidLog07}, we shall refer to such redefinitions as gauge
transformations.  These transformations obviously do not change the abstract
structure of the staggered module (for a more formal statement see
\propref{prop: equivalence}).

It is natural then to enquire about gauge-invariant quantities as one expects
that it is these, and only these, which characterise the staggered module.
When $\ell > 0$, a simple but important example is given by \cite{RidPer07}
\begin{equation} \label{eqnDefBeta}
\beta = \inner{x}{X^{\dagger} y}, \qquad \text{(recall $\omega_0 = X x$).}
\end{equation}
This $\beta$ is obviously gauge-invariant, as
$\tinner{x}{X^{\dagger} u} = \tinner{\omega_0}{u} = 0$ for all
$u \in \LMod_{\ell}$.  In the physics literature, this has been called the
logarithmic coupling for field-theoretic reasons.\footnote{We remark that
when $\ell > 0$, one can extend the definition of the Shapovalov
form to $\LMod \times \StagMod$ by noting that for $u = U x \in \LMod$,
\begin{equation*}
\tinner{\vphantom{U^{\dagger}} u}{y} = \tinner{\vphantom{U^{\dagger}} U x}{y} = \tinner{x}{U^{\dagger} y} \qquad \text{and} \qquad U^{\dagger} y \in \LMod.
\end{equation*}
With this extension, we can write $\beta = \tinner{\omega_0}{y}$.  One
can also define an extended scalar product when $\ell = 0$, but in this 
case $\tinner{x}{x}$ necessarily vanishes:
\begin{equation*}
\tinner{x}{x} = \tinner{x}{\brac{L_0 - \RDim} y} = \tinner{\brac{L_0 - \LDim} x}{y} = 0 \qquad \text{($\ell = 0$).}
\end{equation*}
We must instead take $\tinner{x}{y} = 1$.  These extensions are important in
applications to \lcft{} in which they give specialisations of so-called
two-point correlation functions \cite{FloBit03,GabAlg03}.  However, we
will have no need of them here.  We only mention that the
non-diagonalisability of $L_0$ on $\StagMod$ is not in conflict with its
self-adjointness because such extensions of the Shapovalov form are
necessarily indefinite \cite{BogInd74}.}  Here, we shall just refer to it
as the \emph{beta-invariant}.  Note that since $\tinner{x}{x} = 1$ and
$\dim \LMod_0 = 1$,
\begin{equation} \label{eqnBeta}
X^{\dagger} y = \beta x \qquad \text{($\ell > 0$).}
\end{equation}
We further note that the numerical value of this invariant depends upon the
chosen normalisations of $\omega_0$ and $y$ (which is why we have specified
these normalisations explicitly).\footnote{A historical comment is in
order here.  The notation $\beta$ for a quantity distinguishing staggered
modules dates back to \cite{GabInd96}. There however, $\beta$ was defined
by a particular ``gauge-fixing''.  In our language, their proposal was that
one chooses (gauge-fixes) $y$ such that $L_n y = 0$ for all $n>1$, and then
defines $\beta$ by $L_1^{\ell} y = \beta x$.  Comparing with \eqnref{eqnBeta}
and our normalisation of $X$, we see that this choice of gauge will
reproduce the values of our definition.  But, there remains the question
of whether it is always possible to perform such a convenient gauge-fixing.
For the modules considered in \cite{GabInd96} this was the case, but unfortunately it is not possible in general.
Counterexamples are easy to construct and we
offer the staggered modules with $c=0$ ($t=\tfrac{3}{2}$),
$\LMod = \VerMod_0 / \VerMod_1$ and $\RMod = \VerMod_2 / \VerMod_h$ with
$h = 5$ and $7$ as the simplest such examples.}
It is worth pointing out that if $X$ were composite, $X=X^{(1)} X^{(2)}$
with both $X^{(j)}$ non-trivial, then
\begin{equation} \label{eqnBeta=0}
\beta = \Bigl\langle x , \bigl( X^{(2)} \bigr)^{\dagger} \bigl( X^{(1)} \bigr)^{\dagger} y \Bigr\rangle = \Bigl\langle X^{(2)} x , \bigl( X^{(1)} \bigr)^{\dagger} y \Bigr\rangle = 0,
\end{equation}
because $X^{(2)} x \in \LMod$ is singular and $\bigl( X^{(1)} \bigr)^{\dagger} y \in \LMod$.  The beta-invariant is therefore always trivial in such cases.  Non-trivial invariants can still be defined when $X$ is composite, though their properties necessarily require a little more background.  We will defer a formal discussion of such invariants until \secref{sec: invariants}.

Consider now the right module $\RMod = \VerMod_{\RDim} / \mathcal{J}$.
If $\mathcal{J}$ is non-trivial, then it will be generated as a submodule
of $\VerMod_{\RDim}$ by one or two \svs{} of the same rank (\figref{figVerma}).
When one
generator suffices, we denote it by $\overline{X} v_{\RDim}$; when two
generators are required, they will be denoted by $\overline{X}^- v_{\RDim}$
and $\overline{X}^+ v_{\RDim}$.  As usual, we take all of these to be normalised.  The corresponding grades are
$\overline{\ell}$ or $\overline{\ell}^- < \overline{\ell}^+$, respectively.
However, unless we are explicitly discussing the case of two independent
generators, we shall suppress the superscript indices for clarity.

We have introduced $\omega_0$, $\omega_1$ and $\omega_2$ to specify the action
of $\uea^{\geqslant 0}$ on $y$.  When $\mathcal{J}$ is non-trivial, the action
of $\uea^-$ on $y$ will not be free.  Instead, we have
$\overline{X} \Rhws = 0$ in $\RMod$, hence
\begin{equation} \label{eqnDefOmegaBar}
\overline{X} y = \varpi \qquad \text{(in $\StagMod$)}
\end{equation}
defines a 
vector $\varpi \in \LMod$ (two vectors
$\varpi^{\pm}$ when $\mathcal{J}$ is generated by two \svs{}).  The grade
of $\varpi$ is then $\ell + \overline{\ell}$.
Recalling that $\Stagg$ as a vector space is just the direct sum of $\sH^\lft$
and $\sH^\rgt$, and considering a vector space basis of $\Verma_{h^\rgt}$ that
extends a basis for the submodule $\mathcal{J}$, it is
easy to see that the
Virasoro module structure of $\StagMod$ is completely determined by
$\omega_0$, $\omega_1$, $\omega_2$ and $\varpi$.

The existence of $\varpi$ also leads to the following important structural observation.

\begin{proposition} \label{propVanSVs}
When $\sH^\rgt$ is not Verma, so $\overline{X}$ is defined,
we have $\overline{X} \omega_0 = 0$.
\end{proposition}
\begin{proof}
Since $\overline{X} \in \uea^-_{\overline{\ell}}$,
\begin{equation}
\overline{X} \omega_0 = \overline{X} \brac{L_0 - \RDim} y = \brac{L_0 - \RDim - \overline{\ell}} \overline{X} y = \brac{L_0 - \RDim - \overline{\ell}} \varpi = 0,
\end{equation}
as $\varpi$ is an $L_0$-eigenvector of dimension $\RDim + \overline{\ell}$.
\end{proof}

\noindent We remark that the vanishing of $\overline{X} \omega_0$ implies
that there are no non-zero singular vectors in
$\sH^\lft_{\ell+\overline{\ell}}$.  Indeed, the normalised \sv{} of this
grade is $\overline{X} X x$ (which is composite if $\ell>0$).  Thus we
may interpret \propref{propVanSVs} as saying that if a \sv{} of
$\VerMod_{\RDim}$ is set to
zero in $\RMod$, then the \sv{} of $\VerMod_{\LDim}$ of the same conformal
dimension must also be set to zero in $\LMod$.  Otherwise, the module
$\StagMod$ cannot be staggered.  Contrapositively, if $\LMod$ has a
non-trivial \sv{} (of rank greater than that of $\omega_0$),
then $\RMod$ must have a
non-trivial \sv{} of the same conformal dimension (more formally, there
is a module homomorphism $\RMod \rightarrow \LMod$ which maps
$\Rhws \mapsto \omega_0$).  In particular, if $\LMod$ is a Verma module,
then $\RMod$ must likewise be Verma.

It turns out that there is some redundancy inherent in describing a staggered module in terms of the vectors $\omega_0$, $\omega_1$, $\omega_2$ and $\varpi$.

\begin{proposition} \label{propDetBy}
The vector $\varpi$ is determined by the knowledge of
$\LMod$, $\RMod$, $\omega_1$ and $\omega_2$.
\end{proposition}
\begin{proof}
We consider the action of $L_n$ on $\varpi = \overline{X} y$ for $n>0$,
recalling that $\overline{X} \in \uea^-_{\overline{\ell}}$.  First note
that $L_n \overline{X} \in \uea$ annihilates $v_{\RDim} \in \VerMod_{\RDim}$,
since $\overline{X} v_{\RDim}$ is singular.  Hence, we may write
\begin{equation}
L_n \overline{X} = U_0 \brac{L_0 - \RDim} + U_1 L_1 + U_2 L_2,
\end{equation}
for some $U_0, U_1, U_2 \in \uea$ (depending on $n$).  Such $U_j$ can
clearly be computed, for example by \PBW-ordering
$L_n \overline{X}$ and in each resulting term, rewriting the rightmost
$L_m$ (if $m > 2$) in terms of $L_1$ and $L_2$.  It follows that
\begin{equation} \label{eq: affine relation}
L_n \varpi = U_0 \omega_0 + U_1 \omega_1 + U_2 \omega_2,
\end{equation}
so it remains to demonstrate that knowing $L_n \varpi$ for all $n>0$ is
equivalent to knowing $\varpi \in \LMod_{\ell + \overline{\ell}}$.  But,
the intersection of the kernels of the $L_n$ with $n>0$ on
$\LMod_{\ell + \overline{\ell}}$ is just the set of \svs{} of this subspace. 
The only candidate for such a \sv{} is $\overline{X} \omega_0$, and this
vanishes by \propref{propVanSVs}.
\end{proof}

\noindent We recall that $\omega_0$ is already determined
by $\LMod$ and $\RMod$, which is why it was not referred to explicitly in
\propref{propDetBy}.  We will therefore refer to the pair
\begin{equation*}
\brac{\omega_1 , \omega_2} \in \LMod_{\ell - 1} \oplus \LMod_{\ell - 2}
\end{equation*}
as the \emph{data} of a given staggered module.  That is not to say that
$\omega_0$ and the $\varpi$ will not play an important r\^{o}le in what
follows.  Rather, it just notes that $\omega_1$ and $\omega_2$ are
sufficient to describe $\StagMod$ completely.  One simple consequence
arises when $\ell = 0$, for then there is only one possible choice of data,
$\omega_1=\omega_2=0$:

\begin{corollary} \label{cor: uniqueness l 0}
If $\ell = 0$, there exists at most one staggered module (up to isomorphism) for any given choice of left and right modules.
\end{corollary}

\begin{example} \label{ex: SLE l 0}
In \cite{KytFro08}, staggered modules with $\ell = 0$ were identified
in the context of the Schramm-Loewner evolution curve with
parameters\footnote{For these parameters, we follow here and in later examples the established notation of the Schramm-Loewner evolution literature,
where the curve and its growth process are often denoted
simply by SLE${}_\kappa(\rho)$. Roughly speaking, $\kappa$ determines
the universality class (the central charge and fractal dimension
of the curve), whereas $\rho$ is related to the choice of boundary conditions.
We are also using $\rho$ to denote the rank of a \sv{} (as in
\secref{secNotation}).  We trust that this will not lead to any
confusion as it is clear that \sv{} ranks are completely unrelated to
SLE parameters.}
$\kappa = 4t > 0$ and $\rho = \tfrac{1}{2} \brac{\kappa-4}$.
More precisely, at these parameters a staggered module
$\sS$ with $h^\lft = h^\rgt = h_{0,1} = \tfrac{1}{4} \brac{2-t}$ is realised as a space of local martingales of the SLE${}_\kappa(\rho)$ growth process.
The central charge of this module is $c = c(t) = c(\kappa / 4)$.
The computations do not in general identify the left and the right modules,
but from the Feigin-Fuchs classification, we may for example conclude that
in the case of irrational $\kappa$, $\sH^\lft = \sH^\rgt = \Verma_{h_{0,1}}$
(these Verma modules are of point type).  In other words,
the short exact sequence has the form
\begin{equation}
\dses{\Verma_{h_{0,1}}}{}{\sS}{}{\Verma_{h_{0,1}}} \qquad \text{($t > 0$, $t \notin \QQ$).}
\end{equation}
We illustrate these staggered modules in \figref{fig: SLE l 0} (left).
By \corref{cor: uniqueness l 0}, such staggered modules
are unique when they exist.  But this concrete construction demonstrates
existence, so we can conclude that at least one staggered module exists
for any $t \in \bR_+$, hence two for any central charge $-\infty<c<1$
(one for $c=t=1$).
\end{example}

\begin{example} \label{ex:-2,11,13}
In \cite{GurLog93}, it was shown that the logarithmic singularity in a certain $c=-2$ ($t=2$) \cft{} correlation function implied the existence of a staggered module $\StagMod$ with $h^\lft=h^\rgt=0$.  This module was constructed explicitly in \cite{GabInd96} by fusing the irreducible module $\IrrMod_{-1/8}$ with itself.  The resulting structure is summarised by the short exact sequence
\begin{equation} \label{eqnSES3}
\dses{\VerMod_0 / \VerMod_1}{}{\sS}{}{\VerMod_0 / \VerMod_3},
\end{equation}
and illustrated in \figref{fig: SLE l 0} (right).  In fact, this example is
also related to the SLE construction of \exref{ex: SLE l 0}.  For $\kappa=8$,
the weight $h_{0,1}$ vanishes and the left and right modules can be
computed explicitly to be those given in (\ref{eqnSES3})
\cite{KytFro08}.
\end{example}

{
\psfrag{c=-2}[][]{$c=-2$}
\psfrag{irrational}[][]{$c=c(t)\;$, $\; t \in \bR_+ \setminus \bQ$}
\psfrag{xvec}[][]{$x$}
\psfrag{yvec}[][]{$y$}
\psfrag{Vh01}[][]{$\Verma_{h_{0,1}}$}
\psfrag{Hrgt}[][]{$\Verma_0 / \Verma_3$}
\psfrag{Hlft}[][]{$\Verma_0 / \Verma_1$}
\psfrag{firL0}[][]{$L_0-h_{0,1}$}
\psfrag{secL0}[][]{$L_0$}
\psfrag{0}[][]{$\scriptstyle 0$}
\psfrag{1}[][]{$\scriptstyle 1$}
\psfrag{3}[][]{$\scriptstyle 3$}
\begin{figure}
\begin{center}
\includegraphics[height=4.75cm]{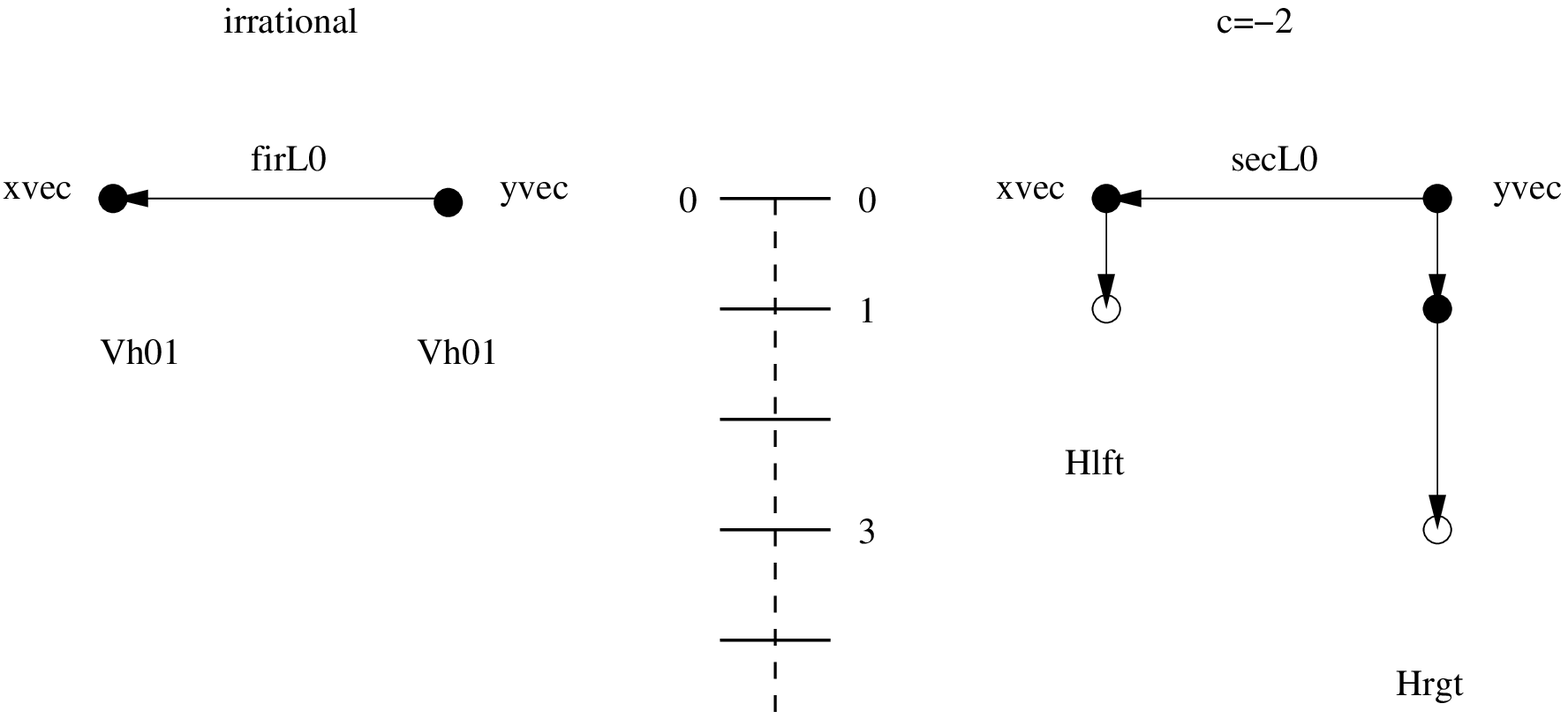}
\caption{An illustration of the staggered modules of Examples
\ref{ex: SLE l 0} (left) and \ref{ex:-2,11,13} (right).  We have indicated
the \sv{} structure of the respective left and right modules by using black
circles for the generating states and \svs{}, and white circles to
indicate \svs{} of the corresponding Verma modules which have been set
to zero.  The dividing scale gives the grades.  It should
be understood that \svs{} of the right module need not ``lift'' to
\svs{} of the staggered module, and are indicated purely to facilitate the
discussion.  (Technically, these lifts are \emph{subsingular
vectors} of the staggered module --- they become singular upon taking an
appropriate quotient.)} \label{fig: SLE l 0}
\end{center}
\end{figure}
}

\noindent We remark that in \exref{ex:-2,11,13}, the vector $L_{-1} y$ is an eigenvector of $L_0$ which does not belong to $\LMod$.  This shows that the submodule of $L_0$-eigenvectors need not coincide with the left module, and in fact need not be a \hwm{} in general.

There is one obvious deficiency inherent in describing staggered modules by
their data $\brac{\omega_1 , \omega_2}$.  This is the fact that neither
$\omega_1$ nor $\omega_2$ are gauge-invariant in general.  Under the gauge
transformations (\ref{eqnGT}), the data transform as follows:
\begin{equation}
\brac{\omega_1 , \omega_2} \longrightarrow
\brac{\omega_1 + L_1 u , \omega_2 + L_2 u} \qquad \text{($u \in \LMod_{\ell}$).}
\end{equation}
This suggests introducing maps $g_u$ for each $u \in \LMod_{\ell}$ which take
$\LMod_{\ell - 1} \oplus \LMod_{\ell - 2}$ into itself via
\begin{equation} \label{eqnDataGT}
\func{g_u}{w_1 , w_2} = \brac{w_1 + L_1 u , w_2 + L_2 u} \qquad
\text{($u \in \LMod_{\ell}$).}
\end{equation}
We will also refer to these maps as gauge transformations.  Clearly
the composition of gauge transformations is the vector space addition of $\Hlft_\ell$.  It is then natural to lift the scalar multiplication of $\Hlft_\ell$ to the set of gauge transformations, making the latter into a vector space itself.  We denote this vector space by $\gauge = \set{g_u : u \in \Hlft_\ell}$.  We further note that the kernel of the map $u \mapsto g_u$ is one-dimensional, spanned by the singular vector $\omega_0$.  Thus, $\gauge$ may be identified with $\LMod_{\ell} / \CC \omega_0$.  In particular, its dimension is
\begin{equation} \label{eqnDimGauge}
\dim \gauge = \dim \Hlft_\ell - 1.
\end{equation}

Because the gauge-transformed data describes the same staggered module as the original data, we will say that the data $\brac{\omega_1, \omega_2}$ and its transforms $\func{g_u}{\omega_1, \omega_2}$ are \emph{equivalent}, for all $u \in \LMod_{\ell}$.  The following result now characterises isomorphic staggered modules completely.

\begin{proposition} \label{prop: equivalence}
Let $\StagMod$ and $\StagMod'$ be staggered modules with the same left and right modules $\LMod$ and $\RMod$ and with respective data $\brac{\omega_1 , \omega_2}$ and $\brac{\omega_1' , \omega_2'}$.  Then, upon identifying the two left modules via $x' = x$, we have $\StagMod' \cong \StagMod$ if and only if the data $\brac{\omega_1 , \omega_2}$ and $\brac{\omega_1' , \omega_2'}$ are equivalent.
\end{proposition}
\begin{proof}
If $\brac{\omega_1' , \omega_2'} = \func{g_u}{\omega_1 , \omega_2}$ for some $u \in \LMod_{\ell}$, then $y' = y + u$ defines the isomorphism $\StagMod' \cong \StagMod$.  Conversely, suppose that $\psi \colon \StagMod' \rightarrow \StagMod$ is an isomorphism extending the identification of the respective left modules (that is, such that $\func{\psi}{x'} = x$).  Then,
\begin{equation}
L_0 y = \RDim y + \omega_0 \qquad \text{and} \qquad L_0 \func{\psi}{y'} = \func{\psi}{\RDim y' + \omega_0'} = \RDim \func{\psi}{y'} + \omega_0,
\end{equation}
so $\func{\psi}{y'} - y$ is an $L_0$-eigenvector of dimension $\RDim$.  We may therefore take $u = \func{\psi}{y'} - y \in \LMod_{\ell}$, hence
\begin{equation}
\func{\psi}{\omega_i'} = L_i \func{\psi}{y'} = L_i \brac{y + u} = \omega_i + L_i u \qquad \text{($i = 1,2$),}
\end{equation}
as required.
\end{proof}

This completes the analysis of when two staggered modules are isomorphic.  It remains however, to study the existence question.  The question of which data $(\omega_1, \omega_2)$ actually correspond to staggered modules is quite subtle, and we will address it in the following sections.  First however, we present two motivating examples from the literature to illustrate this subtlety.

\begin{example} \label{ex:-2,13,15}
In \cite{GabInd96}, it was shown that fusing the two $c=-2$ ($t=2$)
irreducible modules $\IrrMod_{-1/8}$ and $\IrrMod_{3/8}$ results in a
staggered module $\StagMod$ given by the short exact sequence
\begin{equation} \label{eqnSES1}
\dses{\VerMod_0 / \VerMod_3}{}{\StagMod}{}{\VerMod_1 / \VerMod_6}.
\end{equation}
We illustrate $\StagMod$ in \figref{figEx12} (left).  In our notation,
$\ell = 1$, $\omega_0 = L_{-1} x$, $\omega_1 = L_1 y = \beta x$ where
$\beta$ is the beta-invariant of \eqnref{eqnDefBeta}, and
$\omega_2 = L_2 y = 0$.  The explicit calculation shows that $\beta = -1$.

It seems reasonable to suppose that because the data
$\brac{\omega_1 = \beta x , \omega_2 = 0}$ of the staggered module
(\ref{eqnSES1}) is fixed by the beta-invariant, there should exist
a continuum of such modules, one for each value of $\beta$.  This was
suggested in \cite{GabInd96}, referring to Rohsiepe \cite{RohNic96},
but we are not aware of any proof of this fact.  Indeed, one of our
aims (see \exDref{ex: comparison of similar cases}{ex:-2,13,15 again} in \secref{sec: general case}) is to prove and understand why this is indeed the case.
\end{example}

{
\psfrag{c=-2}[][]{$c=-2$}
\psfrag{c=0}[][]{$c=0$}
\psfrag{L0}[][]{$L_0 - 1$}
\psfrag{L1}[][]{$\beta^{-1} L_1$}
\psfrag{w0}[][]{$\omega_0$}
\psfrag{x}[][]{$x$}
\psfrag{y}[][]{$y$}
\psfrag{M13}[][]{$\VerMod_0 / \VerMod_3$}
\psfrag{M15}[][]{$\VerMod_1 / \VerMod_6$}
\psfrag{M12}[][]{$\VerMod_0 / \VerMod_2$}
\psfrag{M14}[][]{$\VerMod_1 / \VerMod_5$}
\psfrag{0}[][]{$\scriptstyle 0$}
\psfrag{1}[][]{$\scriptstyle 1$}
\psfrag{2}[][]{$\scriptstyle 2$}
\psfrag{3}[][]{$\scriptstyle 3$}
\psfrag{5}[][]{$\scriptstyle 5$}
\psfrag{6}[][]{$\scriptstyle 6$}
\psfrag{7}[][]{$\scriptstyle 7$}
\begin{figure}
\begin{center}
\includegraphics[height=8cm]{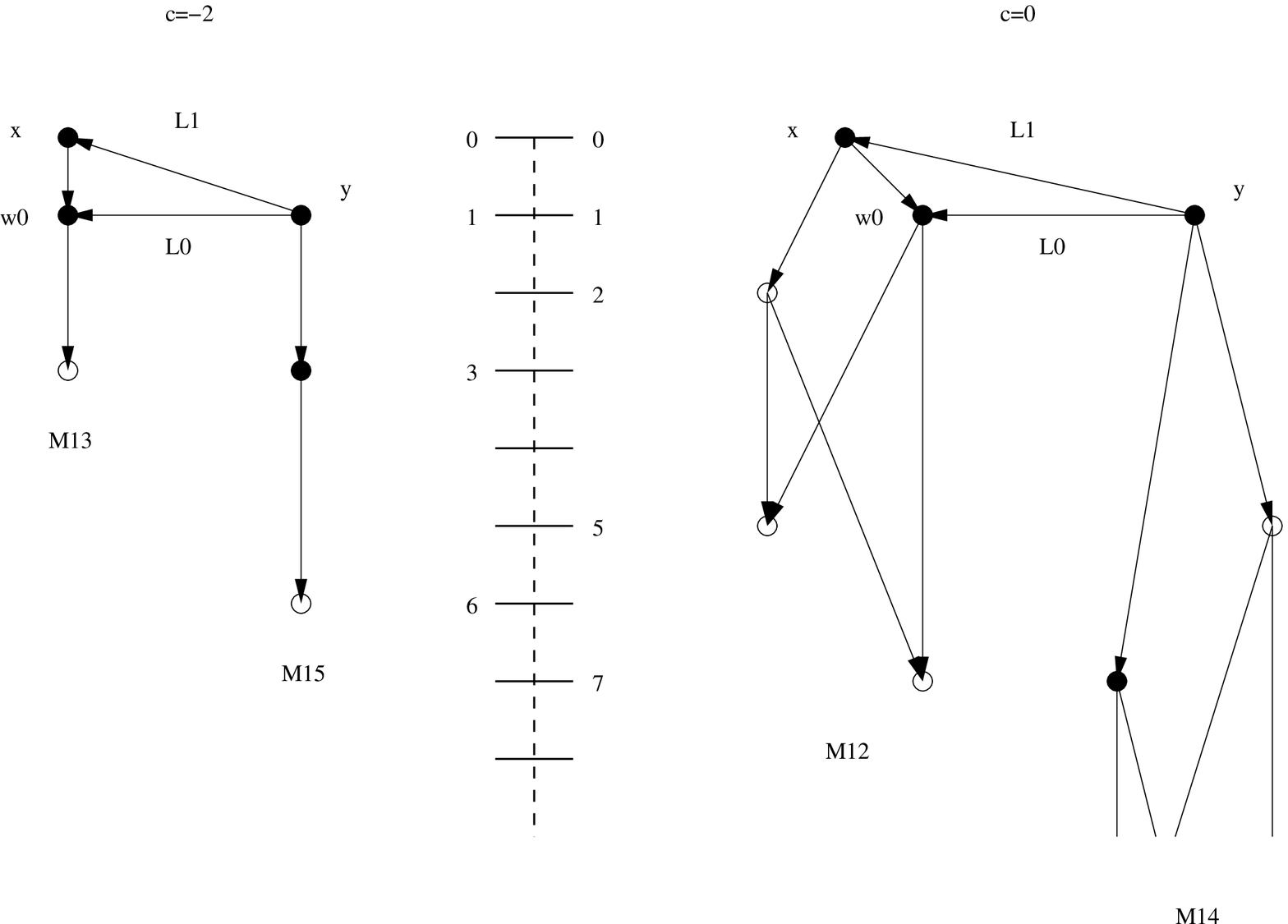}
\caption{An illustration of the staggered modules presented in Examples
\ref{ex:-2,13,15} (left) and \ref{ex:0,12,14} (right).  The structure is
to be interpreted as in \figref{fig: SLE l 0}.  We remark that when
$\beta = 0$, which is possible for the module on the left, the label 
$\beta^{-1} L_1$ should be interpreted as saying that $x$ cannot be
obtained from $y$ under the action of $L_1$, that is, $L_1 y = 0$}.
\label{figEx12}
\end{center}
\end{figure}
}

\begin{example} \label{ex:0,12,14}
A $c=0$ ($t = \tfrac{3}{2}$) staggered module with the short exact sequence
\begin{equation} \label{eqnSES2}
\dses{\VerMod_0 / \VerMod_2}{}{\StagMod}{}{\VerMod_1 / \VerMod_5}
\end{equation}
has appeared several times in the physics literature \cite{GurCon04,EbeVir06,RidPer07}.  We again have $\ell = 1$, $\omega_0 = L_{-1} x$, $\omega_1 = \beta x$ and $\omega_2 = 0$.  This time $\beta$ turns out to be $-\tfrac{1}{2}$.  This module is also illustrated in \figref{figEx12} (right).

One could be forgiven for thinking that because of the similarity of this
example and the last, there will be a continuum of staggered modules with
the exact sequence (\ref{eqnSES2}), parametrised by $\beta$.  But
surprisingly, this is not the case.  It was argued
in \cite{RidLog07} that $\beta = -\tfrac{1}{2}$ is the only possible value
for such a staggered module, and hence that such a staggered module is 
unique (up to isomorphism).  We shall prove this in 
\secref{sec: general case}
(\exDref{ex: comparison of similar cases}{ex:-2,13,15 again}).
\end{example}

There are some obvious structural differences between
\exDref{ex:-2,13,15}{ex:0,12,14}, but it is not immediately clear what
causes the observed restriction on the isomorphism classes of staggered
modules.  In fact, the desire to understand this mechanism is precisely
the original motivation for the research reported here.

\begin{example} \label{ex: SLE l 1}
The above two examples may in fact be regarded as members of another
family of staggered modules parametrised by $t$.  For
$t\in\bR_+\setminus \set{1}$, this family can again be realised
concretely as a module of local martingales of SLEs, with $\kappa=4t$ 
and $\rho=-2$ \cite{KytFro08}.  Each member has $h^\lft=0$ and $h^\rgt = 1$, 
but as in \exref{ex: SLE l 0}, determining the precise identity of $\LMod$
and $\RMod$ requires non-trivial calculations in general.  However, when
$\kappa$ is irrational, these identities are settled automatically,
because then $\Verma_{h^\lft}$ is of link type and $\Verma_{h^\rgt}$ is
of point type (irreducible).  By \propref{propSVw0}, $\omega_0 \in \LMod$
is non-vanishing, so $\LMod = \Verma_{h^\lft}$.  The exact sequence is
therefore
\begin{equation}
\dses{\VerMod_0}{}{\sS}{}{\VerMod_1} \qquad
\text{($t>0$, $t \notin \bQ$).}
\end{equation}
The beta-invariant was computed in \cite{KytFro08} (see also
\cite{RidLog07}) for all $t \in \bR_+ \setminus \set{1}$ to be $\beta=1-t$,
which coincides with the values in \exDref{ex:-2,13,15}{ex:0,12,14}
(when $t=2$ and $t=\tfrac{3}{2}$ respectively).  For these two rational
values, the left and right modules were also computed explicitly in the
SLE picture, finding
agreement with the fusion computations above. Thus this family of
examples shows an interesting interplay of continuously varying
beta-invariant, but discontinuously varying left and right modules.
\end{example}

\section{Constructing Staggered Modules:  Generalities} \label{sec: construction}

In the previous section, we have introduced staggered modules and determined some simple necessary conditions for their existence.  We now turn to the more subtle question of sufficient conditions for existence.  As we have seen in \exref{ex:0,12,14}, it is not true that given left and right modules, every possible choice of data $\brac{\omega_1 , \omega_2}$ describes a staggered module.  We are therefore faced with the task of having to determine which data give rise to staggered modules.  Such data will be termed \emph{admissible}.

One simple reason \cite{RohRed96} why a given set of data $\brac{\omega_1 , \omega_2}$ might fail to correspond to any staggered module is that there could exist an element $U \in \uea$ such that\footnote{We include a seemingly arbitrary ``$-$'' sign in the equation which follows (and in similar later equations) because it turns out to be convenient in the long run to be consistent with expressions such as that found in \eqnref{eq: uea null in Verma}.}
\begin{equation} \label{eqnUL1UL2}
U = U_1 L_1 = -U_2 L_2, \qquad \text{but} \qquad U_1 \omega_1 + U_2 \omega_2 \neq 0.
\end{equation}
For then, $U y = U_1 \omega_1 \neq -U_2 \omega_2 = U y$, a contradiction.  We mention that given any $U = U_1 L_1 = -U_2 L_2 \in \uea L_1 \cap \uea L_2$, the elements $U_1$ and $U_2$ are uniquely determined because $\uea$ has no zero-divisors.

We therefore define the subset
\begin{equation} \label{eqnDefData1}
\data = \set{(w_1, w_2) \in \Hlft_{\ell-1} \oplus \Hlft_{\ell-2}
\st U_1 w_1 + U_2 w_2 = 0 \text{ for all } U = U_1 L_1 = -U_2 L_2 \in \uea L_1 \cap \uea L_2}.
\end{equation}
With this notation, our necessary condition on the data becomes:

\begin{lemma} \label{lem: consistency of L1 L2}
If a staggered module with data $\brac{\omega_1, \omega_2}$ exists, then
$\brac{\omega_1, \omega_2} \in \data$.
\end{lemma}

\noindent We can obtain a useful simplification of this condition through \PBW-ordering the $U \in \uea L_1 \cap \uea L_2$.

\begin{lemma} \label{lem: positive to right in intersection}
$\uea L_1 \cap \uea L_2 = \uea^{\leqslant 0} (\uea^+ L_1 \cap \uea^+ L_2)$.
\end{lemma}
\begin{proof}
If $U \in \uea L_1 \cap \uea L_2$, we may write $U = U_1 L_1 = -U_2 L_2$ with the $U_i$ \PBW-ordered:  $U_i = \sum_n U_{i,n}^{\leqslant 0} U_{i,n}^+$, with $U_{i,n}^{\leqslant 0} \in \uea^{\leqslant 0}$ and $U_{i,n}^+ \in \uea^+$.  Thus,
\begin{equation}
U = \sum_n U_{1,n}^{\leqslant 0} U_{1,n}^+ L_1 = -\sum_n U_{2,n}^{\leqslant 0} U_{2,n}^+ L_2.
\end{equation}
Since similarly ordering $U$ in its entirety will not affect the $U_{i,n}^{\leqslant 0}$ factors, the linear independence of \PBW{} monomials implies that (with an appropriate shuffling of the index $n$) we may take $U_{1,n}^{\leqslant 0} = U_{2,n}^{\leqslant 0}$.  It follows, again from linear independence, that $U_{1,n}^+ L_1 = -U_{2,n}^+ L_2$.  This proves that $\uea L_1 \cap \uea L_2 \subseteq \uea^{\leqslant 0} (\uea^+ L_1 \cap \uea^+ L_2)$ and the reverse inclusion is trivial.
\end{proof}

We apply \lemref{lem: positive to right in intersection} to the conditions of \eqnref{eqnUL1UL2} as follows.  The first of these just states that $U \in \uea L_1 \cap \uea L_2$, hence \lemref{lem: positive to right in intersection} lets us write $U = \sum_n U_n^{\leqslant 0} U_{1,n}^+ L_1 = -\sum_n U_n^{\leqslant 0} U_{2,n}^+ L_2$, for some $U_n^{\leqslant 0} \in \uea^{\leqslant 0}$ and $U_{i,n}^+ \in \uea^+$, where
\begin{equation}
U_{1,n}^+ L_1 + U_{2,n}^+ L_2 = 0,
\end{equation}
for all $n$.  Moreover, the second condition of (\ref{eqnUL1UL2}) is now $\sum_n U_n^{\leqslant 0} U_{1,n}^+ \omega_1 + \sum_n U_n^{\leqslant 0} U_{2,n}^+ \omega_2 \neq 0$, which implies that
\begin{equation}
U_{1,n}^+ \omega_1 + U_{2,n}^+ \omega_2 \neq 0,
\end{equation}
for some $n$.  It follows that in \eqnref{eqnUL1UL2}, we may suppose that $U_1$ and $U_2$ belong to $\uea^+$, without any loss of generality.  In other words, if an element $U \in \uea L_1 \cap \uea L_2$ spoils the admissibility of $\brac{\omega_1 , \omega_2}$, then there is an element spoiling admissibility in $\uea^+ L_1 \cap \uea^+ L_2$.

This somewhat lengthy argument then allows us to conclude that $\data$ may be equivalently defined as
\begin{equation} \label{eqnDefData}
\data = \set{(w_1, w_2) \in \Hlft_{\ell-1} \oplus \Hlft_{\ell-2}
\st U_1 w_1 + U_2 w_2 = 0 \text{ for all }
U = U_1 L_1 = -U_2 L_2 \in \uea^+ L_1 \cap \uea^+ L_2}.
\end{equation}
The value of this slight simplification lies in the fact that the homogeneous subspaces of $\uea^+ L_1 \cap \uea^+ L_2$ are finite-dimensional.

\begin{lemma} \label{lem: L1 L2}
For $m > 0$, the dimension of $\brac{\uea^+ L_1 \cap \uea^+ L_2}_{-m} = \uea^+_{1-m} L_1 \cap \uea^+_{2-m} L_2$ is equal to $\func{d}{m} = \tpartnum{m-1} + \tpartnum{m-2} - \tpartnum{m}$.  When $m=0$, this dimension is $0$.
\end{lemma}
\begin{proof}
As $L_1$ and $L_2$ generate $\vir^+$, we have
$\brac{\uea^+ L_1 + \uea^+ L_2}_{-m} = \uea^+_{-m}$ for $m>0$.
Taking dimensions of this equality we get
$\dmn \uea^+_{1-m} + \dmn \uea^+_{2-m} -
\dmn \brac{\uea^+ L_1 \cap \uea^+ L_2}_{-m} = \dmn \uea^+_{-m}$,
which leads to the
asserted formula.
\end{proof}

\noindent As an aside to the advanced reader, we mention that by
treating $\uea^+$ as a Virasoro module with $h=c=0$
(we set\footnote{The precise way in which one does this parallels that discussed in the context of Verma modules. One starts
with the trivial one-dimensional representation of $\vir^{\leq 0}$
and the induced $\vir$-module is naturally identified
as a graded vector space with $\uea^+$.} $\vir^{\leqslant 0} \: \wun = 0$),
$\uea^+ L_1 \cap \uea^+ L_2$ may be
identified as the submodule generated by the \svs{} at grades $-5$ and
$-7$.  Indeed, thinking of $\uea^+$ as a \emph{lowest} weight Verma
module, our intersection corresponds to the intersection of the
submodules generated by the rank $1$ \svs{} at grades $-1$ and $-2$. 
The Feigin-Fuchs classification for lowest weight Verma modules states
that this is generated by the rank $2$ \svs{}, which turn out to have
grades $-5$ and $-7$ (as stated).

We tabulate the first few of these dimensions for convenience:

\begin{table}[!ht]
\begin{center}
\begin{tabular}{R|*{17}{|C}}
m & 0 & 1 & 2 & 3 & 4 & 5 & 6 & 7 & 8 & 9 & 10 & 11 & 12 & 13 & 14 & 15 & \cdots \\
\hline
\func{d}{m} & 0 & 0 & 0 & 0 & 0 & 1 & 1 & 3 & 4 & 7 & 10 & 16 & 21 & 32 & 43 & 60 & \cdots
\end{tabular}
\end{center}
\end{table}

\noindent Note that if
$U = U_1 L_1 = -U_2 L_2 \in \uea^+_{1-m} L_1 \cap \uea^+_{2-m} L_2$ with
$m > \ell$, then $U_1 w_1$ and $U_2 w_2$ both vanish for all $\brac{w_1 , w_2} \in \LMod_{\ell - 1} \oplus \LMod_{\ell - 2}$ (for dimensional reasons).  We therefore need $\ell \geqslant 5$ to find examples where
$\data \neq \Hlft_{\ell-1} \oplus \Hlft_{\ell-2}$.  We also point out that
$\data$ is not necessarily equal to the set of admissible data. 
\exref{ex:0,12,14} provides an illustration of this fact:  The dimension
of $\data$ is $\dim \brac{\LMod_0 \oplus \LMod_{-1}} = 1$ in this case, but 
the set of admissible data is a singleton.

\begin{example} \label{ex:0,14,18}
A staggered module $\StagMod$ with $c=0$ ($t=\tfrac{3}{2}$) and short exact sequence
\begin{equation}
\dses{\VerMod_1 / \VerMod_5}{}{\StagMod}{}{\VerMod_7 / \VerMod_{15}},
\end{equation}
was constructed in \emph{\cite{EbeVir06}}.  Note that $\ell = 6$.  Its
beta-invariant was shown in \emph{\cite{RidPer07}} to be
$\beta = -\tfrac{10780000}{243}$ (with our normalisation for $\omega_0$), 
where it was also argued to be the unique such value. 
What is interesting here is that the authors noted that this example
presents some subtlety upon trying to ``fix the gauge'' before computing
$\beta$.  It is this subtlety which we want to explain here.

With our notation, the problem arose when the authors tried to determine $\omega_1 \in \LMod_5$ and $\omega_2 \in \LMod_4$ in terms of the (unknown) $\beta$.  Since $\dim \LMod_5 = 6$, $\dim \LMod_4 = 4$ and there are $\dim \gauge = \dim \LMod_6 - 1 = 8$ independent gauge transformations, they could assume that $\omega_1 = 0$ and $\omega_2 = \brac{a L_{-4} + b L_{-2}^2} x$.  There were therefore two unknowns $a$ and $b$.  The definition of the beta-invariant then gave a single linear relation connecting it with $a$ and $b$.

Whilst the authors of \emph{\cite{RidLog07}} were able to divine another linear relation between $a$ and $b$, thereby determining them in terms of $\beta$ and completing the gauge-fixing, we can understand this problem as arising from the existence of non-trivial elements of $\uea^+ L_1 \cap \uea^+ L_2$.  Indeed, $\brac{\uea^+ L_1 \cap \uea^+ L_2}_{-5}$ is spanned by
\begin{equation} \label{eqnGen5}
\brac{L_1^2 L_2 + 6 L_2^2 - L_1 L_3 + 2 L_4} L_1 = \brac{L_1^3 + 6 L_1 L_2 + 12 L_3} L_2,
\end{equation}
and left-multiplying by $L_1$ gives a spanning element of $\brac{\uea^+ L_1 \cap \uea^+ L_2}_{-6}$.  It follows that the assumed data $\brac{\omega_1 = 0 , \omega_2 = \brac{a L_{-4} + b L_{-2}^2} x}$ is not in $\data$ (and hence not admissible) unless
\begin{equation}
L_1 \brac{L_1^2 L_2 + 6 L_2^2 - L_1 L_3 + 2 L_4} \omega_1 = L_1 \brac{L_1^3 + 6 L_1 L_2 + 12 L_3} \omega_2.
\end{equation}
Evaluating this constraint gives the second relation found in \emph{\cite{RidLog07}} through other, less canonical, means.
\end{example}

To attack the question of which $\brac{\omega_1, \omega_2}$ can arise as the
data of a staggered module $\StagMod$, given left and right modules $\LMod$ and $\RMod$, we consider the following explicit construction
(generalising that of Rohsiepe \cite{RohRed96}).
We start with the Virasoro module
$\Hlft \oplus \uea$, where $\vir$ is understood to act on $\uea$ by
left-multiplication.  We let $\sN$ be the submodule of $\Hlft \oplus \uea$
generated by
\begin{equation} \label{eqnNGens}
\begin{split}
&\brac{\omega_0, h^\rgt-L_0}, \qquad
\brac{\omega_1, -L_1}, \qquad \brac{\omega_2,-L_2}, \\
\text{and} \qquad &\brac{\varpi, -\overline{X}} \qquad  \text{or} \qquad
(\varpi^\pm, -\overline{X}^\pm), \qquad  \text{when appropriate.}
\end{split}
\end{equation}
Here, we understand that when required, $\varpi$ (or $\varpi^\pm$) is deduced from the $\omega_j$ as in the proof of \propref{propDetBy}.
The idea is that $\wun \in \uea$ will project onto $y \in \StagMod$ upon quotienting by $\sN$.  More specifically, we will attempt to construct $\StagMod$ as $\brac{\LMod \oplus \uea} / \sN$, requiring then only a precise analysis of when this succeeds.
Denote by $\pi^\rgt: \Hlft \oplus \uea \rightarrow \uea$ the projection
onto the second component.
The question of whether this construction recovers $\StagMod$ turns out to boil down to whether the submodule $\sN^\circ = \sN \cap \Kern \pi^\rgt$ is trivial or not.

\begin{theorem} \label{thmConstruction}
Given $\LMod$, $\RMod$, $\omega_1 \in \LMod_{\ell-1}$ and $\omega_2 \in \LMod_{\ell-2}$, we have the following.
\begin{description}
\item[\textup{(i)}]
If $\sN^\circ = \set{0}$ then $(\Hlft \oplus \uea)/\sN$
is a staggered module with the desired short exact sequence
\begin{equation}
0 \longrightarrow \Hlft \overset{\iota}{\longrightarrow} \frac{(\Hlft \oplus \uea)}{\sN} \overset{\pi}{\longrightarrow} \Hrgt
\longrightarrow 0
\end{equation}
and data $\brac{\omega_1, \omega_2}$.
\item[\textup{(ii)}]
If $\sN^\circ \neq \set{0}$ then a staggered module with
the desired exact sequence and data does not exist.
\end{description}
\end{theorem}
\begin{proof}
Denote by $\pi_\sN : \Hlft \oplus \uea \rightarrow \brac{\Hlft \oplus \uea} / \sN$ the canonical projection, and assume (at first) that $\sN^\circ = \set{0}$.
We will construct the required homomorphisms $\iota$ and $\pi$ by imposing commutativity of the following diagram:
\begin{equation} \label{eqnLeDiagram}
\begin{CD}
0 @>>> \LMod @>\iota^{\lft}>> \LMod \oplus \uea @>\pi^{\rgt}>> \uea @>>> 0 \\
@. @| @VV\pi_{\mathcal{N}}V @VV\pi_{\mathcal{I}}V @. \\
0 @>>> \LMod @>\iota>> {\displaystyle \frac{\LMod \oplus \uea}{\mathcal{N}}} @>\pi>> \RMod @>>> 0
\end{CD}
\mspace{10mu}.
\end{equation}
Here, $\iota^\lft$ denotes the obvious injection $u \mapsto (u,0)$ (the top row is therefore exact) and $\pi_{\mathcal{I}}$ denotes the canonical projection onto the quotient of $\uea$ by the submodule (left ideal) $\mathcal{I}$ generated by $L_0-h^\rgt$, $L_1$, $L_2$ and $\overline{X}$.

Observe then that $\iota = \pi_\sN \circ \iota^\lft$ has kernel
$\Imag \iota^{\lft} \cap \sN = \sN^{\circ} = \set{0}$, hence is injective. 
On the other hand, the map $\pi$ satisfies
$\pi \circ \pi_{\sN} = \pi_{\mathcal{I}} \circ \pi^{\rgt}$, which in fact 
defines it as $\pi_{\mathcal{I}} \circ \pi^{\rgt}$ maps
$\sN = \Kern \pi_{\sN}$ to zero by construction.  The map
$\pi$ is clearly surjective as both $\pi^{\rgt}$ and $\pi_{\mathcal{I}}$ are. 
It remains to check that the bottom row is exact in the middle. From the
exactness of the top row we get
\begin{equation}
\pi \circ \iota = \pi_{\mathcal{I}} \circ \pi^\rgt \circ \iota^{\lft} = 0, \quad \text{hence} \quad \Imag \iota \subseteq \Kern \pi.
\end{equation}
On the other hand, if $\func{\pi \circ \pi_{\sN}}{w,U} = 0$ for some $\brac{w,U} \in \LMod \oplus \uea$, then $U \in \mathcal{I}$ by commutativity of (\ref{eqnLeDiagram}).  By definition of $\mathcal{I}$ and $\sN$, $\brac{w,U} = \brac{w',0} \pmod{\sN}$ for some $w' \in \LMod$, hence
\begin{equation}
\func{\pi_{\sN}}{w,U} = \func{\pi_{\sN} \circ \iota^{\lft}}{w'} = \func{\iota}{w'}, \quad \text{hence} \quad \Kern \pi \subseteq \Imag \iota.
\end{equation}
The module $(\Hlft \oplus \uea)/\sN$ is then staggered and the data are
correct, because
\begin{equation}
(L_0 - \RDim) y = (\omega_0 , 0) = \iota(\omega_0) \quad \text{and} \quad
L_j y = (\omega_j , 0) = \iota(\omega_j) \quad \pmod{\sN},
\end{equation}
where $y = (0, 1)$ and $x = (\Lhws, 0)$ $\pmod{\sN}$.  This proves (i).

If $\sN^\circ \neq \{0\}$, then (given $\LMod$) there exists
$U_0, U_1, U_2, \overline{U} \in \uea$ such that
\begin{equation}
U_0(L_0-h^\rgt) + U_1 L_1 + U_2 L_2 + \overline{U} \, \overline{X} = 0, \quad
\text{but} \quad U_0 \omega_0 + U_1 \omega_1 + U_2 \omega_2 + \overline{U} \varpi \neq 0.
\end{equation}
Suppose that $\sS$ was a staggered module with the desired exact sequence and data, and choose $y \in \sS$ such that $\pi(y)=x^\rgt$ and
$L_j y = \omega_j$. Now applying the first of these equations to $y$ would give zero, contradicting the second.  This proves (ii).
\end{proof}

The r\^{o}le that $\sN^{\circ}$ plays in this construction of a staggered module is best seen by regarding $\sN^{\circ} = \sN \cap \Imag \iota^{\lft}$ as a submodule of $\LMod$.  If non-trivial, $\sN^{\circ}$ is generated by \svs{} of $\LMod$.  The quotient of $\LMod \oplus \uea$ by $\sN$ will then  no longer have a left module isomorphic to $\LMod$, but will be some quotient thereof.  For example, if $x \in \sN^{\circ}$, then all of $\LMod$ is ``quotiented away'' and the above construction gives a \hwm{}, not a staggered module.  Similarly, if $\omega_0 \in \sN^{\circ}$ but $x \notin \sN^{\circ}$, then the construction results in an indecomposable module on which $L_0$ is diagonalisable.  It is only when $\sN^{\circ} = \set{0}$ that $\LMod$ is preserved, and then \thmref{thmConstruction} tells us that we do indeed obtain a staggered module with the correct left and right modules and data.

Before concluding this section, let us first make two brief observations relating to the above construction arguments.  These allow us to answer the question of existence or non-existence of a staggered module, assuming we have already answered the question for another related staggered module.  Roughly speaking, existence becomes easier if we take a smaller left module or a bigger right module.  The precise statement for the left module is as follows.

\begin{proposition} \label{prop: monotonicity left}
Suppose that there exists a staggered module $\Stagg$ with exact sequence
\begin{equation}
\dses{\Hlft}{}{\Stagg}{}{\Hrgt}
\end{equation}
and data $\bigl( \omega_1, \omega_2 \bigr) \in \Hlft_{\ell-1} \oplus \Hlft_{\ell-2}$.  If $\hat{\mathcal{J}}$ is a submodule of $\Hlft$ not containing $\omega_0$, then there exists a staggered module $\hat{\Stagg}$ with exact sequence
\begin{equation}
\dses{\hat{\sH}^\lft}{}{\hat{\Stagg}}{}{\Hrgt} \qquad \textup{($\hat{\sH}^\lft = \Hlft / \hat{\mathcal{J}}$)}
\end{equation}
and data $\bigl( [\omega_1], [\omega_2] \bigr) \in \hat{\sH}^\lft_{\ell-1} \oplus \hat{\sH}^\lft_{\ell-2}$.  Indeed, we may identify $\hat{\Stagg}$ with $\Stagg / \hat{\mathcal{J}}$.
\end{proposition}

\noindent This follows from the fact that $\Hlft$ is a submodule of $\Stagg$.  We only require $\omega_0 \notin \hat{\mathcal{J}}$ to ensure that the quotient $\Stagg / \hat{\mathcal{J}}$ is still staggered.

For the right module we have instead the following, somewhat less trivial, result.

\begin{proposition} \label{prop: monotonicity right}
Suppose that there exists a staggered module $\Stagg$ with exact sequence
\begin{equation}
\dses{\Hlft}{}{\Stagg}{}{\Hrgt}
\end{equation}
and data $\bigl( \omega_1, \omega_2 \bigr) \in \Hlft_{\ell-1} \oplus \Hlft_{\ell-2}$.  If $\Hrgt$ is a quotient of the \hwm{} $\check{\sH}^\rgt$, then there exists a staggered module $\check{\Stagg}$ with exact sequence
\begin{equation}
\dses{\Hlft}{}{\check{\Stagg}}{}{\check{\sH}^\rgt}
\end{equation}
and the same data $\bigl( \omega_1, \omega_2 \bigr) \in \Hlft_{\ell-1} \oplus \Hlft_{\ell-2}$.  Moreover, we may identify $\Stagg$ as a quotient of $\check{\Stagg}$.
\end{proposition}
\begin{proof}
We will show that the submodules of
$\LMod \oplus \uea$ used in the construction of \thmref{thmConstruction}
satisfy $\check{\sN} \subseteq \sN$, so
$\check{\sN}^\circ \subseteq \sN^\circ = \set{0}$ (identifying the left
modules of $\check{\Stagg}$ and $\Stagg$ in the obvious way).  As $\Hrgt$
is a (non-zero) quotient of $\check{\sH}^\rgt$, $\check{h}^\rgt = h^\rgt$,
and we see that $\check{\omega}_0 = \omega_0$.  The proposition states that
the data of $\check{\Stagg}$ and $\Stagg$ are likewise identified, so the
only difference between the generators (\ref{eqnNGens}) of $\check{\sN}$
and $\sN$ is that the former includes
$\bigl( \check{\varpi} , -\check{\overline{X}} \bigr)$, whereas in the
latter we have instead
$\bigl( \varpi , -\overline{X} \bigr)$.\footnote{Here we lighten the
notation by omitting possible superscripts ``$\pm$''.  We also note
that if $\check{\Hrgt}$ were Verma, then the inclusion
$\check{\sN} \subseteq \sN$ would follow immediately.  In the proof
we may therefore exclude this trivial case and assume that both
$\check{\overline{X}}$ and $\overline{X}$ are non-zero.}
But, as $\Hrgt$ is a quotient of $\check{\sH}^\rgt$,
we may write $\check{\overline{X}} = \chi \overline{X}$ for some singular
$\chi \in \uea^-$, so if we can show that
$\check{\varpi} = \chi \varpi$, then $\check{\sN} \subseteq \sN$ follows
and we are done.  Moreover, this would allow us to write
\begin{equation}
\StagMod = \frac{\LMod \oplus \uea}{\sN} = \left. \frac{\LMod \oplus \uea}{\check{\sN}} \middle/ \frac{\sN}{\check{\sN}} \right. = \frac{\check{\StagMod}}{\sN / \check{\sN}},
\end{equation}
realising $\StagMod$ as a quotient of $\check{\StagMod}$.

It remains then to prove that $\check{\varpi} = \chi \varpi$.  This is a
straight-forward check based on \propref{propDetBy}.  To whit, the proof of this
proposition tells us that $\check{\varpi}$ is completely determined by the
conditions (one for each $n>0$)
\begin{equation}
L_n \check{\varpi} = U_0 \omega_0 + U_1 \omega_1 + U_2 \omega_2, \qquad \text{where} \qquad L_n \check{\overline{X}} = U_0 \brac{L_0 - \RDim} + U_1 L_1 + U_2 L_2.
\end{equation}
By hypothesis, $\StagMod$ exists, so there is a $y \in \StagMod$ defining the $\omega_j$ as in \eqnref{eqnDefOmega}.  Now, $\check{\overline{X}} = \chi \overline{X}$ implies that
\begin{equation}
L_n \check{\varpi} = \Bigl( U_0 \brac{L_0 - \RDim} + U_1 L_1 + U_2 L_2 \Bigr) y = L_n \check{\overline{X}} y = L_n \chi \varpi \qquad \text{for all $n>0$}.
\end{equation}
Since $\LMod$ has no (non-zero) \svs{} at the grade of $\check{\varpi}$ (\propref{propVanSVs}), we conclude that $\check{\varpi} = \chi \varpi$, as required.  The proof is therefore complete.
\end{proof}

\begin{corollary} \label{corItsASubmodule}
Every staggered module can be realised as a quotient of a staggered module whose right module is Verma.
\end{corollary}

To summarise, \thmref{thmConstruction} shows that the data
$\brac{\omega_1 , \omega_2}$ is admissible if and only if the module
$\sN^{\circ}$ (whose definition depends upon $\omega_1$ and $\omega_2$)
is trivial.  This construction is therefore fundamental for the question
of existence of staggered modules, but as such is it not yet completely
transparent.  What is missing are easily checked sufficient conditions to 
guarantee that $\sN^\circ = \{0\}$.  The best way to proceed is to first 
analyse the case in which the right module $\sH^\rgt$ is a Verma module. 
By \propref{prop: monotonicity right}, this case is the least restrictive,
and we devote \secref{sec: right Verma} to this task, which is decidedly
non-trivial in itself. The treatment of general $\Hrgt$ can
then be reduced to the analysis of certain submodules of the $\sH^\rgt$ Verma
case, by \corref{corItsASubmodule}. This is the subject of
\secref{sec: general case}.  First however, we must briefly digress in
order to introduce an important auxiliary result which will be used in
both \secDref{sec: right Verma}{sec: general case}.

\section{The Projection Lemma} \label{sec: projections}

This section is devoted to an auxiliary result which we call the \emph{Projection Lemma} (\lemref{lemProjection}).  This will be used at several key places in the sequel, in particular \secDref{sec: data from a subspace}{secRightModProj}, but in slightly different contexts.  We will therefore present it in a somewhat general form.  The relevance to the development thus far should however be readily apparent.

Recall that we defined a set $\data$ in \eqnref{eqnDefData}.  We generalise
this definition slightly:
\begin{equation}
\data_m = \set{\brac{w_1, w_2} \in \Hlft_{m-1} \oplus \Hlft_{m-2} \st U_1 w_1 + U_2 w_2 = 0 \text{ for all } U = U_1 L_1 = -U_2 L_2 \in \uea^+ L_1 \cap \uea^+ L_2}.
\end{equation}
We will always take $m$ to be the grade of a \sv{}, $m=\ell_\rnk$ or $m=\ell_\rnk^\pm$.  Thus $\data$ coincides with $\data_{\ell}$.  Similarly, we defined a vector space $\gauge$ that acts on $\data$, in fact on
$\Hlft_{\ell-1} \oplus \Hlft_{\ell-2}$, by \eqnref{eqnDataGT}.  
We also generalise this, defining $\gauge_m$ to be the vector space of
transformations $g_u$ of $\Hlft_{m-1} \oplus \Hlft_{m-2}$ which take the form
\begin{equation}
\func{g_u}{w_1 , w_2} = \brac{w_1 + L_1 u , w_2 + L_2 u}, \qquad \text{($u \in \Hlft_m$).}
\end{equation}
Again, $\gauge$ coincides with $\gauge_{\ell}$.

We next define a filtration of $\data_m$ which is induced by the singular
vector structure of $\LMod$.  Recall that at the end of \secref{secNotation},
we discussed the Feigin-Fuchs classification of Virasoro Verma modules and
introduced notation for their \svs{}.  The structure and notation
differed according to whether the Verma module was of chain
(and link) or braid type,
and so the explicit forms of our filtration must also differ according to 
these two cases.

\noindent \textbf{Chain case:}
Define subspaces of $\data_m$ in which both
$w_i$ are descendants of the singular vector $X_k x$:
\begin{equation}
\data_m^{(k)} = \set{(w_1,w_2) \in \data_m \st w_1 , w_2 \in \uea X_k x}.
\end{equation}
When $m=\ell_\rnk$, this gives a filtration of $\data_m$ of the form
\begin{equation}
\data_m = \data_m^{\brac{0}} \supseteq \data_m^{\brac{1}} \supseteq
\data_m^{\brac{2}} \supseteq \cdots \supseteq \data_m^{\brac{\rnk-2}} 
\supseteq \data_m^{\brac{\rnk-1}} .
\end{equation}
Clearly, $\data_m^{\brac{k}} = \set{0}$ for all $k \geqslant \rnk$.  An obvious remark that is nevertheless worth keeping in mind is that the spaces $\data_m^{\brac{k}}$ may be trivial even when $k < \rnk$, for example if $X_k x = 0$.

\noindent \textbf{Braid case:}
We define subspaces of $\data_m$ similarly:\footnote{Note that in the braid case the highest weight submodules generated by the \svs{} are not nested, which is why the definition of $\data_m^{(k;-)}$ requires $w_j$ to be in the sum of two highest weight submodules instead.}
\begin{subequations}
\begin{align}
\data_{m}^{(k;+)} &=
    \set{(w_1,w_2) \in \data_m \; | \; w_j \in \uea X_{k}^+ x} \\
\data_{m}^{(k;-)} &= \set{(w_1,w_2) \in \data_m \; | \;
w_j \in \uea X_{k}^- x + \uea X_{k}^+ x}.
\end{align}
\end{subequations}
When $m$ is the grade of a rank $\rnk$ \sv{} ($m=\ell^\pm_\rnk$), 
these subspaces are nested as
\begin{equation}
\data_m = \data_m^{(0;+)} \supseteq \data_m^{(1;-)} \supseteq
\data_m^{(1;+)} \supseteq \data_m^{(2;-)} \supseteq \cdots \supseteq
\data_m^{(\rnk-2;+)} \supseteq \data_m^{(\rnk-1;-)}.
\end{equation}
We note again that if $\Hlft$ contains no (non-zero) \svs{} of rank $k$,
then $\data_m^{(k;\pm)} = \set{0}$.  However, this case differs from the chain case in that there is the possibility that for a certain rank, one of the \svs{} of $\Hlft$ is present whilst the other is not.

\begin{lemma}[The Projection Lemma] \label{lemProjection}
Let $m=\ell_\rnk$ ($m=\ell_\rnk^\pm$) be the grade of a \sv{}. Then
for any $\brac{w_1,w_2} \in \data_m$, there exists a $g_u \in G_m$ such that
$\func{g_u}{w_1, w_2}$ belongs to the subspace $\data_m^{(\rnk-1)}$
($\data_m^{(\rnk-1;-)}$).
\end{lemma}

Before presenting the proof, let us pause to first describe the idea behind it (in non-rigorous terms).  We will prove the required result iteratively.  In the chain case, we will show how to take an element of $\data_m^{(k)}$ and make a gauge transformation so as to get an (equivalent) element of $\data_m^{(k+1)}$.  In the braid case, we will do two slightly different alternating steps, showing how to go from $\data_m^{(k;-)}$ to $\data_m^{(k;+)}$ and from $\data_m^{(k;+)}$ to $\data_m^{(k+1;-)}$.  Composing all of these transformations then gives the required result in each case.

The way in which we transform from one subspace to the next is most
transparent when we assume that we are working within a genuine staggered
module, with data given by $\omega_j = L_j y$ for $j = 1,2$.  Under this 
hypothesis, we will outline the steps required, assuming the chain case 
for notational simplicity.  Suppose then that
$\brac{\omega_1 , \omega_2} \in \data_m^{(k)}$, with $m=\ell$.  
We first note that we can
obtain $X_k x$ from $\omega_1$ or $\omega_2$ by acting with $\uea$ if and 
only if we can obtain it from $y$.  Thus, we take a basis $\set{Z_{\mu}}$
of $\uea^-$ at grade $m - \ell_k$, and consider the complex numbers
$\zeta_{\mu}$ defined by $Z_{\mu}^{\dagger} y = \zeta_{\mu} X_k x$.  By
gauge-transforming $y \rightarrow y' = y + z$ appropriately, it turns out
that we can tune all of the $\zeta_{\mu}$ to zero.  It then follows that we
cannot obtain $X_k x$ from $y'$ by acting with $\uea$, hence we cannot 
obtain it from the corresponding $\omega_j' = L_j y'$, $j=1,2$.  $\omega_1'$
and $\omega_2'$ must therefore generate a proper submodule of $\uea X_k x$, 
and so must be descendants of $X_{k+1} x$.

Of course, we cannot assume from the outset that we are working in a
staggered module, because we want to apply the Projection Lemma to the
study of when staggered modules exist!  Nevertheless, the outline above 
serves to motivate the steps in the general proof below.  There are a few 
technicalities to work through, most of which arise because we must
make sure that our constructions are well-defined in the absence of $y$.  Moreover, we also have to account for the structural, and therefore notational, differences which delineate the chain and braid cases.

\begin{proof}
As already stated, there are two cases leading to three steps to consider.
The constructions are similar in all three, but because of structural variations, we must split the considerations
accordingly.  However we will only provide full details in the chain case, 
limiting ourselves to describing what is different in the braid cases.

\textbf{Chain case $\data_{m}^{(k)} \rightarrow \data_m^{(k+1)}$}:
We assume that $(w_1,w_2) \in \data_m^{(k)}$ with $k < \rnk - 1$,
so $m > \ell_{k+1}$. To find a gauge transformation $g_z \in \gauge_m$
such that $g_z(w_1,w_2) \in \data_{m}^{(k+1)}$, we will introduce a basis
 of $\uea_{m - \ell_k}$ with a certain ``orthonormality'' property.  We make 
this precise as follows.

First, let $X^{(k+1)} \in \uea^-$ be defined by $X_{k+1} = X^{(k+1)} X_k$. 
We choose a basis $\tset{V_{\lambda} X^{(k+1)} v_{h^\lft + \ell_k}}$ at grade
$m - \ell_k$ of the maximal proper submodule of the Verma module
$\Verma_{h^\lft + \ell_k}$ (whose \hws{} has conformal dimension equal to
that of $X_k x$).  Thus, $V_{\lambda} \in \uea^-_{m - \ell_{k+1}}$.  We can
complete this to a basis of $\Verma_{h^\lft + \ell_k}$ at the same grade by
adding vectors $Z_{\mu} v_{h^\lft + \ell_k}$ with
$Z_{\mu} \in \uea^-_{m - \ell_k}$.  Since the quotient of a module by its 
maximal proper submodule has non-degenerate Shapovalov form, we can even 
choose the $Z_{\mu}$ to be orthonormal:\footnote{As usual, we can always find an orthogonal basis $\tset{Z_{\mu}}$ such that
$\abs{\tinner{Z_{\mu}}{Z_{\mu}}} = 1$.  Since every complex scalar is a square, it is trivial to redefine the $Z_{\mu}$ so as to obtain an orthonormal basis.  We mention that if we had chosen the Shapovalov form to be sesquilinear rather than bilinear, then this would not be possible.}
\begin{equation} \label{eqnOrthonormality}
\inner{Z_{\mu} v_{h^\lft + \ell_k}}{Z_{\nu} v_{h^\lft + \ell_k}} = \delta_{\mu \nu}, \qquad \text{that is} \qquad
Z_{\mu}^{\dagger} Z_{\nu} v_{h^\lft + \ell_k} = \delta_{\mu \nu} v_{h^\lft + \ell_k}.
\end{equation}
This then defines a basis $\tset{V_{\lambda} X^{(k+1)}} \cup \tset{Z_{\mu}}$
of $\uea^-_{m - \ell_k}$.

Since the $Z_{\mu}$ are not scalars ($m > \ell_k$), we may write $Z_\mu^\dagger = Z_{\mu;1}^\dagger L_1 + Z_{\mu;2}^\dagger L_2$.  The choice of $Z_{\mu;1}$ and $Z_{\mu;2}$ is not unique, but if $Z_\mu^\dagger = Z_{\mu;1}'^{\, \dagger} L_1 + Z_{\mu;2}'^{\, \dagger} L_2$ is another choice then
\begin{equation}
(Z_{\mu;1}^\dagger - Z_{\mu;1}'^{\, \dagger}) L_1 = -(Z_{\mu;2}^\dagger - Z_{\mu;2}'^{\, \dagger}) L_2 \in \uea^+ L_1 \cap \uea^+ L_2.
\end{equation}
It follows that each $Z_{\mu}$ gives rise to a well-defined element $Z_{\mu;1}^\dagger w_1 + Z_{\mu;2}^\dagger w_2$ of $\brac{\uea^- X_k x}_{\ell_k}$, as $(w_1,w_2) \in \data_m^{(k)}$.  We may therefore define $\zeta_{\mu} \in \CC$ by
\begin{equation} \label{eqnDefZeta}
Z_{\mu;1}^\dagger w_1 + Z_{\mu;2}^\dagger w_2 = \zeta_{\mu} X_k x.
\end{equation}
We can similarly write $V_\lambda^\dagger = V_{\lambda;1}^\dagger L_1 + V_{\lambda;2}^\dagger L_2$, as the $V_{\lambda}$ are also not scalars ($m > \ell_{k+1}$).  However, the analogues of the $\zeta_{\mu}$ all vanish as
\begin{equation} \label{eqnAnnihilate}
\inner{X_k x}{X^{(k+1) \dagger} \bigl( V_{\lambda;1}^\dagger w_1 + V_{\lambda;2}^\dagger w_2 \bigr)}_{\uea X_k x} = \inner{X_{k+1} x}{V_{\lambda;1}^\dagger w_1 + V_{\lambda;2}^\dagger w_2}_{\uea X_k x} = 0.
\end{equation}
Here, recall that $\tinner{\cdot}{\cdot}_{\uea X_k x}$ denotes the
Shapovalov form of the submodule $\uea X_k x$.

To tune the constants $\zeta_\mu$ to zero,
we set $z = -\sum_\nu \zeta_\nu Z_\nu X_k x \in \uea X_k x$ and apply the
transformation $g_z$.
Letting $w_j' = w_j + L_j z$, for $j = 1,2$, explicit computation gives
\begin{equation}
\zeta_{\mu}' X_k x = Z_{\mu;1}^\dagger w_1' + Z_{\mu;2}^\dagger w_2' = 0,
\end{equation}
for all $\mu$.  Here we use the orthonormality 
of the $Z_{\mu}$, \eqnref{eqnOrthonormality} (which clearly continues to
hold upon projecting $\Verma_{h^\lft + \ell_k}$ onto $\uea X_k x$).  We need
now only verify that each $w_j' \in \uea X_{k+1} x$ (which
is the kernel of the Shapovalov form in the submodule $\uea X_k x$)
by showing that there
is no element of $\uea$ which takes $w_j'$ to $X_k x$.  We will detail
this for $j=1$, the case $j=2$ being entirely analogous.

Clearly, we need only consider elements $U \in \uea^+_{-m+1 + \ell_k}$.  Write $L_{-1} U^\dagger \in \uea^-_{m-\ell_k}$ in the basis defined above to get
\begin{align}
\nonumber
U L_1 & = \sum_\lambda a_\lambda \bigl( X^{(k+1)} \bigr)^\dagger V_\lambda^\dagger
    + \sum_\mu b_\mu Z_\mu^\dagger \\
\label{eq: projection gets to the kernel of Shapovalov}
& = \Bigl( \bigl( X^{(k+1)} \bigr)^\dagger \sum_\lambda a_\lambda V_{\lambda;1}^\dagger + \sum_\mu b_\mu Z_{\mu;1}^\dagger \Bigr) L_1
  + \Bigl( \bigl( X^{(k+1)} \bigr)^\dagger \sum_\lambda a_\lambda V_{\lambda;2}^\dagger + \sum_\mu b_\mu Z_{\mu;2}^\dagger \Bigr) L_2,
\end{align}
where the $a_{\lambda}$ and $b_{\mu}$ denote coefficients.  
Let $U'_1$ and $U'_2$ be the respective prefactors of $L_1$ and $L_2$ appearing
in \eqnref{eq: projection gets to the kernel of Shapovalov}.
This equation then becomes
$(U_1' - U) L_1 = - U_2' L_2 \in \uea L_1 \cap \uea L_2$, so we obtain the
equality $U w_1' = U_1' w_1' + U_2' w_2'$ as
$\brac{w_1' , w_2'} \in \data_m^{(k)}$.
But, $\bigl( X^{(k+1)} \bigr)^\dagger$ annihilates all of
$(\uea X_k x)_{\ell_{k+1}}$ (compare with \eqnref{eqnAnnihilate}), so we see
that by tuning the $\zeta_{\mu}'$ to zero, we have guaranteed that
\begin{equation}
U w_1' = U_1' w_1' + U_2' w_2' =
\sum_\mu b_\mu \bigl( Z_{\mu;1}^\dagger w_1' + Z_{\mu;2}^\dagger w_2' \bigr) = \sum_\mu b_\mu \zeta_{\mu}' X_k x = 0,
\end{equation}
by \eqnref{eqnDefZeta}.  Since this holds for all
$U \in \uea^+_{-m+1 + \ell_k}$, $w_1' \in \uea X_{k+1} x$.  After repeating
this argument for $w_2'$, we have completed the proof: 
$g_z(w_1 , w_2) = (w_1' , w_2') \in \data^{(k+1)}_m$.

{
\psfrag{0}[][]{$\scriptstyle{0}$}
\psfrag{1}[][]{$\scriptstyle{1}$}
\psfrag{2}[][]{$\scriptstyle{2}$}
\psfrag{k}[][]{$\scriptstyle{k}$}
\psfrag{k+1}[][]{$\scriptstyle{k+1}$}
\psfrag{chnX}[][]{} 
\psfrag{chnproj}[][]{$\data^{(k)} \rightarrow \data^{(k+1)}$}
\psfrag{brdproj1}[][]{$\data^{(k;-)} \rightarrow \data^{(k;+)}$}
\psfrag{brdproj2}[][]{$\data^{(k;+)} \rightarrow \data^{(k+1;-)}$}
\psfrag{brdXm}[][]{} 
\psfrag{brdXp}[][]{} 
\begin{figure}
\begin{center}
\includegraphics[width=12cm]{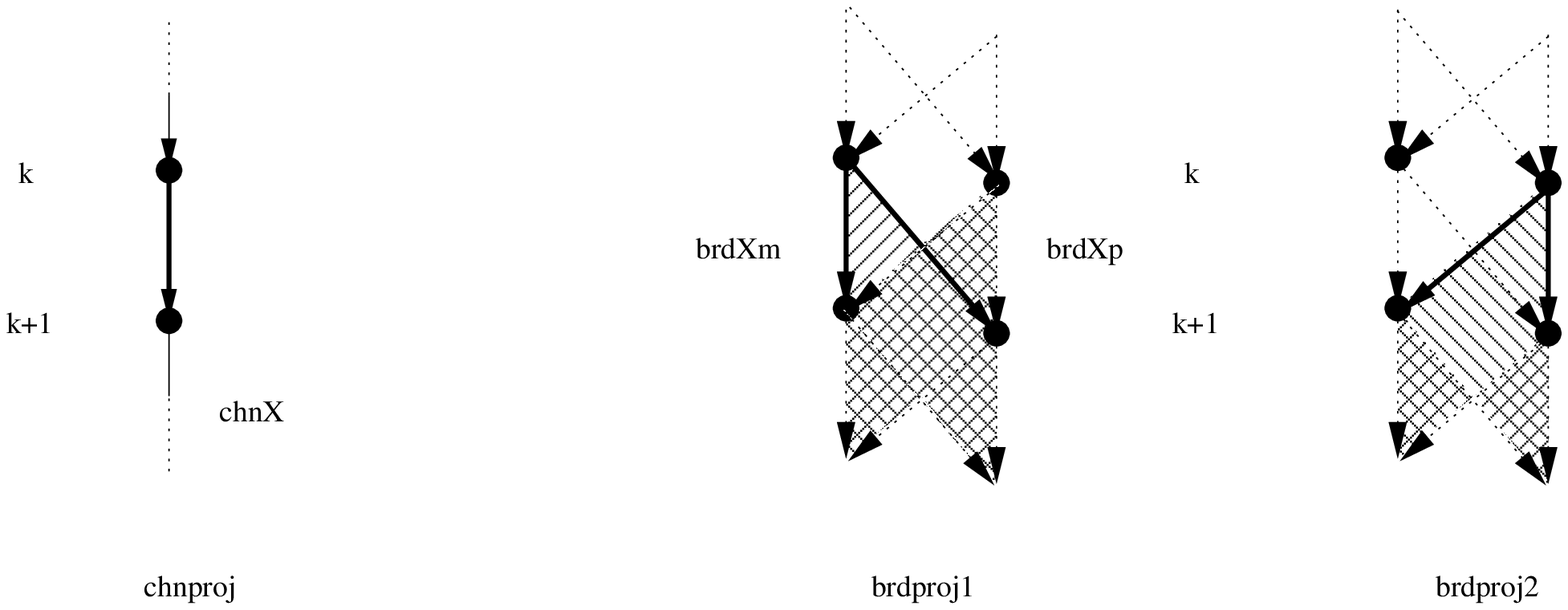}
\caption{An illustration of the projections constructed in the proof of \lemref{lemProjection}. On the left we portray the chain
case, in which the projection involves taking $w'_j$ from the module $\uea X_k x$ (itself a submodule of $\LMod$) to its (maximal) submodule $\uea X_{k+1} x$.
On the right are the braid cases. We alternate
between steps of two types, going from the module
$\uea X_{k}^- x + \uea X_{k}^+ x$ to its submodule $\uea X_{k}^+ x$ (left), and from the module $\uea X_{k}^+ x$ to its submodule $\uea X_{k+1}^- x + \uea X_{k+1}^+ x$ (right). The shading indicates schematically the module we start from and the submodule we arrive at, and the emphasised arrows indicate the singular elements $X^{(k+1)}$ and $X^{(k+1 ; \pm)}$ which are used in the proof.} \label{fig: projections}
\end{center}
\end{figure}
}
\textbf{Braid case $\data_m^{(k;-)} \rightarrow \data_m^{(k;+)}$:}
Suppose that $(w_1,w_2) \in \data_m^{(k;-)}$ and $k < \rnk - 1$, so $m > \ell_{k+1}^+$.  Define $X^{(k+1;\pm)} \in \uea^-$ by $X_{k+1}^\pm = X^{(k+1;\pm)} X_k^-$.  We choose a basis, $\tset{V_{\lambda^-}^- X^{(k+1;-)} v_{\LDim + \ell_k^-}} \cup \tset{V_{\lambda^+}^+ X^{(k+1;+)} v_{\LDim + \ell_k^-}}$ say, of the maximal proper submodule of $\VerMod_{\LDim + \ell_k^-}$ at grade $m - \ell_k^-$, and extend it to a basis of $\VerMod_{\LDim + \ell_k^-}$ itself, at the same grade, by adding orthonormal elements $Z_{\mu} v_{\LDim + \ell_k^-}$.  This defines our basis of $\uea^-_{m - \ell_k^-}$ as in the chain case.

Again, $Z_\mu^\dagger = Z_{\mu;1}^\dagger L_1 +  Z_{\mu;2}^\dagger L_2$
defines constants $\zeta_\mu$ by $Z_{\mu;1}^\dagger w_1 + Z_{\mu;2}^\dagger w_2
= \zeta_\mu X_k^- x$, and we use these to define $z =
- \sum_\mu \zeta_\mu Z_\mu X_k^- x$ and
$(w_1',w_2') = g_z(w_1, w_2) \in \data_m^{(k;-)}$.
The check that $U w_j' = 0$ for any $U \in \uea^+_{-m+j+\ell_k^-}$
is done by writing $L_{-j} U^\dagger \in \uea^-_{m-\ell_k^-}$
in the above basis: We thereby obtain the analogue of
\eqnref{eq: projection gets to the kernel of Shapovalov} (but with
separate terms for the $X^{(k+1;+)}$ and $X^{(k+1;-)}$ contributions). 
This leads to $U w'_j = 0$ for all $U$ as in the chain case.
However, from this we are only able to conclude that
$w_j' \in \uea X_k^+ x$, not that $w_j'$ belongs to the maximal proper
submodule $\uea X_{k+1}^- x + \uea X_{k+1}^+ x$ of $\uea X_k^- x$ (for this, we need the last case below).  We therefore have $(w'_1,w'_2) \in \data_m^{(k;+)}$.

\textbf{Braid case $\data_m^{(k;+)} \rightarrow \data_m^{(k+1;-)}$:}
In this final case we suppose that $(w_1,w_2) \in \data_m^{(k;+)}$ and again
$k < \rnk - 1$, to guarantee that $m > \ell_{k+1}^+$.  We choose a basis
of $\uea^-_{m - \ell_k^+}$ as in the first braid case, and use this to
construct $z$ so that $g_z(w_1,w_2)$ is in $\data^{(k+1;-)}$.  Everything
now works as in the previous cases.  We only mention that proving
$U w'_j = 0$ for all $U \in \uea^+_{-m+j+\ell_k^+}$ here lets us conclude
that the $w_j'$ belong to the maximal proper submodule
$\uea X_{k+1}^- x + \uea X_{k+1}^+ x$ because we have been working entirely in $\uea X_k^+ x$.  Thus, $(w'_1,w'_2) \in \data_m^{(k+1;-)}$ as required.
\end{proof}

We conclude this section with two small remarks pertaining to this proof.  First, we call this result the Projection Lemma because each subsequent gauge transformation can be thought of as projecting the $\brac{w_1 , w_2}$ onto the next-smallest subspace in the filtration.  Indeed, if $\brac{w_1 , w_2}$ is already in the next-smallest subspace, then the $\zeta_{\mu}$ defined in the proof must already vanish, hence $z=0$ and $g_z$ is the identity map.

The second remark addresses why the sequence of projections defined in the proof terminates. Once in the submodule corresponding to the rank $k$ \sv{}(s) $\data_m^{(k)}$ ($\data_m^{(k;\pm)}$), we were
able to project further provided that $m > \ell_{k+1}$ ($m > \ell_{k+1}^{\pm}$).  This guaranteed that the $V$-type basis elements of the maximal proper submodule of $\uea^-_{m - \ell_k}$ ($\uea^-_{m - \ell_k^{\pm}}$) were not scalars, and so could be written as a sum of terms with $L_1$ or $L_2$ on the right.  As soon as $k = \rnk-1$, we find that some $V$-type basis elements are scalars, and so cannot be written in this form.  The proof then breaks down at the point of \eqnref{eq: projection gets to the kernel of Shapovalov} and its analogues.

And so it should:  In the chain case with $m = \ell_{\rnk}$, the grade of the
$w_j$ would be $\ell_{\rnk} - j$, so it is completely unreasonable to expect
that we can construct $w_j'$ belonging to $\uea X_{\rnk} x$.  In the braid
case, we get the same conclusion if $m = \ell_{\rnk}^-$.  When
$m = \ell_{\rnk}^+$, one might hope to be able to find $w_j'$ belonging to 
$\uea X_{\rnk}^- x$.  However, it is possible to show (using
\propref{prop: monotonicity right} and \thmref{thm: invariants} below for
example)
that this is only possible in a rather trivial case:  Essentially,
the ``data'' $(w_1 , w_2)$ must be equivalent to $(0,0)$.

\section{Construction when the Right Module is Verma} \label{sec: right Verma}

Throughout this section we assume that $\sH^\rgt = \Verma_{h^\rgt}$.
In particular, this means that in the construction of \secref{sec: construction}, the submodule $\sN$ of $\LMod \oplus \uea$ is generated by $(\omega_0, h^\rgt-L_0)$, $(\omega_1,-L_1)$ and $(\omega_2,-L_2)$ (there is no $\varpi$ or $\overline{X}$).  The corresponding exact sequence is
\begin{equation} \label{eq: exact seq Verma}
\dses{\Hlft}{}{\Stagg}{}{\Verma_{h^\rgt}}.
\end{equation}

In principle, we have everything we need for our attack on the question of
existence of staggered modules $\Stagg$ with exact sequence
(\ref{eq: exact seq Verma}).  However, the proofs which follow are necessarily
rather technical, given that they apply to completely general left modules.  
We will therefore first briefly outline the main ideas behind them.  We
also suggest that the reader might like to keep in mind the simplest case
in which $\omega_0$ is the \sv{} of minimal (positive) grade in $\LMod$.  
This case not only avoids the most troublesome technicalities (for example, 
we do not need the Projection Lemma for this case), but it also has the 
advantage of covering the majority of staggered modules which have thus 
far found physical application.\footnote{Actually, the physically relevant 
modules we have in mind here do not always have right module Verma.  
However, \propref{prop: monotonicity right} suggests that the relevant 
modules with non-Verma right modules should be recovered from this case 
as quotients.  We will turn to this in \secref{sec: general case}.}

Our overall plan is straight-forward.  The analysis of the $\RMod$ Verma
case turns out to afford an important simplification, namely that the
admissibility of the data is completely captured by the set $\data$,
defined in \eqnref{eqnDefData}.  This allows us to identify the set of 
isomorphism classes of staggered modules with exact sequence
(\ref{eq: exact seq Verma}) as the vector space $\data / \gauge$, thereby 
settling the existence question when $\ell = 0$ (\thmref{thmClassRVermaL=0}). 
We then turn to the computation of the dimension of the space
$\data / \gauge$.  First, we use the Projection Lemma to reduce this to the 
dimension of an equivalent space $\data' / \gauge'$, where 
$\data' \subseteq \data$ is significantly smaller in general
(\propref{propRVerIsoClasses'}).  This allows us to separate the 
computation into four cases, according to the \sv{} structure of $\LMod$ 
around $\omega_0$.  In each case, we reformulate the definition of 
$\data'$ so as to realise it as an intersection of kernels of certain 
linear functionals (\thmref{thm: right Verma characterization of data}).  
The computation of the dimension of $\data'$ is then just an exercise in
linear algebra, albeit a rather non-trivial one.  The results of this 
computation are given in \thmref{thm: moduli space right Verma}.  
Finally, we discuss generalisations of the beta-invariant of 
\eqnref{eqnDefBeta} which reduce the identification of a staggered module
with exact sequence (\ref{eq: exact seq Verma}) to the computation of at 
most two numbers.

\subsection{Admissibility} \label{secRVermaAdmiss}

In this section, we study the question of admissibility of data $\brac{\omega_1 , \omega_2}$ under the hypothesis that the right module is Verma.  The result is reported in \propref{propAdmissible} below.  First however, we need a simple but very useful lemma.  Recall that the submodule $\sN^{\circ}$ may be naturally viewed as a submodule of $\LMod$.

\begin{lemma} \label{lemUseful}
When $\Hrgt$ is Verma, $u \in \sN^{\circ}$ if and only if there exist
$U_1 , U_2 \in \uea$ such that
\begin{equation}
U_1 \omega_1 + U_2 \omega_2 = u \qquad \text{and} \qquad U_1 L_1 + U_2 L_2 = 0.
\end{equation}
\end{lemma}
\begin{proof}
By definition, $u \in \sN^{\circ}$ if and only if there exist $U_0 , U_1 , U_2 \in \uea$ such that
\begin{subequations}
\begin{align}
U_0 \omega_0 + U_1 \omega_1 + U_2 \omega_2 &= u \qquad \text{(in $\LMod$)} \label{eqnLMod} \\
\text{and} \qquad U_0 \brac{L_0 - \RDim} + U_1 L_1 + U_2 L_2 &= 0 \qquad \text{(in $\uea$)}, \label{eqnUEA}
\end{align}
\end{subequations}
so one direction is trivial.  What we need to show is that we may take $U_0 = 0$, without loss of generality.  Note that by taking $u \in \LMod$ homogeneous, we may assume that $U_0$, $U_1$ and $U_2$ are homogeneous in $\uea$.

Consider $U_1 L_1 + U_2 L_2$.  \PBW-ordering this combination will give a variety of terms, each of which must have a positive index on the rightmost mode.  If \PBW-ordering $U_0$ produced any term which did not have a positive index on the rightmost mode, then right-multiplying by $\brac{L_0 - \RDim}$ would preserve the ordering, and so this term could not be cancelled by any (ordered) term of $U_1 L_1 + U_2 L_2$.  This contradicts (\ref{eqnUEA}), so all the ordered terms of $U_0$ must have a positive index on the rightmost mode.  Then, $U_0 \omega_0 = 0$, and (\ref{eqnLMod}) has the desired form.

But if every \PBW-ordered term of $U_0$ has a positive index on the rightmost mode, we may write $U_0 = U_1' L_1 + U_2' L_2$ for some $U_1' , U_2' \in \uea$.  Hence (for $U_0 \in \uea_m$),
\begin{align}
U_0 \brac{L_0 - \RDim} + U_1 L_1 + U_2 L_2 &= \brac{L_0 - \RDim - m} U_0 + U_1 L_1 + U_2 L_2 \notag \\
&= \brac{U_1 + \brac{L_0 - \RDim - m} U_1'} L_1 + \brac{U_2 + \brac{L_0 - \RDim - m} U_2'} L_2 = 0,
\end{align}
and a simple redefinition of $U_1$ and $U_2$ will put (\ref{eqnUEA}) in the required form.  This redefinition would affect (\ref{eqnLMod}), but for the fact that
\begin{equation}
\brac{L_0 - \RDim - m} \brac{U_1' \omega_1 + U_2' \omega_2} = 0,
\end{equation}
as $U_1' \omega_1 + U_2' \omega_2$ is an $L_0$-eigenvector of eigenvalue $\RDim + m$.
\end{proof}

Recall that \lemref{lem: consistency of L1 L2} gave a necessary
condition for $(\omega_1 , \omega_2)$ to be data of a
staggered module.  \thmref{thmConstruction} and \lemref{lemUseful} now tell
us that under the hypothesis that $\RMod$ is Verma, this condition is also
sufficient:  $\sN^{\circ} = \set{0}$ if and only if
\begin{equation}
U_1 \omega_1 + U_2 \omega_2 = 0 \qquad \text{for all} \qquad U = U_1 L_1 = -U_2 L_2 \in \uea L_1 \cap \uea L_2.
\end{equation}
In the language of \secref{sec: construction} (see \eqnref{eqnDefData1} in particular), this becomes:

\begin{proposition} \label{propAdmissible}
When $\RMod$ is Verma, $\brac{\omega_1 , \omega_2}$ is admissible if and only
if $\brac{\omega_1 , \omega_2} \in \data$.
\end{proposition}

\noindent \exref{ex:0,12,14} shows that this hypothesis is not superfluous. 
Combining this result with \propref{prop: equivalence} now gives the
following important characterisation.

\begin{proposition} \label{propRVerIsoClasses}
The space of (isomorphism classes of) staggered modules with exact sequence
(\ref{eq: exact seq Verma}) may be identified with the vector space
$\data / \gauge$.
\end{proposition}

\begin{example} \label{ex:-2,13,15,13,33}
At $c=-2$ ($t = 2$), one can use the algorithm detailed in \cite{GabInd96}
to fuse $\IrrMod_{-1/8}$ with $\VerMod_{3/8}$ and $\IrrMod_1$ with $\VerMod_0$.
In both cases, a staggered module is obtained with the short exact sequence
\begin{equation}
\dses{\VerMod_0}{}{\StagMod}{}{\VerMod_1}.
\end{equation}
The respective beta-invariants turn out to be $\beta = -1$ (as in
\exref{ex:-2,13,15}) and $\beta = \tfrac{1}{2}$.  This exact sequence
therefore admits two distinct staggered modules, hence by 
\propref{propRVerIsoClasses}, there is (at least) a one-parameter family
of such modules.

This example highlights in a novel way the physical importance of a
good theory of staggered modules. It shows concretely how physically relevant
constructions (here fusion products) can result in modules that
cannot be distinguished from each other by their characters
(graded dimensions), or even by the action of $L_0$ alone.
\end{example}

Finally, since $\ell = 0$ implies that $\omega_1 = \omega_2 = 0$, we thereby
obtain the first piece of our classification, the case when
$\rank \omega_0 = 0$.\footnote{Although formulated differently and obtained
by slightly different means, the result in this case has already appeared in
\cite{RohRed96}. In fact the result is
obtained there (and could have been obtained here) without the lengthy
preparation that our more general results require.}
For consistency with \secref{sec: allowed parameters} below, we will refer
to this case as case (0).

\begin{theorem}[Case (0) of the classification] \label{thmClassRVermaL=0}
Given left and right modules $\LMod$ and $\RMod$, for which the latter is Verma and $\LDim = \RDim$, there exists a \emph{unique} staggered module $\StagMod$ with short exact sequence (\ref{eqnDefStag}).
\end{theorem}

\noindent We remark that it should not be surprising that the precise form
of $\LMod$ plays no r\^{o}le in this result.  For existence when $\LMod$ is
also Verma implies existence for general $\LMod$ (subject only to the
non-vanishing of $\omega_0$) by \propref{prop: monotonicity left}.

\subsection{Choosing Data} \label{sec: data from a subspace}

We have determined that the space of (isomorphism classes of) staggered
modules with exact sequence (\ref{eq: exact seq Verma}) is naturally
realised as the quotient of $\data$ under the action of $\gauge$, by
\propref{propRVerIsoClasses}.  These spaces are a little large in general,
so it proves convenient to prune them into something a little more
manageable.  This will be achieved by applying the Projection Lemma
(\lemref{lemProjection}).

Let us denote by $\minsubmod$ the submodule of $\Hlft$ generated by the 
\svs{} whose rank is one less than that of $\omega_0$.  For
example,\footnote{The case $\rank \omega_0 = 0$ (that is, $\ell = 0$) has
already been analysed in \thmref{thmClassRVermaL=0}, but would be consistent
with $\minsubmod = \set{0}$.} if $\rank \omega_0 = 1$ we have
$\minsubmod = \Hlft$.  When $\rank \omega_0 > 1$, $\minsubmod$ is generated
by one or two \svs{} according as to whether $\LMod$ is of chain or braid
type (this follows from $\omega_0 \neq 0$).  We now define our ``pruned'' 
space of admissible data to be
\begin{equation}
\data' = \bigl\{ (\omega_1, \omega_2) \in \data \st \omega_1 , \omega_2 \in \minsubmod \bigr\}.
\end{equation}
The Projection Lemma with $m = \ell$ (so $\rnk = \rank \omega_0$)
immediately gives:

\begin{lemma} \label{lem: data from minsubmod}
For any $(\omega_1, \omega_2) \in \data$, there exists $g_u \in G$ such that
$g_u (\omega_1, \omega_2) = (\omega_1', \omega_2') \in \data'$.
\end{lemma}

\noindent The proof only requires realising that in this application, the
subspace $\data_m^{(\rnk-1)}$ or $\data_m^{(\rnk-1;-)}$ appearing in the
Projection Lemma is precisely $\data'$.

The new choice of data $(\omega_1', \omega_2')$ is equivalent to the old data
$(\omega_1, \omega_2)$, so the underlying staggered module remains unchanged.
Of course, we still have some freedom in the choice.  There is a residual set
of gauge transformations, namely
$\gauge' = \set{g_u \in \gauge : u \in \minsubmod_{\ell}} \subseteq \gauge$,
which preserves $\data'$.
Analogous to the case of the full $\gauge$
(\secref{secStag}), we have $\gauge' \isom \minsubmod_{\ell} / \CC \omega_0$
(as vector spaces), hence
\begin{equation} \label{eqnDimGauge'}
\dim \ \gauge' = \dim \ \minsubmod_{\ell} - 1.
\end{equation}
Moreover, \propref{propRVerIsoClasses} can now be 
replaced by

\begin{proposition} \label{propRVerIsoClasses'}
The space of (isomorphism classes of) staggered modules with exact
sequence (\ref{eq: exact seq Verma}) may be identified with the vector
space $\data' / \gauge'$.
\end{proposition}

We point out that $\omega_0$ need not be the \sv{} of lowest grade in
$\minsubmod$ (excluding of course the obvious generating ones). In the braid
case when $\omega_0 = X x = X_\rnkO^+ x$ (with $\rnkO = \rank \omega_0$),
$X_\rnkO^- x$ may be a non-zero \sv{}.
Then, $X_\rnkO^- x \in \minsubmod$ has the same rank as $\omega_0$,
but its grade is strictly less than that of $\omega_0$.  
This case is the source of the most trouble in the following analysis.

\subsection{Characterising Admissible Data} \label{sec: allowed parameters}

In this section we give a tractable characterisation of the admissibility of pairs $(\omega_1, \omega_2) \in \minsubmod_{\ell-1} \oplus \minsubmod_{\ell-2}$.
As the case $\rank \omega_0 = 0$ (that is, $\ell=0$) has already been settled,
we will assume that $\rank \omega_0 \equiv \rnkO \geqslant 1$ for the rest of the section.

We will separate this characterisation into four cases, according to the
number of generating \svs{} of $\minsubmod$ and whether there is a
non-generating (non-zero) \sv{} in $\minsubmod$ whose grade is less than
$\ell$ (the troublesome cases).  Explicitly, the cases are
\begin{description}
\item[Case (1)] $\minsubmod$ is generated by a single \sv{} and this is
the only \sv{} in $\minsubmod$ of grade less than $\ell$.  This applies
in two situations:  When $\LMod$ is of chain (or link) type, and when
$\LMod$ is of braid type with either $\omega_0 = X_1^- x$, or
$\omega_0 = X_1^+ x$ and $X_1^- x = 0$.
\item[Case (1')] $\minsubmod$ is generated by a single \sv{} and there is
another \sv{} in $\minsubmod$ of grade less than $\ell$.  This only applies
when $\LMod$ is of braid type with $\omega_0 = X_1^+ x$ and $X_1^- x \neq 0$.
\item[Case (2)] $\minsubmod$ is generated by two distinct \svs{} and these
are the only \svs{} in $\minsubmod$ of grades less than $\ell$.  This only
applies when $\LMod$ is of braid type with either\footnote{We mention that
this also covers the possibility that $\omega_0$ is the \sv{} of
\emph{maximal} grade in a braid type Verma module with $t<0$
(\secref{secNotation}).} $\omega_0 = X_{\rnkO}^- x$, or
$\omega_0 = X_{\rnkO}^+ x$ and $X_{\rnkO}^- x = 0$.
\item[Case (2')] $\minsubmod$ is generated by two distinct \svs{} and
there is another \sv{} in $\minsubmod$ of grade less than $\ell$.  This
only applies when $\LMod$ is of braid type with $\omega_0 = X_{\rnkO}^+ x$
and $X_{\rnkO}^- x \neq 0$.
\end{description}
It is easy to verify that any possibility is covered by exactly one of these
cases. We illustrate them for convenience in \figref{fig: cases}.

{
\psfrag{caseone}[][]{Case (1)}
\psfrag{casepone}[][]{Case (1')}
\psfrag{casetwo}[][]{Case (2)}
\psfrag{caseptwo}[][]{Case (2')}
\psfrag{om0}[][]{$\omega_0$}
\begin{figure}
\begin{center}
\includegraphics[width=12cm]{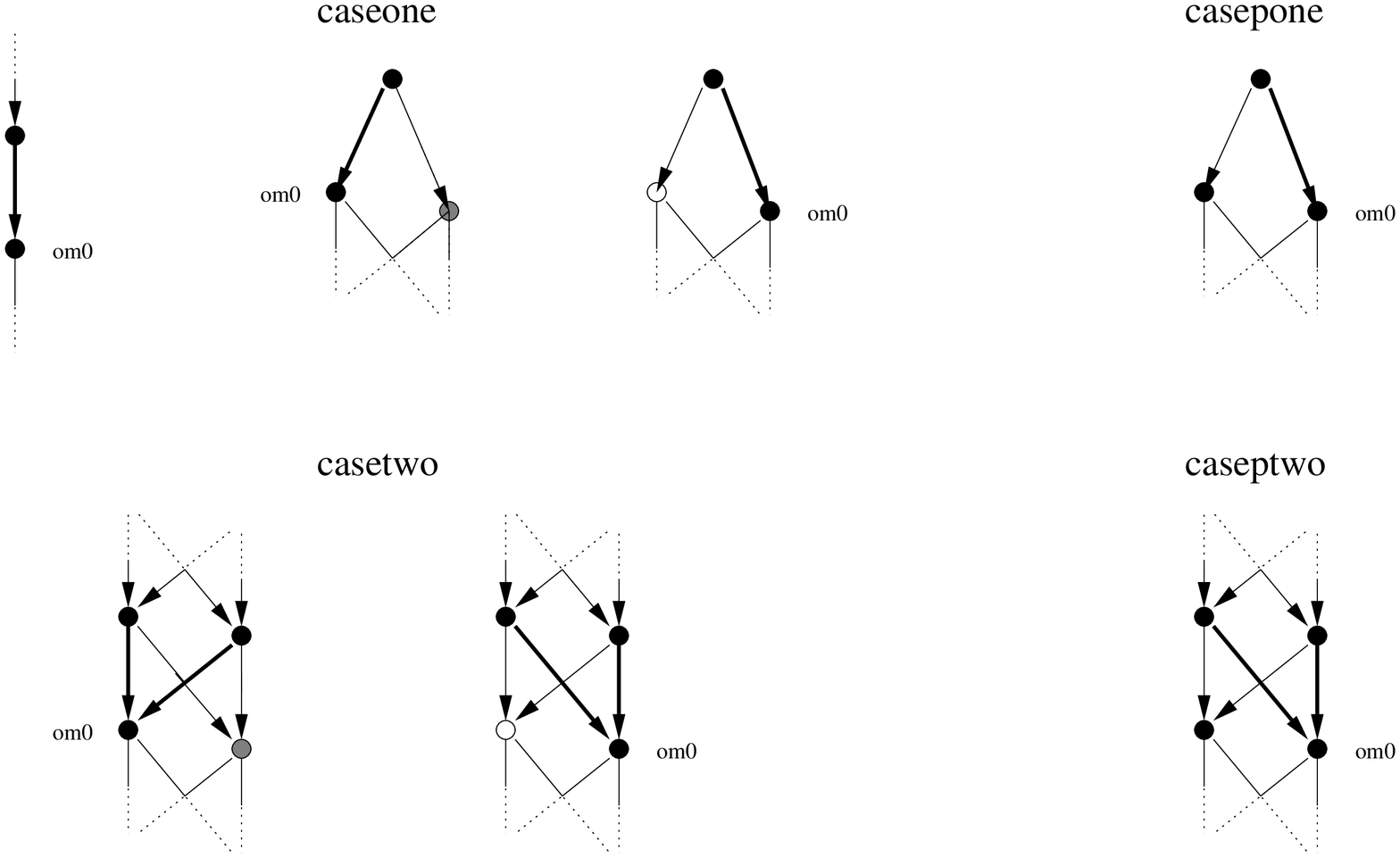}
\caption{An illustration of the possible structures of the left module $\LMod$ in cases (1), (1'), (2) and (2').  As with earlier figures, a black circle represents a \sv{} of $\LMod$, whereas a white circle indicates a \sv{} of the corresponding Verma module which has been set to zero.  We use a grey circle when it does not matter if the \sv{} has been set to zero or not.
Note that the picture corresponding to case (1) with $t \notin \bQ$ has been omitted --- it is understood as a subcase of the chain case pictured.  Similarly, the degenerate braid case ($t \in \bQ$, $t<0$) has not been explicitly portrayed --- it is regarded as a subcase of case (2).} \label{fig: cases}
\end{center}
\end{figure}
}

To analyse each of these cases further, it is useful to first sharpen the
conclusions of \lemref{lemUseful} somewhat.  Specifically, we show that taking
$u$ to be a \sv{} of ``minimal rank'' allows us to choose $U_1$ and $U_2$ in
$\uea^+$.

\begin{lemma} \label{lem: sufficient}
If $(\omega_1, \omega_2) \in \Hlft_{\ell-1} \oplus \Hlft_{\ell-2}$ is not
admissible (and $\Hrgt$ is Verma), then $\sN^\circ$ contains \svs{} of
$\Hlft$ of grade less than $\ell$.  For a \sv{} $x'$ whose rank is
\emph{minimal} among those in $\sN^{\circ}$, there then exist
$U_1 , U_2 \in \uea^+$ such that
\begin{equation} \label{eqnDesired}
U_1 \omega_1 + U_2 \omega_2 = x' \qquad \text{and} \qquad
U_1 L_1 + U_2 L_2 = 0.
\end{equation}
\end{lemma}

\noindent We have stated only one direction, but the converse is already implied
by \lemref{lemUseful}.

\begin{proof}
Suppose that $(\omega_1, \omega_2)$ is not admissible, which in view of 
\thmref{thmConstruction}, means that $\sN^\circ$ is a non-zero submodule
of $\Hlft$.  Therefore $\sN^\circ$ contains non-zero \svs{}, and it is
generated by its minimal rank \svs{}.  Take $x'$ to be one such generator.

By \lemref{lemUseful}, we can find $U_1 , U_2 \in \uea$ such that both equations in (\ref{eqnDesired}) are satisfied.  But, if we \PBW{}-order $U_1$ and $U_2$, we see that terms with negative modes on the left cannot contribute to $U_1 \omega_1 + U_2 \omega_2$ by the assumption that $x'$ was of minimal rank.  We therefore drop them.  Furthermore, any $L_0$ on the left may be replaced by the appropriate eigenvalue, so we may assume that $U_1,U_2 \in \uea^+$ in the first equation.  Linear independence of \PBW{}-monomials then allows us to likewise drop the terms with negative modes in the second equation.  We may therefore write $U_1 L_1 + U_2 L_2 = \sum_n L_0^n U^{\brac{n}} = 0$ with $U^{\brac{n}} \in \uea^+$.  Independence and the lack of zero-divisors in $\uea$ then means that each $U^{\brac{n}}$ must vanish separately, so we can certainly replace each $L_0$ by its eigenvalue here too.  This means that the $U_1,U_2 \in \uea^+$ of the first equation also satisfy the second.  Finally, we conclude from $U_1 , U_2 \in \uea^+$ in the first equation that the grade of $x'$ must be less than $\ell$.
\end{proof}

Assuming that
$(\omega_1, \omega_2) \in \minsubmod_{\ell-1} \oplus \minsubmod_{\ell-2}$,
the submodule $\sN^\circ$ is contained in $\minsubmod$ by \lemref{lemUseful}.
The minimal rank referred to in \lemref{lem: sufficient} is then either
$\rnkO-1$ or $\rnkO$.
In concrete terms, we need to check whether the rank $\rnkO-1$
\svs{} are in $\sN^\circ$, and if this can be ruled out, we do the
same for the rank $\rnkO$ \sv{} of grade less than $\ell$ if necessary
(cases (1') and (2') only).
Below, we introduce functionals $\psi$, $\psi^\pm$
and $\psi^\cap$ with the aim of reducing these checks to a problem in linear algebra.  We first separate our considerations according to the number of rank $\rnkO-1$ \svs{} in $\Hlft$, and then analyse the further constraints stemming from the presence of a second rank $\rnkO$ singular vector.

\subsubsection{Cases \textup{(1)} and \textup{(1')}:}
\label{sec: parameters case one}

In these cases, $\minsubmod$ is generated by the normalised \sv{}
$X_{\rnkO -1} x$ of grade $\ell_{\rnkO -1}$.  Making use of the fact that
$\minsubmod_{\ell_{\rnkO -1}}$ is one-dimensional, we define for each
$U = U_1 L_1 = -U_2 L_2 \in
\brac{\uea^+ L_1 \cap \uea^+ L_2}_{\ell_{\rnkO -1} - \ell}$, a linear functional
\begin{equation} \label{eqnDefPsi}
\psi_U \colon \minsubmod_{\ell-1} \oplus \minsubmod_{\ell-2} \rightarrow
\CC \qquad \text{by} \qquad U_1 \omega_1 + U_2 \omega_2 =
\func{\psi_U}{\omega_1, \omega_2} X_{\rnkO -1} x.
\end{equation}
Taking
$(\omega_1, \omega_2) \in \minsubmod_{\ell-1} \oplus \minsubmod_{\ell-2}$, 
the submodule $\sN^\circ$ contains no \svs{} of rank less than $\rnkO-1$.
In view of \lemref{lem: sufficient}, $\sN^{\circ}$ contains the rank
$\rnkO-1$ singular vector if and only if
$\func{\psi_U}{\omega_1, \omega_2} \neq 0$ for some
$U \in (\uea^+ L_1 \cap \uea^+ L_2)_{\ell_{\rnkO -1} - \ell}$.
We formulate this result as follows:

\begin{proposition} \label{prop: right Verma 1}
In cases \textup{(1)} and \textup{(1')}, assuming $(\omega_1, \omega_2) \in \minsubmod_{\ell-1} \oplus \minsubmod_{\ell-2}$, we have $X_{\rnkO -1} x \notin \sN^\circ$ if and only if
\begin{equation}
(\omega_1, \omega_2) \in \bigcap_{U \in (\uea^+ L_1 \cap \uea^+ L_2)_{\ell_{\rnkO -1} - \ell}} \Kern \psi_U.
\end{equation}
\end{proposition}

\noindent We point out that in case (1), $X_{\rnkO-1} x$ is the only \sv{} in
$\minsubmod$ of grade less than $\ell$, so by \lemref{lem: sufficient}, 
the above condition completely characterises the admissible data
$(\omega_1, \omega_2) \in \data'$.  In case (1'), there is another such
\sv{}, and so we will have to work harder to get a complete characterisation
(\secref{sec: pain in the A}). This proposition is of course crucial for
case (1') as well, since it tells us how to rule out the rank
$\rnkO-1$ \svs{}.  After that, $\rnkO$ becomes the candidate for the minimal rank referred to in \lemref{lem: sufficient}.

\subsubsection{Cases \textup{(2)} and \textup{(2')}:}
\label{sec: parameters case two}

In this case there are two rank $\rnkO-1$ highest weight vectors in $\Hlft$,
namely $X_{\rnkO-1}^{\pm} x$, and the submodule
$\minsubmod = \uea X_{\rnkO-1}^{-} x + \uea X_{\rnkO-1}^{+} x$
is not a highest weight module. We have
\begin{equation}
\minsubmod_{\ell-j} = \uea^-_{\ell-\ell_{\rnkO-1}^{-}-j} X_{\rnkO-1}^{-} x +
\uea^-_{\ell-\ell_{\rnkO-1}^{+}-j} X_{\rnkO-1}^{+} x \qquad \text{for $j=1,2$,}
\end{equation}
where the sum is direct in case (2), but not in case (2').
In either case, given
$(\omega_1, \omega_2) \in \minsubmod_{\ell-1} \oplus \minsubmod_{\ell-2}$,
we can write $\omega_j = \omega_j^- + \omega_j^+$ with
$\omega_j^\pm \in \uea X_{\rnkO-1}^{\pm} x$. The two conditions we
will obtain below can be understood as one for each part, ``$-$'' and ``$+$''.

In analogy with the $\psi_U$ above, we define the functionals
$\psi^\pm_{U^\pm} : \minsubmod_{\ell-1} \oplus \minsubmod_{\ell-2}
\rightarrow \bC$ by the formulae
\begin{subequations} \label{eqnDefPsiPM}
\begin{align}
U^-_1 \omega_1 + U^-_2 \omega_2 &= \psi^-_{U^-}(\omega_1, \omega_2) X_{\rnkO-1}^{-} x \\
\text{and} \qquad U^+_1 \omega_1 + U^+_2 \omega_2 &= \psi^+_{U^+}(\omega_1, \omega_2) X_{\rnkO-1}^{+} x \pmod{\uea X_{\rnkO-1}^{-} x},
\end{align}
\end{subequations}
where $U^\pm = U_1^\pm L_1 = -U_2^\pm L_2 \in
(\uea^+ L_1 \cap \uea^+ L_2)_{\ell_{\rnkO-1}^{\pm}-\ell}$.  These definitions again rely on the fact that both $\minsubmod_{\ell_{\rnkO-1}^{-}}$ and
$(\minsubmod / \uea X_{\rnkO-1}^{-} x)_{\ell_{\rnkO-1}^{+}}$ are
one-dimensional.

Assuming $(\omega_1,\omega_2) \in \minsubmod_{\ell-1} \oplus \minsubmod_{\ell-2}$, so there can again
be no highest weight vectors of rank less than $\rnkO-1$ in $\sN^\circ$,
\lemref{lem: sufficient} tells us under which condition the singulars
$X_{\rnkO-1}^{\pm} x$ are in $\sN^\circ$. Precisely as above,
$X_{\rnkO-1}^- x \in \sN^\circ$ if and only if there is a $U^-$ such that
$\psi^-_{U^-} (\omega_1, \omega_2) \neq 0$. The case of $X_{\rnkO-1}^+ x$
works out similarly, despite the slightly
more involved definition of $\psi^+$. The easy direction is given by
\lemref{lem: sufficient}: If $X_{\rnkO-1}^+ x \in \sN^\circ$, then there exists
$U^+=U^+_1 L_1 = -U^+_2 L_2$ such that
$\psi^+_{U^+}(\omega_1, \omega_2) = 1$. To see the converse, assume that there
exists $U^+$ such that $\psi^+_{U^+} (\omega_1,\omega_2) \neq 0$
and without loss of generality choose $U^+$ so that this value is unity.
Explicitly, this means that
\begin{equation}
U_1^+ \omega_1 + U_2^+ \omega_2  = X_{\rnkO-1}^+ x + u \qquad \text{for some $u \in \uea X_{\rnkO-1}^- x$.}
\end{equation}
If $u=0$, we are done, so assume that $u = V^- X_{\rnkO-1}^- x \neq 0$
with $V^- \in \uea^-$. The maximal proper submodule of $\uea X_{\rnkO-1}^- x$ is trivial at the grade of $u$, so there must exist $V^+ \in \uea^+$ such that $V^+ u = X_{\rnkO-1}^- x$.  As $u = V^- V^+ u$, it now follows that
\begin{equation}
\brac{\wun - V^- V^+} \brac{U_1^+ \omega_1 + U_2^+ \omega_2} = \brac{\wun - V^- V^+} \brac{X_{\rnkO-1}^+ x + u} = X_{\rnkO-1}^+ x.
\end{equation}
Applying \lemref{lemUseful} to $\brac{\wun - V^- V^+} U_j^+$, we conclude that $X_{\rnkO-1}^+ x \in \sN^\circ$.

In conclusion, $X_{\rnkO-1}^\pm x \in \sN^\circ$ if and only if for some $U^\pm$
the value of $\psi^\pm_{U^\pm}(\omega_1,\omega_2)$ is non-zero.
This gives the analogue of \propref{prop: right Verma 1}:

\begin{proposition} \label{prop: right Verma 2}
In cases \textup{(2)} and \textup{(2')}, assuming $(\omega_1, \omega_2) \in \minsubmod_{\ell-1} \oplus \minsubmod_{\ell-2}$, we have $X_{\rnkO -1}^\pm x \notin \sN^\circ$ if and only if
\begin{equation}
(\omega_1, \omega_2) \in \bigcap_{U^\pm \in (\uea^+ L_1 \cap \uea^+ L_2)_{\ell_{\rnkO -1}^\pm - \ell}} \Kern \psi^\pm_{U^\pm}.
\end{equation}
\end{proposition}

\noindent In case (2), the above two conditions again completely characterise
when $(\omega_1, \omega_2) \in \data'$.  As with case (1') however, case (2')
involves an additional \sv{} which leads to a further condition to check.  However, we can now use \propref{prop: right Verma 2} to rule out the rank $\rnkO - 1$ \svs{}, so we may assume that the minimal rank of
\lemref{lem: sufficient} is $\rnkO$. We now turn to the derivation of conditions for the additional rank $\rnkO$ \sv{} in cases (1') and (2').

\subsubsection{Further conditions in cases \textup{(1')} and \textup{(2')}} \label{sec: pain in the A}

When $(\omega_1, \omega_2) \in \minsubmod_{\ell-1} \oplus \minsubmod_{\ell-2}$,
\propDref{prop: right Verma 1}{prop: right Verma 2} 
give complete characterisations of the absence of rank $\rnkO-1$ \svs{}
in $\sN^\circ$, which suffices to settle the existence question of staggered
modules in cases (1) and (2). In cases (1') and (2'), \lemref{lem: sufficient}
still leaves the possibility that $\sN^\circ$ is non-trivial.  We must therefore also characterise the absence of the \sv{} $X_{\rnkO}^- x$ (which has a lower grade than $\omega_0 = X_{\rnkO}^+ x$) in $\sN^\circ$.

The derivation of this characterisation is similar in flavour to
the considerations of \secDref{sec: parameters case one}{sec: parameters case two}, although there are also important differences.
The most immediate difference is that we must assume from the outset that the $\rnkO-1$ \svs{} have already been ruled out.  This is necessary for the application of \lemref{lem: sufficient}, and we will see that the definition of the functional $\psi^\cap_{U^\cap}$ below will only make sense when $(\omega_1,\omega_2)$ satisfies this condition.
We point out also another difference that will be relevant later.  In
\secref{sec: invariants}, we will construct invariants of staggered modules in a manner closely related to the considerations of the two previous sections.  However, there will be no invariant related to what we have to do next.  We will return to this point in \secref{sec: invariants}.

To decide whether $X_{\rnkO}^- x$ is in $\sN^\circ$, we will define yet
another set of functionals $\psi^\cap_{U^\cap}$.
We recall that cases (1') and (2') both require $\LMod$ to be of braid type,
the former corresponding to $\rnkO=1$ and the latter to $\rnkO>1$.
To uniformise notation, we understand in the following that if $\rnkO=1$ then
$\psi^+_{U^+}$ stands for $\psi_U$ (as given in \secref{sec: parameters case one}) and $\psi^-_{U^-}$ is ignored (that is, each $\psi^-_{U^-}$ is to be regarded as the zero map).
For $U^\cap = U^\cap_1 L_1 = -U^\cap_2 L_2 \in \brac{\uea^+ L_1 \cap \uea^+ L_2}_{\ell_{\rnkO}^{-} - \ell}$ and $(\omega_1, \omega_2) \in \big( \bigcap_{U^-} \Kern \psi^-_{U^-}
\big) \cap \big( \bigcap_{U^+} \Kern \psi^+_{U^+} \big)$,
the defining formula is
\begin{equation} \label{eqnDefPsiCap}
U^\cap_1 \omega_1 + U^\cap_2 \omega_2 = \psi^\cap_{U^\cap} (\omega_1, \omega_2) X_{\rnkO}^{-} x.
\end{equation}
The definition makes sense, but only because of the restriction that
$(\omega_1, \omega_2)$ is already annihilated by every $\psi^\pm_{U^\pm}$.
This follows from the fact that the maximal proper submodule of $\uea X_{\rnkO-1}^{\pm}$ is generated by the rank $\rnkO$ \svs{}. For if
$U^\cap_1 \omega_1 + U^\cap_2 \omega_2$ were not proportional to $X_{\rnkO}^{-} x$, so $U^\cap_1 \omega_1 + U^\cap_2 \omega_2$ would not be in the submodule $\uea X_{\rnkO}^{-} x \subset \minsubmod$, there would exist a $U \in \uea^+$ such that $\psi^\pm_{U U^\cap} (\omega_1, \omega_2)$ is equal to either $X_{\rnkO-1}^- x$ or $X_{\rnkO-1}^+ x$, a contradiction.

The reason for this definition is the same as always. Assuming that both
$X_{\rnkO -1}^- x$ and $X_{\rnkO -1}^+ x$ are not in $\sN^\circ$, so that
$\psi^\cap_{U^\cap} (\omega_1, \omega_2)$ can be defined, \lemref{lem: sufficient} tells us that $\sN^\circ$ is either zero
or generated by $X_{\rnkO}^- x$. The analogue of \propDref{prop: right Verma 1}{prop: right Verma 2} is then:

\begin{proposition} \label{prop: intermediate singulars}
In the cases \textup{(1')} and \textup{(2')}, assuming that
$(\omega_1, \omega_2) \in \minsubmod_{\ell-1} \oplus \minsubmod_{\ell-2}$
is such that $\sN^\circ$ contains no rank $\rnkO-1$ \svs{},
we have $\sN^\circ = \set{0}$ if and only if
\begin{equation}
(\omega_1, \omega_2) \in \bigcap_{U^\cap \in (\uea^+ L_1 \cap \uea^+ L_2)_{\ell_{\rnkO}^- - \ell}} \Kern \psi^\cap_{U^\cap}.
\end{equation}
\end{proposition}

\noindent This completes the characterisation of admissibility in these cases.

\subsubsection{Summary}

We have defined functionals $\psi_U$, $\psi^\pm_{U^\pm}$ and $\psi^\cap_{U^\cap}$ whose kernels characterise when data is admissible.  We note that each of these functionals is manifestly gauge-invariant, so these kernels respect the gauge freedom we have in choosing the data.  By combining \lemref{lem: data from minsubmod} with \propTref{prop: right Verma 1}{prop: right Verma 2}{prop: intermediate singulars}, we now arrive at the complete classification of the admissible data in terms of these functionals.

\begin{theorem}[Cases (1), (1'), (2), (2') of the classification]
\label{thm: right Verma characterization of data}
Given $\LMod$ and $\RMod \cong \VerMod_{\RDim}$ with $\ell > 0$, choose $(\omega_1', \omega_2') \in \minsubmod_{\ell-1} \oplus \minsubmod_{\ell-2}$.  Then, $(\omega_1', \omega_2') \in \data'$, so is the data of a staggered module $\Stagg$ (with exact sequence (\ref{eq: exact seq Verma})), if and only if the appropriate condition below is satisfied:
\begin{description}
\item[Case (1)] $\psi_U (\omega_1', \omega_2') = 0$ for all
    $U \in (\uea^+ L_1 \cap \uea^+ L_2)_{\ell_{\rnkO -1} - \ell}$.
\item[Case (1')] In addition to condition \textup{(1)},
    $\psi^\cap_{U^\cap} (\omega_1', \omega_2') = 0$ for all
    $U^\cap \in (\uea^+ L_1 \cap \uea^+ L_2)_{\ell_1^- - \ell}$.
\item[Case (2)] $\psi^\pm_{U^\pm} (\omega_1', \omega_2') = 0$ for all
    $U^\pm \in (\uea^+ L_1 \cap \uea^+ L_2)_{\ell_{\rnkO -1}^{\pm} - \ell}$.
\item[Case (2')] In addition to condition \textup{(2)},
    $\psi^\cap_{U^\cap} (\omega_1', \omega_2') = 0$ for all
    $U^\cap \in (\uea^+ L_1 \cap \uea^+ L_2)_{\ell_{\rnkO}^- - \ell}$.
\end{description}
Here, $\rnkO = \rank \omega_0$, and the relevant condition to use matches the case numbering given at the beginning of \secref{sec: allowed parameters}.  Moreover, $(\omega_1, \omega_2) \in \Hlft_{\ell-1} \oplus \Hlft_{\ell-2}$
is in $\data$, hence is the data of a staggered module $\Stagg$, if and only
if there exist equivalent data $(\omega_1', \omega_2') \in \data'$.
\end{theorem}

\noindent We remark that the single case excluded from the above theorem ($\ell = 0$, case (0)) was already settled in \thmref{thmClassRVermaL=0}.

\subsection{Counting Dimensions} \label{sec: dimensions}

The results of \thmref{thm: right Verma characterization of data}
are very concrete descriptions of the possible data of staggered modules with
$\Hrgt = \Verma_{h^\rgt}$, even if they might seem somewhat technical.
Their value is that they involve linear maps with simple definitions, and
so allow reasonably straight-forward computations, in each case, of the 
dimensions of the vector space $\data' / \gauge'$ of inequivalent staggered 
modules.

To use \thmref{thm: right Verma characterization of data} to compute the
dimension of $\data' / \gauge'$, we will analyse the functionals $\psi_U$,
$\psi^-_{U^-}$, $\psi^+_{U^+}$ and $\psi^\cap_{U^\cap}$.  In fact, it proves
convenient to abstract one level further and consider also the induced
maps
\begin{subequations}
\begin{align}
\psi &\colon \brac{\uea^+ L_1 \cap \uea^+ L_2}_{\ell_{k-1} - \ell}
\longrightarrow \brac{\minsubmod_{\ell-1} \oplus \minsubmod_{\ell-2}}^*,
& U &\longmapsto \psi_U, \label{eqnDefMapPsi} \\
\psi^- &\colon \brac{\uea^+ L_1 \cap \uea^+ L_2}_{\ell_{k-1}^- - \ell}
\longrightarrow \brac{\minsubmod_{\ell-1} \oplus \minsubmod_{\ell-2}}^*,
& U^- &\longmapsto \psi^-_{U^-}, \label{eqnDefMapPsi-} \\
\psi^+ &\colon \brac{\uea^+ L_1 \cap \uea^+ L_2}_{\ell_{k-1}^+ - \ell}
\longrightarrow \brac{\minsubmod_{\ell-1} \oplus \minsubmod_{\ell-2}}^*,
& U^+ &\longmapsto \psi^+_{U^+}, \label{eqnDefMapPsi+} \\
\psi^\cap &\colon \brac{\uea^+ L_1 \cap \uea^+ L_2}_{\ell_k^- - \ell}
\longrightarrow \Bigl( \textstyle \brac{\bigcap_{U^-} \Kern \psi^-_{U^-}} \cap \brac{\bigcap_{U^+} \Kern \psi^+_{U^+}} \Bigr)^*,
& U^\cap &\longmapsto \psi^\cap_{U^\cap}. \label{eqnDefMapPsiCap}
\end{align}
\end{subequations}
All of these analyses are somewhat similar so we present instead two abstract
results along these lines from which the required dimension results will be
extracted on a case-by-case basis.
So consider a highest weight module $\ghwm$ with highest
weight $(\ghw , c)$ and cyclic highest weight vector $\ghwv$.
Fix a grade $m$.
Then for $(w_1, w_2) \in \ghwm_{m-1} \oplus \ghwm_{m-2}$
and $U = U_1 L_1 = -U_2 L_2 \in (\uea^+ L_1 \cap \uea^+ L_2)_{-m}$, we
define $\gpsi_U(w_1,w_2)$ by
\begin{equation}
U_1 w_1 + U_2 w_2 = \gpsi_U(w_1, w_2) \ghwv.
\end{equation}
This definition is clearly in the same spirit as those of $\psi$,
$\psi^\pm$ and $\psi^\cap$.  As above, $\gpsi_U$ is then
the corresponding functional on $\ghwm_{m-1} \oplus \ghwm_{m-2}$,
and $\gpsi$ alone stands for the map from 
$(\uea^+ L_1 \cap \uea^+ L_2)_{-m}$ to $(\ghwm_{m-1} \oplus \ghwm_{m-2})^*$
that associates to any $U$ the functional $\gpsi_U$.

We want to know when $\gpsi_U$ is non-trivial. This is the subject of
the following result.

\begin{lemma} \label{lem: kernel of psi}
The functional
$\gpsi_U \in (\ghwm_{m-1} \oplus \ghwm_{m-2})^*$ is zero if and only if
$U = U_1 L_1 = -U_2 L_2 \in (\uea^+ L_1 \cap \uea^+ L_2)_{-m}$
is such that the $U_j^\dagger v_{\ghw}$ ($j=1,2$) are in the maximal proper
submodule of the Verma module $\Verma_{\ghw}$.
\end{lemma}

\noindent In particular, we will often use this result to establish the
injectivity of $\gpsi$ by noting that if there is no proper singular
vector in $\Verma_{\ghw}$ of grade less than $m$, then the only
$U$ for which $\gpsi_U$ vanishes is $U=0$ (at grades $m-j$ the maximal
proper submodule is trivial).

\begin{proof}
Write $U = U_1 L_1 = -U_2 L_2$.  By definition, $\gpsi_U = 0$ means
$U_1 w_1 + U_2 w_2 = 0$ for all $(w_1, w_2) \in \ghwm_{m-1} \oplus \ghwm_{m-2}$.
Taking $w_1=0$ and $w_2=0$ (separately), we see that this is equivalent to
$U_j w_j = 0$ for all $w_j \in \ghwm_{m-j}$ ($j=1,2$).  Writing
$w_j = V_j \ghwv$, we can further reformulate this as $U_j V_j \ghwv = 0$
for all $V_j \in \uea^-_{m-j}$ ($j=1,2$), from which we derive
\begin{equation}
0 = \tinner{U_j V_j \ghwv}{\ghwv}_\ghwm
= \tinner{V_j \ghwv}{U_j^\dagger \ghwv}_\ghwm
= \tinner{V_j v_\ghw}{U_j^\dagger v_\ghw}_{\VerMod_\ghw},
\end{equation}
where we have distinguished the Shapovalov forms by a subscript displaying
the relevant \hwm{}.  We therefore conclude that $\gpsi_U = 0$ if and only
if both $U_1^\dagger v_\ghw$ and $U_2^\dagger v_\ghw$ belong to the maximal
proper submodule of $\Verma_{\ghw}$.
\end{proof}

Let us now assume that
there is a non-zero \emph{prime} singular vector $\gX \ghwv$ in $\ghwm$
(playing the r\^{o}le of $\omega_0$), where
$\gX \in \uea^-_{\gell}$ is singular (and normalised, say).
We take $m = \gell$. If the corresponding Verma module $\Verma_{\ghw}$
has another (normalised, prime) \sv{} of grade less than $\gell$, we will
denote it by $\gXoth v_{\ghw}$ and its grade by $\gellm < \gell$.
In this new setup, the content of \lemref{lem: kernel of psi} is simply
described as follows. If $\gXoth$ is not defined, then $\gpsi_U=0$
only if $U=0$, as follows from the remark immediately following the statement.
On the other hand, if $\gXoth$ is defined, then we see that
$\gpsi_U=0$ if and only if
$U \in (\gXoth)^\dagger (\uea^+ L_1 \cap \uea^+ L_2)_{\gellm-\gell}$.
This follows from factorising each $U_j^\dagger$ as $(U'_j)^\dagger \gXoth$,which leads to $U = U_1 L_1 =  (\gXoth)^\dagger U_1' L_1$ and
$U = -U_2 L_2 = -(\gXoth)^\dagger U_2' L_2$.

The following result will allow us to compute the dimensions
of the space of inequivalent staggered modules.  We mention that the first
of the three cases appearing here was at the heart of Rohsiepe's analysis
\cite{RohRed96}, although he only stated it for modules of chain type.

\begin{lemma} \label{lem: dimension of the annihilator of psis}
The subspace of $\ghwm_{\gell-1} \oplus \ghwm_{\gell-2}$ that
is annihilated by every $\gpsi_U$ has dimension given by
\begin{equation}
\dim \ \bigcap_U \Kern \gpsi_U = 
\begin{cases}
\partnum{\gell} & \text{if $\gXoth$ is not defined,} \\
\partnum{\gell} - \partnum{\gell-\gellm} & \textrm{if $\gXoth \ghwv = 0$,} \\
\partnum{\gell} - \partnum{\gell-\gellm} + \partnum{\gell-\gellm - 1} + \partnum{\gell-\gellm - 2} & \textrm{if $\gXoth \ghwv \neq 0$.}
\end{cases}
\end{equation}
In the first two cases the result coincides with $\dim \ghwm_{\gell}$
and in the third case with $\dim \ghwm_{\gell} +
\dim (\uea^+ L_1 \cap \uea^+ L_2)_{\gellm - \gell}$.
\end{lemma}
\begin{proof}
Taking $U_\mu$ such that $\set{\gpsi_{U_\mu}}$ is a basis for $\Imag \gpsi$,
the mapping
\begin{equation}
(w_1,w_2) \mapsto
(\gpsi_{U_1}(w_1,w_2) , \ldots , \gpsi_{U_{n}}(w_1,w_2)) \in \bC^{n}
\end{equation}
has kernel given by $\cap_U \Kern \gpsi_U$ and rank equal to $\dmn \ \Imag \gpsi$.  In other words, each linearly independent equation
$\gpsi_U (w_1, w_2) = 0$ reduces the dimension we want to compute by one:
\begin{align}
\dim \ \bigcap_U \Kern \gpsi_U &= \dim \ (\ghwm_{\gell-1} \oplus \ghwm_{\gell-2}) - \dim \ \Imag \gpsi \notag \\
&= \dim \ (\ghwm_{\gell-1} \oplus \ghwm_{\gell-2}) - \dim \ (\uea^+ L_1 \cap \uea^+ L_2)_{-\gell} + \dim \ \Kern \gpsi. \label{eqnsDimComputers}
\end{align}

Consider therefore the case in which $\Verma_{\ghw}$ has no \sv{} of grade
less than $\gell$ (except $v_{\ghw}$), so $\gXoth$ is not defined.
Then we have $\dim \ \ghwm_{\gell-j} = \tpartnum{\gell-j}$ for $j=1,2$. But,
\lemref{lem: kernel of psi} tells us that in this case
(with $m = l$), $U \mapsto \gpsi_U$ has a trivial kernel:
$\dim \ \Kern \gpsi = 0$.
Plugging these facts and the result of \lemref{lem: L1 L2} into
\eqnref{eqnsDimComputers}, the first formula follows.

Consider now the cases for which $\gXoth$ is defined. Regardless of whether
$\gXoth \ghwv$ vanishes or not, \lemref{lem: kernel of psi} gives (with
$m = l$)
$\Kern \gpsi = (\gXoth)^\dagger (\uea^+ L_1 \cap \uea^+ L_2)_{\gellm - \gell}$,
hence
\begin{equation}
\dmn \Kern \gpsi = \dmn (\uea^+ L_1 \cap \uea^+ L_2)_{\gellm - \gell}
= \partnum{\gell-\gellm - 1} + \partnum{\gell-\gellm - 2} -
    \partnum{\gell-\gellm},
\end{equation}
by \lemref{lem: L1 L2}.  When $\gXoth \ghwv = 0$, the graded dimensions of
$\ghwm$ are $\dmn \ghwm_{\gell-j} = p(\gell-j) - p(\gell-\gellm - j)$ for 
$j=1,2$.  Plugging everything in and observing cancellations gives the
second formula.
On the other hand, if $\gXoth \ghwv \neq 0$ the graded dimensions are
$\dmn \ghwm_{\gell-j} = p(\gell-j)$ and the third formula follows.
\end{proof}

With help of \lemref{lem: dimension of the annihilator of psis},
we are ready to state and prove one of our main results,
that giving the dimensions of the space of isomorphism classes of
staggered modules, $\data' / \gauge'$, when the right module is Verma.

\begin{theorem} \label{thm: moduli space right Verma}
The dimension of the vector space $\data' / \gauge'$ of isomorphism classes of
staggered modules $\Stagg$ with short
exact sequence (\ref{eq: exact seq Verma})
is the number of rank $\rnkO-1$ \svs{} in $\sH^\lft$.
Explicitly,
\begin{center}
\begin{tabular}{rcl}
\textup{\textbf{Case (0)}:} \quad & ($\ell=0$) & \quad $\dmn \data' / \gauge' = 0$, \\
\textup{\textbf{Cases (1)}} and \textup{\textbf{(1')}:} \quad & ($\LMod$ of chain type or $\rnkO=1$ braid type)
& \quad $\dmn \data' / \gauge' = 1$, \\
\textup{\textbf{Cases (2)}} and \textup{\textbf{(2')}:} \quad & ($\LMod$ of $\rnkO>1$ braid type)
& \quad $\dmn \data' / \gauge' = 2$.
\end{tabular}
\end{center}
\end{theorem}
\begin{proof}
Case (0) being already done (\thmref{thmClassRVermaL=0}),
we will have to work out the cases (1), (1'), (2) and (2') of
\thmref{thm: right Verma characterization of data} separately.
As we already know that $\dmn G' = \dmn \minsubmod_{\ell} - 1$ (\eqnref{eqnDimGauge'}), it remains to be shown that 
$\dmn \data' = \dmn \minsubmod_{\ell}$ in cases (1) and (1'), 
and that $\dmn \data' = \dmn \minsubmod_{\ell} + 1$ in cases (2) and (2').

\textbf{Case (1):} Let $\ghwm = \minsubmod = \uea X_{\rnkO-1} x$ and define
$\gX$ by $X \equiv X_\rnkO =\gX X_{\rnkO-1}$.  This $\gX$
is then normalised and prime, and $\gell$ is given by $\ell - \ell_{\rnkO-1}$.
Let $\gpsi$ be $\psi$, as defined in \eqnref{eqnDefMapPsi}.
When $\LMod$ is of chain type or of braid type with $\omega_0 = X_1^- x$,
\lemref{lem: dimension of the annihilator of psis} applies with $\gXoth$
undefined. Since $\data' = \bigcap_U \Kern \psi_U$, we read off the dimension
\begin{equation}
\dmn \data' = \partnum{\ell - \ell_{\rnkO-1}} = \dmn \minsubmod_{\ell}.
\end{equation}
The outstanding possibility, when $\LMod$ is of braid type with
$\omega_0 = X_{1}^{+} x$, is such that
\lemref{lem: dimension of the annihilator of psis} applies with
$\gXoth = X_1^-$, hence $\gellm = \ell_1^-$. But for case (1),
$X_1^- x = 0$, so the second formula in the lemma also gives the
dimension of $\data'$ as
\begin{equation}
\dmn \data' = \partnum{\ell} - \partnum{\ell_{1}^{-}} = \dmn \minsubmod_{\ell}.
\end{equation}

\textbf{Case (1'):}
This can only occur in the $\rnkO=1$ braid case with $\omega_0 = X_{1}^{+} x$
and $X_1^- x \neq 0$. We set $\ghwm = \minsubmod = \Hlft$ and
$\gX = X_{1}^{+} = X$, $\gXoth = X_1^-$, so $\gell = \ell$ and
$\gellm = \ell_1^-$.
From the third case of \lemref{lem: dimension of the annihilator of psis},
we read off
\begin{equation}
\dim \ \bigcap_U \Kern \psi_U = \dim \ \minsubmod_{\ell}
+ \dim \ (\uea^+ L_1 \cap \uea^+ L_2)_{\ell_1^- - \ell}.
\end{equation}
But in case (1'), $\data'$ is a only a subset of this intersection:
$\data' = \bigcap_{U^\cap} \psi^\cap_{U^\cap} \subseteq \bigcap_U \Kern \psi_U$
(and the inclusion is typically strict).  Accounting for the extra
conditions imposed by the $\psi^\cap_{U^\cap}$ means that the dimension of
the admissible data is reduced by $\dim \ \Imag \psi^\cap$, which is of
course given by
\begin{equation}
\dim \ \Imag \psi^\cap =
\dim \ (\uea^+ L_1 \cap \uea^+ L_2)_{\ell_1^- - \ell} - \dim \ \Kern \psi^\cap.
\end{equation}
Thus, $\dim \ \data' = \dim \ \minsubmod_{\ell} + \dim \ \Kern \psi^\cap$.

To show injectivity of $\psi^\cap$ and complete the computation, note first
that $(\uea X_1^- x)_{\ell-1} \oplus (\uea X_1^- x)_{\ell-2} \subseteq
\bigcap_U \Kern \psi_U$, so $\psi^\cap$ is defined on this subspace.
Now we apply \lemref{lem: kernel of psi} to $\ghwm = \uea X_1^- x$,
$m = \ell - \ell_1^-$ and $\gpsi = \psi^\cap$.  Since
$\VerMod_{\LDim + \ell_1^-}$ has no \svs{} of grade less than $\gell$
(except $v_{\LDim + \ell_1^-}$ itself), we conclude that
$\psi^\cap_{U^\cap} = 0$ implies $U^\cap = 0$.

\textbf{Case (2):}
In the braid case with $\rnkO>1$ we have
\begin{equation} \label{eqnMSMDecomp}
\minsubmod = \uea X_{\rnkO-1}^{-} x + \uea X_{\rnkO-1}^{+} x.
\end{equation}
In case (2), the sum is direct at grades smaller than $\ell$, so we may
uniquely decompose every $w_j \in \minsubmod_{\ell - j}$ as 
\begin{equation} \label{eqnWDecomp}
w_j = w_j^- + w_j^+, \qquad
\text{with $w_j^\pm \in \brac{\uea X_{\rnkO-1}^{\pm} x}_{\ell - j}$.}
\end{equation}
We proceed by considering the ``$-$'' and ``$+$'' pieces separately.

The space whose dimension we want to compute is 
\begin{equation}
\data' = \brac{\bigcap_{U^-} \Kern \psi^-_{U^-}} \ \cap \ 
\brac{\bigcap_{U^+} \Kern \psi^+_{U^+}}.
\end{equation}
We take $\ghwm = \uea X_{\rnkO-1}^\pm x$, $\gpsi = \psi^\pm$,
$X = \gX X_{\rnkO-1}^\pm$, $\gell = \ell - \ell_{\rnkO-1}^\pm$, and if
defined, $X_\rnkO^- = \gXoth X_{\rnkO-1}^\pm$ and
$\gellm = \ell_\rnkO^- - \ell_{\rnkO-1}^\pm$.  Then, the first or second formula
of \lemref{lem: dimension of the annihilator of psis} (as appropriate)
gives the dimension of $\bigcap_{U^\pm} \Kern \psi^\pm_{U^\pm}$, where the
$\psi^\pm_{U^\pm}$ are restricted to the direct sum of the subspaces
$\brac{\uea X_{\rnkO-1}^{\pm} x}_{\ell - 1} \oplus
\brac{\uea X_{\rnkO-1}^{\pm} x}_{\ell - 2}$
(spanned by the $(w_1^\pm,w_2^\pm)$ of \eqnref{eqnWDecomp}).  The result 
is that this dimension coincides with that of
$\ghwm_\gell = \brac{\uea X_{\rnkO-1}^\pm x}_\ell$.

But from the definition of the $\psi^\pm_{U^\pm}$, \eqnref{eqnDefPsiPM},
we quickly determine that the $\psi^\pm_{U^\pm}$ always annihilate the
subspace $\brac{\uea X_{\rnkO-1}^{\mp} x}_{\ell - 1} \oplus
\brac{\uea X_{\rnkO-1}^{\mp} x}_{\ell - 2}$.  The dimension we want is
therefore just the sum
\begin{equation}
\dim \ \data' = \dim \ \brac{\uea X_{\rnkO-1}^- x}_\ell + \dim \ \brac{\uea X_{\rnkO-1}^+ x}_\ell = \dim \ \minsubmod_\ell + 1,
\end{equation}
where the additional $1$ derives from the fact that the decomposition (\ref{eqnMSMDecomp}) is not direct at grade $\ell$ because of the one-dimensional intersection spanned by $\omega_0$.

\textbf{Case (2'):}
As in the previous case, we use
\lemref{lem: dimension of the annihilator of psis} to compute the dimension
of $\bigcap_{U^\pm} \Kern \psi^\pm_{U^\pm}$, where the $\psi^\pm_{U^\pm}$ are
restricted to act on pairs of descendants (of the appropriate grade) of
$X_{\rnkO-1}^\pm x$.  This time we must use the third formula, with the 
result that this dimension is
\begin{equation*}
\dim \ \brac{\uea X_{\rnkO-1}^\pm x}_\ell + \dim \ \brac{\uea^+ L_1 \cap \uea^+ L_2}_{\ell_\rnkO^- - \ell}.
\end{equation*}
The sum (\ref{eqnMSMDecomp}) is no longer direct at grades less than
$\ell$, but we still know that each $\psi^\pm_{U^\pm}$ annihilates pairs
$(w_1 , w_2)$ whose elements $w_j$ are in
$\brac{\uea X_{\rnkO-1}^\mp x}_{\ell - j}$. 
Consequently, any pair whose elements are in the intersection of
these subspaces, $\brac{\uea X_{\rnkO}^- x}_{\ell - j}$, is also
annihilated.  It follows then that
\begin{align}
\dim \ \Bigl( \bigl( {\textstyle \bigcap_{U^-} \Kern \psi^-_{U^-}} \bigr) \ &\cap \ \bigl( {\textstyle \bigcap_{U^+} \Kern \psi^+_{U^+}} \bigr) \Bigr) \notag \\
&= \dim \ \brac{\uea X_{\rnkO-1}^- x}_\ell + \dim \ \brac{\uea X_{\rnkO-1}^+ x}_\ell + 2 \dim \ \brac{\uea^+ L_1 \cap \uea^+ L_2}_{\ell_\rnkO^- - \ell} \notag \\
& \mspace{300mu} - \dim \brac{\uea X_{\rnkO}^-}_{\ell-1} - \dim \brac{\uea X_{\rnkO}^- x}_{\ell-2} \notag \\
&= \partnum{\ell-\ell_{\rnkO-1}^-} + \partnum{\ell-\ell_{\rnkO-1}^+} + \partnum{\ell-\ell_{\rnkO}^- - 1} + \partnum{\ell-\ell_{\rnkO}^- - 2} - 2 \partnum{\ell-\ell_{\rnkO}^-} \notag \\
&= \dim \ \minsubmod_\ell + 1 + \dim \ \brac{\uea^+ L_1 \cap \uea^+ L_2}_{\ell_{\rnkO}^{-} - \ell}.
\end{align}

Finally, we recall that
$\data' = \bigcap_{U^\cap} \Kern \psi^\cap_{U^\cap} \subseteq
\big( \bigcap_{U^-} \Kern \psi^-_{U^-} \big) \cap \big( \bigcap_{U^+} \Kern \psi^+_{U^+} \big)$.
As in case (1'), this implies that
\begin{equation}
\dim \ \data' = \dim \ \minsubmod_{\ell} + 1 + \dim \ \Kern \psi^\cap,
\end{equation}
and the injectivity of $\psi^\cap$ follows from the same argument as before.
This completes our computations.
\end{proof}

\subsection{Invariants as Coordinates} \label{sec: invariants}

We have seen in \thmref{thm: moduli space right Verma} that the number of
rank $\rnkO-1$ \svs{} of $\LMod$ coincides with the dimension
of the vector space $\data' / \gauge'$ (equivalently $\data / \gauge$) of
isomorphism classes of staggered modules with the short exact sequence
\eqref{eq: exact seq Verma}.
Next, we will construct coordinates on this vector space by defining invariants
$\beta$ or $\beta_\pm$ of the data defining the staggered module.

In cases (1) and (1'), recall that $\minsubmod$ is generated by the \sv{} $X_{\rnkO-1} x$.  We define $\gX \in \uea^-$ so that $X = \gX X_{\rnkO-1}$ ($\gX$ is then singular, normalised and prime). Since $\gX$ is not a scalar, we may write $\gX^\dagger = Y_1 L_1 + Y_2 L_2$, though $Y_1$ and $Y_2$ are not uniquely specified.  Nevertheless, every choice of $Y_1$ and $Y_2$ defines a functional $\tilde{\beta} \in (\minsubmod_{\ell-1} \oplus \minsubmod_{\ell-2})^*$ by
\begin{equation} \label{eq: invariant definition 1}
Y_1 \omega_1 + Y_2 \omega_2 = \tilde{\beta}(\omega_1,\omega_2) X_{\rnkO-1} x.
\end{equation}
Because $\gX$ is singular, this functional is invariant under the action of the gauge group $\gauge'$:
\begin{equation}
\tilde{\beta}(\omega_1 + L_1 u,\omega_2 + L_2 u) - \tilde{\beta}(\omega_1,\omega_2) = \brac{Y_1 L_1 + Y_2 L_2} u = \gX^\dagger u = 0 \qquad \text{($u \in \minsubmod_\ell$).}
\end{equation}
Moreover, it should be clear that if the data is admissible, $\brac{\omega_1 , \omega_2} \in \data'$, then $\tilde{\beta}$ does not depend upon the choice made for $Y_1$ and $Y_2$.  In this case, gauge invariance implies that we have a well-defined functional on $\data' / \gauge'$.  This is our coordinate, and we denote it by $\beta$.


Similarly, in cases (2) and (2'), $\minsubmod$ is generated by the \svs{} $X_{\rnkO-1}^{-} x$ and $X_{\rnkO-1}^{+} x$, and we define $\gX_\pm \in \uea^-$ so that $X = \gX_\pm X_{\rnkO-1}^\pm$ (making the $\gX_\pm$ singular, normalised and prime).  Again, the $\gX_\pm$ are not scalars, hence we may write $(\gX_\pm)^\dagger = Y^\pm_1 L_1 + Y^\pm_2 L_2$ (non-uniquely), and define functionals $\tilde{\beta}_\pm \in (\minsubmod_{\ell-1} \oplus \minsubmod_{\ell-2})^*$ by
\begin{subequations} \label{eq: invariant definition 2}
\begin{align}
Y^-_1 \omega_1 + Y^-_2 \omega_2 &= \tilde{\beta}_-(\omega_1,\omega_2) X_{\rnkO-1}^{-} x \\
\text{and} \qquad Y^+_1 \omega_1 + Y^+_2 \omega_2 &= \tilde{\beta}_+(\omega_1,\omega_2) X_{\rnkO-1}^{+} x \pmod{\uea X_{\rnkO-1}^{-} x}.
\end{align}
\end{subequations}
As above, the singularity of the $\gX_\pm$ implies that these functionals are invariant under $\gauge'$, and when $\brac{\omega_1 , \omega_2} \in \data'$, the definitions do not depend upon the choice of $Y^\pm_1$ and $Y^\pm_2$.  Thus, we obtain two coordinates on $\data' / \gauge'$ in this case, and we denote the corresponding functionals by $\beta_\pm$.

We remark that even in the cases (1') and (2'), we do not define an invariant
related to the singular vector $X_{\rnkO}^- x$.  We cannot even write down a
formula analogous to \eqnDref{eq: invariant definition 1}{eq: invariant
definition 2}, because $\omega_0 = X_{\rnkO}^+ x$ is \emph{not} a descendant
of $X_{\rnkO}^- x$.  Even if one could concoct such a formula, it is difficult
to imagine why the corresponding quantity should be gauge invariant.  In any
case, we will see in \thmref{thm: invariants} below that the invariants we
have already defined are sufficient to completely characterise a staggered
module with exact sequence (\ref{eq: exact seq Verma}).

Note that if $\brac{\omega_1 , \omega_2} \in \data'$, so we do indeed have a staggered module (with right module Verma), then $\omega_j = L_j y$ for $j=1,2$.  Hence we may write (abusing notation in an obvious manner)
\begin{equation}
\beta X_{\rnkO-1} x = \brac{Y_1 L_1 + Y_2 L_2} y
    = \gX^\dagger y \qquad \iff \qquad \beta
    = \inner{X_{\rnkO-1} x}{\gX^\dagger y}_{\uea X_{\rnkO-1} x}.
\end{equation}
Similar formulae may be written for $\beta_\pm$, although for $\beta_+$, one should include a projection from $\minsubmod$ onto $\uea X_{\rnkO-1}^+ x$.  It is in this form that we may compare these invariant coordinates with the beta-invariant defined in \eqnref{eqnDefBeta}.

It is the latter invariant which has been used in the literature to
distinguish staggered modules with the same exact sequence, though we
have already noted (\eqnref{eqnBeta=0}) that this beta-invariant vanishes
whenever $\rnkO = \rank \omega_0 > 1$.  This has not been found to
problematic thus far because, to the best of our knowledge, only modules
with $\rnkO \leqslant 1$ have been found to be relevant in applications. 
Nevertheless, this vanishing is a conceptual problem which is solved by
the invariant coordinates introduced above.  Namely, when $\rnkO = 1$
(cases (1) and (1')), the beta-invariant of \eqnref{eqnDefBeta} coincides
with the (value of the) coordinate $\beta$, because $\gX = X$ (this is why
we have risked some confusion by using the same notation
for the coordinates and invariants).  When
$\rnkO > 1$ and the beta-invariant vanishes identically, we have instead
the coordinates $\beta$ (cases (1) and (1')) or $\beta_\pm$ (cases (2) 
and (2')).  We therefore feel justified in concluding that the invariant
coordinates defined here should \emph{replace} the (in hindsight,
na\"{\i}ve) definition of the beta-invariant given in \secref{secStag}.

There is one point that remains to be addressed.  The beta-invariant of
\secref{secStag} vanishes identically when $\rnkO > 1$, hence is useless
in this case for distinguishing staggered modules with the same exact 
sequence.  We claim that the invariant coordinates defined above are 
superior in this respect, so we need to establish that the invariant
coordinates $\beta$ or $\beta_\pm$ are \emph{linearly independent}
functionals on the vector space $\data' / G'$, that is, that they are
actually coordinates.  We remark that this would complete our analysis 
of staggered modules when the right module is Verma. Indeed, the vector 
space of inequivalent staggered modules with a given short exact sequence
(\ref{eq: exact seq Verma}) was seen in \thmref{thm: moduli space right 
Verma} to have dimension $0$, $1$ or $2$. As the number of
coordinates we
have constructed precisely matches the dimension of $\data'/G'$ in each
case, they completely characterise the staggered module (again, given a 
short exact sequence).  Practically, this means that the formulae given
in \eqnDref{eq: invariant definition 1}{eq: invariant definition 2} reduce
the identification of a staggered module (\ref{eq: exact seq Verma}) to 
the computation of one or two numbers.

\begin{theorem} \label{thm: invariants}
In cases (1) and (1'), the functional $\beta$ is not identically zero on the one-dimensional vector space $\data' / \gauge'$, and so parametrises it.  In cases (2) and (2'), the functionals $\beta_-$ and $\beta_+$ are non-zero and linearly independent on the two-dimensional vector space $\data' / \gauge'$, and so parametrise it.
\end{theorem}
\begin{proof}
We first note that to show that a functional $\tilde{\beta}$ on a
finite-dimensional vector space $V$ is non-vanishing on the intersection
of the kernels of a collection of functionals $\set{\psi_U}$, it is enough
to prove that $\tilde{\beta}$ is linearly independent of this collection.  This 
follows quite readily by taking a basis for the span of
$\set{\psi_U}$, extending
it to a basis of $V^*$, and then considering the action of $\tilde{\beta}$ 
on the dual basis (identifying $V^{**}$ and $V$ in the standard way).  Our
strategy below is therefore to prove that $\tilde{\beta}$ and its variants
are linearly independent of the $\psi_U$ (and its variants), so $\beta$
is non-zero.

\textbf{Case (1):}
Assume that $\tilde{\beta}$ is a linear combination of the $\psi_U$:  $\tilde{\beta} = \sum_U b_U \psi_U = \psi_{B} \in
(\minsubmod_{\ell-1} \oplus \minsubmod_{\ell-2})^*$ for some
$B = \sum_U b_U U = B_1 L_1 = - B_2 L_2$.  
Then, from the definitions (\ref{eqnDefPsi}) and (\ref{eq: invariant definition 1}), we get
\begin{equation}
Y_1 w_1 + Y_2 w_2 = B_1 w_1 + B_2 w_2 \qquad \text{for all $w_1 \in \minsubmod_{\ell-1}$ and $w_2 \in \minsubmod_{\ell-2}$,}
\end{equation}
where $Y_1 L_1 + Y_2 L_2 = \gX^\dagger$ is such that $X = \gX X_{\rnkO-1}$ (so $\gX$ is non-zero and singular).  Setting $w_2 = 0$, we find that $Y_1 - B_1$ must annihilate $\minsubmod_{\ell-1}$.  However, this implies that
\begin{equation} \label{eqnNullState}
\inner{\brac{Y_1-B_1}^\dagger X_{\rnkO-1} x}{\minsubmod_{\ell-1}}_\minsubmod = \inner{\vphantom{\brac{Y_1-B_1}^\dagger} X_{\rnkO-1} x}{\brac{Y_1-B_1} \minsubmod_{\ell-1}}_\minsubmod = 0,
\end{equation}
hence that $\brac{Y_1-B_1}^\dagger X_{\rnkO-1} x$ is a grade $\ell-1$ descendant of a (non-cyclic) \sv{} of $\minsubmod$.  But, in case (1), $\minsubmod$ has no non-trivial \svs{} of grade less than $\ell$ (except of course for $X_{\rnkO-1} x$ itself).  Thus, $\brac{Y_1-B_1}^\dagger X_{\rnkO-1} x = 0$.

When the Verma module corresponding to $\minsubmod$ has no \svs{} of grade less than $\ell$, we may conclude that $Y_1 = B_1$, and repeating this argument for $w_1 = 0$, that $Y_2 = B_2$.  Then, we obtain a contradiction:
\begin{equation}
\gX^\dagger = Y_1 L_1 + Y_2 L_2 = B_1 L_1 + B_2 L_2 = 0.
\end{equation}
However, case (1) also includes the possibility that $\minsubmod = \LMod$ is of braid type with $\rnkO=1$, $\gX = X = X_1^+$ and $X_1^- x = 0$.  Then, we can only conclude that $\brac{Y_1-B_1}^\dagger = V_1 \gXoth$ for some $V_1 \in \uea^-$, where $\gXoth = X_1^-$ is singular.  Similarly, taking $w_2 = 0$ now leads to $\brac{Y_2-B_2}^\dagger = V_2 \gXoth$ for some $V_2 \in \uea^-$, and we arrive at
\begin{equation} \label{eqnContradiction2}
\gX^\dagger = Y_1 L_1 + Y_2 L_2 = B_1 L_1 + B_2 L_2 + \brac{\gXoth}^\dagger \brac{V_1^\dagger L_1 + V_2^\dagger L_2} = \brac{\gXoth}^\dagger \brac{V_1^\dagger L_1 + V_2^\dagger L_2}.
\end{equation}
This is again a contradiction, because it implies that $\gX x = X_1^+ x$ is a descendant of $\gXoth x = X_1^- x$.  It therefore follows that in case (1), $\tilde{\beta}$ is linearly independent of the $(\psi_U)$, so $\beta \in (\data'/\gauge')^*$ is non-vanishing.

\textbf{Case (1'):}
In this case, $\data' \subseteq \bigcap_U \Kern \psi_U$, so
we again need
$\tilde{\beta}$ to be linearly independent of the $\psi_U$.
If this were not the case, we would use the argument which settles case
(1) to derive the contradiction of \eqnref{eqnContradiction2} (the sole
difference arises because $\gXoth x \neq 0$ ($\gXoth = X_1^-$), so
\eqnref{eqnNullState} would give $\brac{Y_j-B_j}^\dagger x = V_j \gXoth x$
for some $V_j \in \uea^-$, recovering $\brac{Y_j-B_j}^\dagger = V_j \gXoth$).
Therefore, $\tilde{\beta}$ does not vanish identically on
$\bigcap_U \Kern \psi_U$.  However, we still have to rule out the
possibility that $\tilde{\beta}$ might vanish on the (typically proper)
subset $\data' = \bigcap_{U^\cap} \Kern \psi^\cap_{U^\cap}$.  To do this,
note that there exists a pair
$(w_1, w_2) \in \bigcap_U \Kern \psi_U$ for which 
$\tilde{\beta}(w_1, w_2) \neq 0$. We will use this pair to construct a
pair $(w_1', w_2') \in \bigcap_{U^\cap} \Kern \psi^\cap_{U^\cap}$ which has
the same (non-zero) value as $(w_1, w_2)$ under $\tilde{\beta}$, thereby
establishing that $\tilde{\beta} \neq 0$ on $\data'$.

The key observation is that any
$(w_1^\cap, w_2^\cap) \in (\uea \gXoth x)_{\ell-1} \oplus 
(\uea \gXoth x)_{\ell-2}$
is annihilated by $\tilde{\beta}$ and every $\psi_U$, but not in general
by the $\psi^\cap_{U^\cap}$.  We may therefore ``shift'' our pair
$(w_1, w_2)$ by any such $(w_1^\cap, w_2^\cap)$ without affecting
membership in $\bigcap_U \Kern \psi_U$ or changing its value under 
$\tilde{\beta}$.  Take then a basis
$\tset{\psi^\cap_{U^\cap_\mu}}$ of $\Imag \psi^\cap$, and notice that as
the restriction to $(\uea \gXoth x)_{\ell-1} \oplus (\uea \gXoth x)_{\ell-2}$
of $\psi^\cap_{U^\cap}$ is zero only for $U^\cap=0$
(\lemref{lem: kernel of psi}), this remains a basis for the restrictions.
Extend arbitrarily to a basis of 
$\bigl( (\uea \gXoth x)_{\ell-1} \oplus (\uea \gXoth x)_{\ell-2} \bigr)^*$.
Let the corresponding dual basis of 
$(\uea \gXoth x)_{\ell-1} \oplus (\uea \gXoth x)_{\ell-2}$ be denoted by 
$\tset{\bigl( w_1^{(\mu)} , w_2^{(\mu)} \bigr)}$, so in particular, 
$\tfunc{\psi^\cap_{U^\cap_\mu}}{w_1^{(\nu)} , w_2^{(\nu)}} = \delta_{\mu,\nu}$.  
Choosing now
\begin{equation}
w_j^\cap = \sum_\mu \psi^{\cap}_{U^\cap_\mu}(w_1,w_2) \ w_j^{(\mu)},
\end{equation}
we quickly compute that $(w_1',w_2') = (w_1-w_1^\cap, w_2-w_2^\cap)$ is annihilated by every $\psi^\cap_{U^\cap}$. Since $(w_1',w_2') \in \bigcap_U \Kern \psi_U$ and $\tilde{\beta}(w_1',w_2') = \tilde{\beta}(w_1,w_2) \neq 0$, this proves that $\tilde{\beta} \neq 0$ on $\data'$.

\textbf{Cases (2) and (2'):}
In these cases, we once again use the decomposition
\begin{equation} \label{eqnSplit}
\minsubmod_{\ell-j} =
\brac{\uea X_{\rnkO-1}^{-} x}_{\ell-j} + 
\brac{\uea X_{\rnkO-1}^{+} x}_{\ell-j},
\end{equation}
where the sum is direct in case (2) but not in case (2').  We therefore write $w_j = w_j^- + w_j^+$ with
$w_j^\pm \in \uea X_{\rnkO-1}^\pm x$ ($j=1,2$).
The non-uniqueness of this decomposition in case (2') leads to no difficulties in what follows.

We start by observing that the restrictions of our functionals to the ``wrong'' subspaces are trivial:
\begin{equation} \label{eqnWrongSubspace}
\psi^\pm_{U^\pm} = \tilde{\beta}_\pm = 0 \qquad \text{on} \qquad \brac{\uea X_{\rnkO-1}^{\mp} x}_{\ell-1} \oplus \brac{\uea X_{\rnkO-1}^{\mp} x}_{\ell-2}.
\end{equation}
In particular, in case (2'), all these functionals vanish on the intersection $\brac{\uea X_{\rnkO}^{-} x}_{\ell-1} \oplus \brac{\uea X_{\rnkO}^{-} x}_{\ell-2}$ (which is why non-uniqueness leads to no difficulties).  It follows from this that if the $\tilde{\beta}_\pm$ are non-zero on $\data'$, their linear independence, and hence that of the $\beta_\pm$, follows for free.

However, proving that the functionals $\beta_\pm$ are non-zero reduces
to demonstrating (separately for ``$-$'' and ``$+$'') that the corresponding $\tilde{\beta}_\pm$ are linearly independent of the $\psi^\pm_{U^\pm}$ and furthermore (in case (2') only), to checking that the $\tilde{\beta}_\pm$ do not vanish identically on
$\bigcap_{U^\cap} \Kern \psi^\cap_{U^\cap}$.  After splitting the $(w_1 , w_2$) according to \eqnref{eqnSplit}, the arguments establishing these results are identical to those presented in cases (1) and (1'), so we do not repeat them here.
\end{proof}

We close this section with a couple of examples illustrating the formalism constructed above.  The first illustrates a simple case in which there are two invariant coordinates $\beta_\pm$.

\begin{example} \label{ex:0,V0,V5}
By \thmref{thm: moduli space right Verma}, there is a two-dimensional space
of isomorphism classes of staggered modules $\StagMod$ with $c=0$
($t=\tfrac{3}{2}$)  and short exact sequence
\begin{equation} \label{eqnSES4}
\dses{\VerMod_0}{}{\StagMod}{}{\VerMod_5},
\end{equation}
because $\LMod = \VerMod_0$ is of braid type and its grade $\ell = 5$
\sv{} $\omega_0$ has rank $2$ (this is a case \textup{(2)} example). 
The dimensionality of $\data' / \gauge'$ can also be demonstrated directly 
as follows.

The normalised rank $1$ \svs{} generating the submodule $\minsubmod$ of
$\LMod$ are
\begin{equation}
L_{-1} x \qquad \text{and} \qquad \brac{L_{-1}^2 - \tfrac{2}{3} L_{-2}} x.
\end{equation}
This example is rather special because the only states of $\LMod$ not
in $\minsubmod$ are $x$ and its (non-zero) multiples
(the irreducible \hwm{} $\IrrMod_0$ is one-dimensional).
It follows that $\data' = \data \cap \minsubmod = \data$.  
Since $\ell = 5$, we should check the constraint on the possible data 
$\brac{\omega_1 , \omega_2} \in \minsubmod_4 \oplus \minsubmod_3$ coming
from the non-trivial element of $\brac{\uea^+ L_1 \cap \uea^+ L_2}_{-5}$
given in \eqnref{eqnGen5}:
\begin{equation}
\brac{L_1^2 L_2 + 6 L_2^2 - L_1 L_3 + 2 L_4} \omega_1 =
    \brac{L_1^3 + 6 L_1 L_2 + 12 L_3} \omega_2.
\end{equation}
However, both sides must be proportional to $x \notin \minsubmod$, hence
must vanish for all $\omega_1$ and $\omega_2$.  There is therefore no
constraint upon the data.

Since $\dim \minsubmod_4 = 5$ and $\dim \minsubmod_3 = 3$, the space of
admissible data has dimension $8$.  As the space of gauge transformations
$\gauge' = \gauge$  has dimension $\dim \minsubmod_5 - 1 = 6$, we conclude
that the space of inequivalent staggered modules with exact sequence
(\ref{eqnSES4}) is two-dimensional, as expected.  Finally, as $\omega_0$ 
may be represented in the forms
\begin{subequations}
\begin{align}
\omega_0 &= \brac{L_{-1}^4 - \tfrac{20}{3} L_{-2} L_{-1}^2 + 4 L_{-2}^2 + 4 L_{-3} L_{-1} - 4 L_{-4}} L_{-1} x \\
&= \brac{L_{-1}^3 - 6 L_{-2} L_{-1} + 6 L_{-3}} \brac{L_{-1}^2 - \tfrac{2}{3} L_{-2}} x,
\end{align}
\end{subequations}
it follows from \eqnref{eq: invariant definition 2} and \thmref{thm: invariants} that this space is parametrised by two invariants:
\begin{subequations}
\begin{align}
\beta_- L_{-1} x &= \brac{L_1^4 - \tfrac{20}{3} L_1^2 L_2 + 4 L_2^2 + 4 L_1 L_3 - 4 L_4} y \\
\text{and} \qquad \beta_+ \brac{L_{-1}^2 - \tfrac{2}{3} L_{-2}} x &= \brac{L_1^3 - 6 L_1 L_2 + 6 L_3} y \pmod{\bC L_{-1}^2 x}.
\end{align}
\end{subequations}
Any choice of values for these beta-invariants corresponds to a distinct
staggered module.
\end{example}

This example is admittedly special, because $\minsubmod$ coincides with $\LMod$ at all positive grades.  One consequence is that both $\beta_-$ and $\beta_+$ are defined for all $(\omega_1,\omega_2) \in \data$ and are invariant under the full group of gauge transformations $\gauge$.  In general however, this is not true.  Practically, the beta-invariants may be viewed as numbers to be computed in order to identify representations.  It is therefore somewhat inconvenient that they are in general only defined for data $(\omega_1,\omega_2) \in \data'$, hence for only certain choices of $y$, and are consequently only invariant under the restricted set of gauge transformations $\gauge' \subseteq \gauge$ preserving $\data'$.

While the Projection Lemma, \lemref{lemProjection}, guarantees that we can always choose (equivalent) data in $\data'$, it is sometimes desirable to define the invariants so that one can easily compute them for general data $(\omega_1,\omega_2) \in \data$, and hence for general choices of $y$.  In the following example, we illustrate how to combine the content of the Projection Lemma with the above definitions of the beta-invariants to deduce a generally valid formula.

\begin{example} \label{ex:-2,V0,V3}
We consider the one-dimensional space of $c=-2$ ($t=2$) staggered modules with short exact sequence
\begin{equation} \label{eqnSES5}
\dses{\VerMod_0}{}{\StagMod}{}{\VerMod_3},
\end{equation}
$\LMod$ is of chain type with \svs{} $L_{-1} x$ and $\omega_0 = \brac{L_{-1}^2 - 2 L_{-2}} L_{-1} x$ at grades $1$ and $3$ respectively.  $\omega_0$ is therefore composite, of rank $2$, and $\minsubmod$ is generated by $L_{-1} x$.
We note first of all, supposing that $y$ is chosen such that
$(\omega_1,\omega_2) \in \data'$, that the invariant $\beta$ of \eqnref{eq: invariant definition 1} may be defined by
\begin{equation} \label{eqnCrapDef}
\beta L_{-1} x = \brac{L_1^2 - 2 L_2} y.
\end{equation}
Our aim is to derive a similar formula that can be used with any choice
of $y$ (assuming only that it is correctly normalised).

To do this, we recall that in the proof of the Projection Lemma, we constructed projections onto appropriate submodules of $\LMod$ which take data to equivalent data.  This was achieved by considering an orthonormal basis of the complement of the submodule (as a vector space) at the right grade.  In the case at hand, we only need one projection to get from $\data$ to $\data'$, the submodule we want to project onto is $\minsubmod$, and the grade of our basis is $\ell = 3$.  Since $\dim \minsubmod_3 = 2$ and $\dim \LMod_3 = 3$, we may take $Z = \tfrac{\ii}{2} L_{-3}$ to define our orthonormal basis $\tset{Z x}$ (recall that the Shapovalov form is assumed bilinear, not sesquilinear). On the other hand, the vectors $L_{-1}^3 x$ and $L_{-2}L_{-1} x$ span $\minsubmod$ at
grade $3$ (they are the $V_\lambda L_{-1} x$ in the notation of \secref{sec: projections}).

Given data in $\data$, the key step in the proof of the Projection Lemma was to find equivalent data in $\data'$ using a carefully chosen gauge transformation $g_z$.  In the case at hand, one can check that the choice amounts to $z = -Z Z^\dagger y$.  In terms of gauge-transforming $y$, this corresponds to applying the operator $\wun - Z Z^\dagger$ to obtain the new $y$ (for which the corresponding data is in $\data'$).  This immediately yields an improved version of the definition (\ref{eqnCrapDef}) of $\beta$,
\begin{equation}
\beta L_{-1} x = \bigl( L_1^2 - 2 L_2 \bigr) \bigl( 1 - Z Z^\dagger \bigr) y = \bigl( L_1^2 - 2 L_2 \bigr) \bigl( 1 + \tfrac{1}{4} L_{-3} L_3 \bigr) y,
\end{equation}
which may be used for any (admissible) choice of $y$.
\end{example}

It should be clear that the same strategy will recover formulae for the
beta-invariants of general staggered modules (with exact sequence
(\ref{eq: exact seq Verma})) which are valid for every $y$ corresponding
to admissible data.  All that will change is that the orthonormal basis may
consist of several elements $Z_\mu$, and that one might need several
consecutive projections.
Indeed, in the chain case we let $Z^{(k)}_\mu$ denote the basis elements
chosen at the $k$-th step of the projections
of \secref{sec: projections} and define
\begin{equation} \label{eqnDefPChain}
\wun - \mathbb{P} =
    \Bigl( \wun - \sum_{\mu} Z^{(\rnkO-1)}_{\mu}
        \bigl( Z^{(\rnkO-1)}_{\mu} \bigr)^\dagger \Bigr) \cdots
        \Bigl( \wun - \sum_{\mu} Z^{(2)}_{\mu}
        \bigl( Z^{(2)}_{\mu} \bigr)^\dagger \Bigr)
        \Bigl( \wun - \sum_{\mu} Z^{(1)}_{\mu}
        \bigl( Z^{(1)}_{\mu} \bigr)^\dagger \Bigr).
\end{equation}
The formula defining the invariant now becomes
\begin{equation} \label{eqnDefBetaChain}
\beta X_{\rnkO-1} x = \chi^\dagger \brac{\wun - \mathbb{P}} y.
\end{equation}
In the braid case, projecting from rank $k-1$ to $k$ required two steps
and we will denote the corresponding orthonormal bases by $Z^{(k-1;+)}_\mu$
and $Z^{(k;-)}_\mu$. Now,
\begin{multline} \label{eqnDefPBraid}
\wun - \mathbb{P} =
    \Bigl( \wun - \sum_{\mu} Z^{(\rnkO-1;-)}_{\mu}
      \bigl( Z^{(\rnkO-1;-)}_{\mu} \bigr)^\dagger \Bigr)
      \Bigl( \wun - \sum_{\mu} Z^{(\rnkO-2;+)}_{\mu}
      \bigl( Z^{(\rnkO-2;+)}_{\mu} \bigr)^\dagger \Bigr)
      \Bigl( \wun - \sum_{\mu} Z^{(\rnkO-2;-)}_{\mu}
      \bigl( Z^{(\rnkO-2;-)}_{\mu} \bigr)^\dagger \Bigr) \\
\cdots \Bigl( \wun - \sum_{\mu} Z^{(2;-)}_{\mu}
      \bigl( Z^{(2;-)}_{\mu} \bigr)^\dagger \Bigr)
      \Bigl( \wun - \sum_{\mu} Z^{(1;+)}_{\mu}
      \bigl( Z^{(1;+)}_{\mu} \bigr)^\dagger \Bigr)
      \Bigl( \wun - \sum_{\mu} Z^{(1;-)}_{\mu}
      \bigl( Z^{(1;-)}_{\mu} \bigr)^\dagger \Bigr),
\end{multline}
and the invariants are defined by
\begin{equation} \label{eqnDefBetaBraid}
\beta_- X_{\rnkO-1}^- x = (\chi_-)^\dagger \brac{\wun - \mathbb{P}} y 
\qquad \text{and} \qquad 
\beta_+ X_{\rnkO-1}^+ x = (\chi_+)^\dagger \brac{\wun - \mathbb{P}} y \pmod{\uea X_{\rnkO-1}^- x}.
\end{equation}

As a final simplification in such formulae, we can even remove the
inconvenient quotient in the definition of $\beta_+$.
Specifically, we can choose one more basis $\set{W_\mu}$, this time for
$\uea^-_{\ell_{\rnkO-1}^+ - \ell_{\rnkO-1}^-}$,
such that the corresponding basis
$\bigl\{ W_{\mu} v_{h^\lft + \ell_{\rnkO-1}^-} \bigr\}$ of
$\Verma_{h^\lft + \ell_{\rnkO-1}^-}$ is orthonormal.
Then for any $u \in \minsubmod_{\ell_{\rnkO-1}^+}$, the vector
$\sum_\mu W_\mu W_{\mu}^\dagger u$ is in $\uea X_{\rnkO-1}^- x$, and
$u - \sum_\mu W_\mu W_{\mu}^\dagger u$ is proportional to $X_{\rnkO-1}^+ x$
by virtue of orthonormality.
A completely explicit formula for $\beta_+$ is thus
\begin{equation}
\beta_+ X_{\rnkO-1}^+ x =
(\wun - \sum_\mu W_\mu W_\mu^\dagger)(\chi_+)^\dagger \brac{\wun - \mathbb{P}} y.
\end{equation}

We make one final remark about this way of defining invariants.
The explicit forms of the projections $\wun - \mathbb{P} \in \uea$ will
in general depend upon the choices of orthonormal bases used. 
However, the values taken by the beta-invariants of course do not.

\section{General Right Modules} \label{sec: general case}

In view of the general construction of \thmref{thmConstruction}, the existence question in the the case of arbitrary right modules $\Hrgt$ seems at first to be more involved than the $\Hrgt$ Verma case elucidated in \secref{sec: right Verma}.  The vectors $(\varpi, -\overline{X})$ that are added to the list of generators of $\sN$ (\eqnref{eqnNGens}) contain both $\varpi$, which is only determined by the data in a rather indirect fashion, and $\overline{X}$, for which no simple, explicit general formula is known.  However, \propref{prop: monotonicity right} suggests an alternate strategy.  Indeed, given left and right modules $\Hlft$ and $\Hrgt$, we can first use \thmref{thm: moduli space right Verma} to determine the space of isomorphism classes of staggered modules
$\check{\Stagg}$ with exact sequence
\begin{equation} \label{eqnSES6}
\dses{\Hlft}{\check{\iota}}{\check{\Stagg}}{\check{\pi}}{\Verma_{h^\rgt}},
\end{equation}
and then for each isomorphism class, decide whether the right module can be replaced by $\Hrgt$, obtaining
\begin{equation} \label{eqnSES7}
\dses{\Hlft}{\iota}{\Stagg}{\pi}{\Hrgt}.
\end{equation}
In practise, the isomorphism classes are determined by the beta-invariants of $\check{\Stagg}$, so our task in this section is to analyse, in terms of these coordinates, when such a replacement is permitted.  Throughout this section, we will assume that the right module $\Hrgt$ of $\Stagg$ is not a Verma module.

Note that the definitions of $\data$, $\gauge$, $\data'$, $\gauge'$ and $\minsubmod$ depend only upon the left module $\Hlft$ which is unchanged in the replacement proposed above.  We will therefore continue to use these notations in this section without comment.  Similarly, the important definitions
(\eqnDref{eq: invariant definition 1}{eq: invariant definition 2}) of $\beta$ and $\beta_\pm$ make perfect sense for $\Stagg$.  Indeed, since the data of $\Stagg$ and $\check{\Stagg}$ coincide by \propref{prop: monotonicity right}, it follows that their beta-invariants coincide too.  We therefore obtain, as an immediate consequence of \thmDref{thm: moduli space right Verma}{thm: invariants}, a uniqueness result covering every case except that which was already treated in \corref{cor: uniqueness l 0}.

\begin{corollary}
There exists at most one staggered module $\Stagg$ (up to isomorphism)
for any given choice of left and right modules and beta-invariants $\beta$
or $\beta_\pm$ of \secref{sec: invariants} (as appropriate).
\end{corollary}

\subsection{Singular Vectors of Staggered Modules} \label{secStagModSVs}

It was observed in \corref{corItsASubmodule} that $\Stagg$
may be realised as a quotient of $\check{\Stagg}$.  We give in this section a sharpening of this result.  First however, we recall from \propref{propVanSVs} that when $\overline{X} \in \uea^-_{\overline{\ell}}$ is defined, it is necessary for the existence of $\Stagg$ that $\overline{X} \omega_0 = 0$ in $\LMod$.  A similar statement holds if both $\overline{X}^- \in \uea^-_{\overline{\ell}^-}$ and $\overline{X}^+ \in \uea^-_{\overline{\ell}^+}$ are defined.  We therefore assume in what follows that $\LMod$ satisfies this requirement.

\begin{proposition} \label{prop: singular vectors in staggered modules}
When $\overline{X}$ is defined, a staggered module $\Stagg$ exists if and only if $\check{\Stagg}$ has a \sv{} $\overline{y}$ at grade $\ell + \overline{\ell}$.  Then, $\Stagg = \left. \check{\Stagg} \middle/ \uea \overline{y} \right.$.  When $\overline{X}^-$ and $\overline{X}^+$ are defined, $\Stagg$ exists if and only if $\check{\Stagg}$ has \svs{} $\overline{y}^-$ and $\overline{y}^+$ at grades $\ell + \overline{\ell}^-$ and $\ell + \overline{\ell}^+$, respectively.  Then, $\Stagg = \left. \check{\Stagg} \middle/ (\uea \overline{y}^- + \uea \overline{y}^+) \right.$.
\end{proposition}

We remark immediately that by \propref{propVanSVs}, the left module $\Hlft$
does not have a (non-zero) singular vector at grade $\ell+\overline{\ell}$, so the \svs{} $\overline{y}$ or $\overline{y}^\pm$ in $\check{\Stagg}$ are not annihilated by the projection onto $\Verma_{h^\rgt}$.  Indeed, we may assume the normalisations
\begin{equation} \label{eqnDefxbar}
\check{\pi}(\overline{y}) = \overline{X} v_{h^\rgt} \qquad \text{or} \qquad \check{\pi}(\overline{y}^\pm) = \overline{X}^\pm v_{h^\rgt}.
\end{equation}
The \svs{} therefore have the form $\overline{X} y - \varpi$ or $\overline{X}^\pm y - \varpi^\pm$, where $\varpi, \varpi^\pm \in \LMod$.  The uniqueness of such \svs{} follows again from \propref{propVanSVs}.  We mention that it is in considering situations such as these that the terminology employed by Rohsiepe in \cite{RohRed96} becomes inconvenient.  In particular, we see once again that for $\check{\Stagg}$, Rohsiepe's lower module, which he defines as the subspace of $L_0$ eigenvectors, is not a \hwm{} (it contains $\overline{y}$).

\begin{proof}
We first assume that $\Stagg$ exists. Denote by $\check{\sN}$ and $\sN$
the submodules of $\Hlft \oplus \uea$
in the constructions (\thmref{thmConstruction}) of $\check{\Stagg}$ and $\Stagg$
respectively. As we have seen in the proof of \propref{prop: monotonicity right}, $\check{\sN} \subseteq \sN$. We will show that
$L_n \left(\varpi, -\overline{X}\right) \in \check{\sN}$ for all $n>0$ and
$\left(L_0 - h^\rgt - \overline{\ell}\right) \left(\varpi, -\overline{X}\right) \in \check{\sN}$,
thereby establishing that $\left(\varpi, -\overline{X}\right)$ becomes singular in the
quotient $\left. \left(\Hlft \oplus \uea\right) \middle/ \check{\sN} \right. = \check{\Stagg}$ (we will only detail this direction in the $\overline{X}$ case, that of $\overline{X}^\pm$ being identical).

We first write $L_n \overline{X} = U_0 \left(L_0 - h^\rgt\right) + U_1 L_1 + U_2 L_2$, as usual because $\overline{X}$ is singular. Then in $\LMod \oplus \uea$, the definition of $\check{\sN}$ gives
\begin{align}
L_n \left(\varpi, -\overline{X}\right) &= L_n \left(\varpi,0\right) - U_0 \left(0, L_0-h^\rgt\right) - U_1 \left(0, L_1\right) - U_2 \left(0, L_2\right) \notag \\
&= L_n \left(\varpi,0\right) - U_0 \left(\omega_0,0\right) - U_1 \left(\omega_1,0\right) - U_2 \left(\omega_2,0\right) \pmod{\check{\sN}},
\end{align}
and each of the four terms on the right hand side are obviously in $\iota^\lft \left(\Hlft\right) = \Hlft \oplus \set{0}$.  Now, $\left(\varpi, -\overline{X}\right)$ is one of the generators of $\sN$, so the sum of the four terms is in $\sN$ (since $\check{\sN} \subseteq \sN$).  But the existence of $\Stagg$ implies that $\sN \cap \iota^\lft\left(\Hlft\right) = \sN^\circ = \set{0}$ by \thmref{thmConstruction}, hence that the sum of these four terms is zero.  We conclude that 
$L_n \left(\varpi, -\overline{X}\right) \in \check{\sN}$ as required.

The argument for $\left(L_0 - h^\rgt - \overline{\ell}\right)$ is similar.  In fact,
we have $\left(L_0-h^\rgt-\overline{\ell}\right) \overline{X} = \overline{X} \left(L_0 - h^\rgt\right)$, so
\begin{equation}
\left(L_0 - h^\rgt - \overline{\ell}\right) \left(\varpi, -\overline{X}\right) = -\overline{X} \left(0, L_0-h^\rgt\right) = -\overline{X} \left(\omega_0,0\right) \pmod{\check{\sN}}.
\end{equation}
But we are assuming that $\overline{X} \omega_0 = 0$ (\propref{propVanSVs}), hence we find that $\left(L_0 - h^\rgt - \overline{\ell}\right) \left(\varpi, -\overline{X}\right) \in \check{\sN}$, as required.  This completes the proof that the class of $\left(\varpi, -\overline{X}\right)$ modulo $\check{\sN}$ is a singular vector of $\check{\Stagg}$.

The other direction requires us to show that the existence of the singular vector $\overline{y} \in \check{\Stagg}$ implies that the quotient
by the submodule generated by $\overline{y}$ is the desired staggered
module $\Stagg$.  The strategy here is rather similar to that used to prove \thmref{thmConstruction}.  First observe from \eqnref{eqnDefxbar} that $\check{\pi} \left(\uea \overline{y}\right) = \sJ$, where $\Hrgt = \Verma_{h^\rgt} / \sJ$. We denote the projections from $\check{\Stagg}$ to $\left. \check{\Stagg} \middle/ \uea \overline{y} \right.$ and from $\Verma_{h^\rgt}$ to $\Hrgt$ by $\overline{\pi}$ and $\pi_{\sJ}$, respectively.  We then define module homomorphisms $\iota$ and $\pi$ so as to make the following diagram commutative:
\begin{equation} \label{eqnDiagram2}
\begin{CD}
0 @>>> \LMod @>\check{\iota}>> \check{\StagMod} @>\check{\pi}>> \VerMod_{\RDim} @>>> 0 \\
@. @| @VV\overline{\pi}V @VV\pi_{\mathcal{J}}V @. \\
0 @>>> \LMod @>\iota>> {\displaystyle \frac{\check{\StagMod}}{\uea \overline{y}}} @>\pi>> \RMod @>>> 0
\end{CD}
\mspace{10mu}.
\end{equation}

Our task is now to show that the bottom row is exact.  For injectivity of $\iota = \overline{\pi} \circ \check{\iota}$, we must show that $\Kern \overline{\pi} \cap \Imag \check{\iota} = \uea \overline{y} \: \cap \: \check{\iota}(\Hlft) = \set{0}$.  But, if $U \overline{y} \in \check{\iota}(\Hlft)$ for some $U \in \uea$, then we can assume that $U \in \uea^-$ by the singularity of $\overline{y}$.  Exactness of the top row now gives $0 = \check{\pi} (U \overline{y}) = U \overline{X} v_{\RDim} \in \VerMod_{\RDim}$, hence $U = 0$ as Verma modules are free as $\uea^-$-modules.  This proves that $\iota$ is injective.  The projection $\pi$ is well-defined by $\pi \circ \overline{\pi} = \pi_\sJ \circ \check{\pi}$ because $\Kern \overline{\pi}$ is annihilated by the right hand side by construction.  Its surjectivity follows from that of $\pi_\sJ$ and $\check{\pi}$.

Exactness then follows from that of the top row, whence $\pi \circ \iota = \pi \circ \overline{\pi} \circ \check{\iota} = \pi_\sJ \circ \check{\pi} \circ \check{\iota} = 0$, and the following argument:  If $\pi \circ \overline{\pi} (u) = 0$ for some $u \in \check{\StagMod}$, then $\pi_\sJ \circ \check{\pi} (u) = 0$, hence $\check{\pi} (u) = U \overline{X} v_{\RDim}$ for some $U \in \uea$.  We therefore conclude that
\begin{equation}
u = U \overline{y} \pmod{\check{\iota} (\LMod)}, \qquad \text{so} \qquad \overline{\pi} (u) = 0 \pmod{\overline{\pi} \check{\iota} (\LMod) = \iota (\LMod)}.
\end{equation}
As $\overline{\pi}$ is surjective, we are done.

We have therefore constructed an exact sequence with the left and right modules of $\StagMod$.  The data of $\left. \check{\Stagg} \middle/ \uea \overline{y} \right.$ is obtained by acting on $y = \overline{\pi} (\check{y})$ (which is indeed mapped to $\Rhws \in \RMod$ under $\pi$), and coincides with that of $\StagMod$ (and $\check{\StagMod}$).  By \propref{prop: equivalence}, $\StagMod \cong \left. \check{\Stagg} \middle/ \uea \overline{y} \right.$ as required.

The argument for the $\overline{X}^\pm$ case is similar, though slightly
more complicated.  First we construct an exact sequence for
$\left. \check{\Stagg} \middle/ \uea \overline{y}^- \right.$ as above
(with injection $\iota'$ and surjection $\pi'$), obtaining a commutative
diagram very similar to (\ref{eqnDiagram2}).  The arguments for this step
are exactly the same as those above.  Then, we define $\iota$ and $\pi$
so as to make the following augmented diagram commute:
\begin{equation} \label{eqnDiagram3}
\begin{CD}
0 @>>> \LMod @>\check{\iota}>> \check{\StagMod} @>\check{\pi}>> \VerMod_{\RDim} @>>> 0 \\
@. @| @VV\overline{\pi}^-V @VV\pi_{\mathcal{J}^-}V @. \\
0 @>>> \LMod @>\iota'>> {\displaystyle \frac{\check{\StagMod}}{\uea \overline{y}^-}} @>\pi'>> {\displaystyle \frac{\VerMod_{\RDim}}{\sJ^-}} @>>> 0 \\
@. @| @VV\overline{\pi}^+V @VV\pi_{\mathcal{J}^+}V @. \\
0 @>>> \LMod @>\iota>> {\displaystyle \frac{\check{\StagMod}}{\uea \overline{y}^- + \uea \overline{y}^+}} @>\pi>> \RMod @>>> 0
\end{CD}
\mspace{10mu}.
\end{equation}
Here, $\overline{\pi}^\pm$ corresponds to quotienting by the submodule
generated by $\overline{y}^\pm$ and $\pi_{\mathcal{J}^\pm}$ corresponds to
quotienting by the submodule $\mathcal{J}^\pm$ generated by
$\overline{X}^- v_{\RDim}$ or $\overline{X}^+ x'$ as appropriate, where
$x'$ is the \hws{} of $\left. \VerMod_{\RDim} \middle/ \sJ^- \right.$. 
The arguments demonstrating exactness and the isomorphism of $\StagMod$ and
$\left. \check{\Stagg} \middle/
\left( \uea \overline{y}^- + \uea \overline{y}^+ \right) \right.$
are also identical to those above, except as regards the proof that
$\iota$ is injective.

Note then that $\iota$ will be injective if the only
$U \in \uea^-$ for which $U \overline{y}^+ \in \iota' (\LMod)$ are such that
$U \overline{y}^+ \in \uea \overline{y}^-$.
But applying $\pi'$ and using the exactness
of the middle row of (\ref{eqnDiagram3}) gives
$U \overline{X}^+ x' = 0$.  The module generated by the \hws{} $x'$ is not
Verma, so we can only conclude that $U = V^- \chi^-_+ + V^+ \chi^+_+$ for
some $V^\pm \in \uea^-$, where the $\chi^\pm_+ \overline{X}^+ v_{\RDim}$ 
denote the (normalised) \svs{} in $\VerMod_{\RDim}$ whose rank is one higher
than that of $\overline{X}^+ v_{\RDim}$ (in particular, the $\chi^\pm_+$ are
singular).  Thus,
\begin{equation} \label{eqnFinally}
U \overline{y}^+ = V^- \chi^-_+ \overline{y}^+ + V^+ \chi^+_+ \overline{y}^+.
\end{equation}

We use the fact that $\chi^\pm_+ \overline{X}^+ = \chi^\pm_- \overline{X}^-$
for some singular $\chi^\pm_- \in \uea^-$ (which follows from the Feigin-Fuchs
classification of singular vectors in Verma modules).
It is easily verified that
$\chi^\pm_+ \overline{y}^+ - \chi^\pm_- \overline{y}^- \in \check{\Stagg}$
are singular vectors and furthermore that they are in
$\Kern \check{\pi} = \Imag \check{\iota}$. By \propref{propVanSVs},
we then have $\chi^\pm_+ \overline{y}^+ = \chi^\pm_- \overline{y}^-$.
Substituting into \eqnref{eqnFinally}, we therefore see that
the vector $U \overline{y}^+ \in \check{\Stagg}$ must be in
the submodule generated by $\overline{y}^-$, establishing the injectivity of
$\iota$.
\end{proof}

This result validates the practical technique proposed in \cite{RidLog07} to find constraints on the beta-invariant of a staggered module $\StagMod$ by searching for \svs{} in the corresponding $\check{\StagMod}$.  The power of \propref{prop: singular vectors in staggered modules}, when combined with the classification of \thmref{thm: moduli space right Verma}, is evidenced by the following examples.

\begin{example} \label{ex: comparison of similar cases}
We are finally ready to demonstrate the claims made in
\exDref{ex:-2,13,15}{ex:0,12,14} concerning the allowed
values of $\beta$.  In the former case, the staggered
module $\StagMod$ had $c=-2$ ($t=2$),
$\Hlft=\Verma_0/\Verma_3$ and $\Hrgt=\Verma_1/\Verma_6$.
 By \thmref{thm: moduli space right Verma}, there is a
one-dimensional space of staggered modules $\check{\StagMod}$
with the same left module but $\RMod = \VerMod_1$, parametrised
by $\beta$ (\thmref{thm: invariants}).  We search in
$\check{\StagMod}$ for a \sv{} at grade $6$, finding
one for \emph{every} $\beta \in \CC$\textup{:}
\begin{multline} \label{eqnUglySV}
\overline{y} = \brac{L_{-1}^3 - 8 L_{-2} L_{-1} + 12 L_{-3}} \brac{L_{-1}^2 - 2 L_{-2}} y - \Bigl( -\tfrac{16}{3} \brac{\beta + 1} L_{-2}^2 L_{-1}^2 + \tfrac{4}{3} \brac{14 \beta + 5} L_{-3} L_{-2} L_{-1} \Bigr. \\
\Bigl. - 6 \beta L_{-3}^2 - 6 \brac{\beta - 2} L_{-4} L_{-1}^2 + 8 \beta L_{-4} L_{-2} - \tfrac{2}{3} \brac{5 \beta + 2} L_{-5} L_{-1} + 4 \beta L_{-6} \Bigr) x.
\end{multline}
Here, we have used
$\brac{L_{-1}^2 - 2 L_{-2}} L_{-1} x = 0$
to eliminate terms of the form $\uea^- L_{-1}^3 x$.\footnote{This is nothing but the requirement $\overline{X} \omega_0 = 0$ of \propref{propVanSVs}.
Combined with the observation that the gauge freedom
here is trivial, $\LMod_1 = \bC \omega_0$, this also explains why
\eqnref{eqnUglySV} is valid independent of the choice of $y$.}
It now follows from \propref{prop: singular vectors in staggered
modules} that there also exists a one-dimensional space
of staggered modules $\StagMod$ (likewise parametrised
by $\beta$) with the desired left and right modules.

The case of \exref{ex:0,12,14} is different. The staggered module $\StagMod$ had $c=0$ ($t=\frac{3}{2}$), $\Hlft=\Verma_0/\Verma_2$ and $\Hrgt=\Verma_1/\Verma_5$.  Searching for a grade $5$ \sv{} in the $\check{\StagMod}$ (with unknown $\beta$), we find that a \sv{} exists \emph{if and only if} $\beta = -\tfrac{1}{2}$, in which case it has the form\textup{:}
\begin{equation}
\overline{y} = \brac{L_{-1}^4 - \tfrac{20}{3} L_{-2} L_{-1}^2 + 4 L_{-2}^2 + 4 L_{-3} L_{-1} - 4 L_{-4}} y - \brac{-\tfrac{32}{9} L_{-3} L_{-2} + \tfrac{16}{3} L_{-4} L_{-1} + 2 L_{-5}} x.
\end{equation}
Here, we have used $\brac{L_{-1}^2 - \tfrac{2}{3} L_{-2}} x = 0$ to eliminate terms of the form $\uea^- L_{-1}^2 x$.  \propref{prop: singular vectors in staggered modules} now states that there is a unique staggered module $\StagMod$ with the desired left and right modules, and that it has beta-invariant $\beta = -\tfrac{1}{2}$.
\end{example}

Whilst searching for \svs{} gives a useful general technique to determine how many staggered modules correspond to a given exact sequence, it is clear that this method is computationally intensive.  For instance, even the relatively simple module discussed in \exref{ex:0,14,18} requires searching for \svs{} at grade $14$, hence determining the form of $\varpi$ (when it exists) within a space of dimension $\dim \LMod_{14} = \tpartnum{14} - \tpartnum{10} = 93$.  Clearly, it would be very helpful to have stronger existence results, and it is these that we turn to now.

\subsection{Submodules and the Projection Lemma} \label{secRightModProj}

The previous section reduces the existence question for $\Stagg$ to
a question about \svs{} $\overline{y}$ (or $\overline{y}^\pm$) in
$\check{\Stagg}$.  We will first develop the idea of this section
in the case in which there is only one $\overline{X}$, briefly
noting afterwards the slight changes needed in the
$\overline{X}^\pm$ case.  Recall that these \svs{}
$\overline{y}$ necessarily take the form $\overline{X} y - \varpi$,
where $\varpi \in \LMod_{\ell + \overline{\ell}}$.  In searching for
these \svs{}, we are naturally led to consider the set of elements
obtained from $\overline{X} y$ through translating by an element
of $\LMod_{\ell + \overline{\ell}}$.  This translation is strongly
reminiscent of gauge transforming data, and it is this similarity
that we shall exploit in this section.

To make matters more transparent, let us consider instead of
$\check{\Stagg}$, a staggered module $\OthStagMod$ that differs 
only in that its left module is also Verma.  This does not
change the dimension of the space of isomorphism classes,
by \thmref{thm: moduli space right Verma}, and we have the
usual definitions of $x$, $y$, $\omega_0$, $\omega_1$,
$\omega_2$ and $\beta$ (or $\beta_\pm$).  However, this
slight change of viewpoint necessitates a reinterpretation
of the results of the previous section, because a Verma left
module obviously conflicts with the conclusion of
\propref{propVanSVs} when the right module is not
Verma (upon setting a \sv{} of $\Hrgt$ to zero).  Instead
of searching for \svs{} of the form $\overline{X} y - \varpi$,
we will therefore instead consider the submodules of
$\OthStagMod$ generated by the $\overline{X} y - u$,
where $u$ ranges over
$\brac{\VerMod_{\LDim}}_{\ell + \overline{\ell}}
\subset \OthStagMod$.

More precisely, let us consider the submodules
$\func{\overline{\OthStagMod}}{u} \subset \OthStagMod$ which are generated
by $x$ and $\overline{X} y - u$.  Because we have insisted that the left
module is Verma (and this is why we are insisting upon this in
the first place), these are all staggered modules with exact sequence
\begin{equation}
\dses{\VerMod_{\LDim}}{}{\func{\overline{\OthStagMod}}{u}}{}{\VerMod_{\RDim + \overline{\ell}}}.
\end{equation}
Indeed, putting $\overline{y} = \overline{X} y - u$, we define in the usual
way $\overline{\omega}_0 = \brac{L_0 - \RDim - \overline{\ell}} \overline{y}
= \overline{X} \omega_0$, $\overline{\omega}_1 = L_1 \overline{y}$,
$\overline{\omega}_2 = L_2 \overline{y}$, and thence $\overline{\beta}$
(or $\overline{\beta}_\pm$) by \eqnref{eq: invariant definition 1}
(or (\ref{eq: invariant definition 2})).  Varying 
$u \in \brac{\VerMod_{\LDim}}_{\ell + \overline{\ell}}$ then really
does amount to performing gauge transformations on any given representative,
$\func{\overline{\OthStagMod}}{0}$ say. 
In particular, all the $\func{\overline{\OthStagMod}}{u}$ are isomorphic.

Apply now the Projection Lemma, \lemref{lemProjection}, to the staggered
module $\func{\overline{\OthStagMod}}{0}$.  This tells us that we can always make a gauge transformation so that the transformed data
$\brac{\overline{\omega}_1' , \overline{\omega}_2'}$ belongs to the
submodule $\overline{\minsubmod}$ of $\VerMod_{\LDim}$ generated by the
\svs{} of rank $\rnkO + \overline{\rnkO} - 1$, where $\rnkO$ is the rank of
$\omega_0 = X x$ in $\VerMod_{\LDim}$ and $\overline{\rnkO}$ is the
rank of $\overline{X} v_{\RDim}$ in $\VerMod_{\RDim}$
(so $\rnkO+\overline{\rnkO}$
is the rank of $\overline{\omega}_0 \in \VerMod_{\LDim}$).  In other
words, there exists $\varpi \in \brac{\VerMod_{\LDim}}_{\ell + \overline{\ell}}$
such that $\vir^+ \brac{\overline{X} y - \varpi}
\subseteq \overline{\minsubmod}$.  The submodule
$\uea \brac{\overline{X} y - \varpi} \subset \OthStagMod$
is then a staggered module with left module $\overline{\minsubmod}$
(or even some submodule thereof), right module
$\VerMod_{\RDim + \overline{\ell}}$, and beta-invariant 
$\overline{\beta}$ (or $\overline{\beta}_\pm$).

Consider now the quotient of $\OthStagMod$ by
$\overline{\minsubmod} \subseteq \VerMod_{\LDim}$. If we assume that
$\omega_0 \notin \overline{\minsubmod}$ (which is equivalent to assuming
that $\overline{\rnkO} > 1$), then \propref{prop: monotonicity left}
tells us that this is a staggered module $\check{\StagMod}$ with exact
sequence
\begin{equation}
\dses{\left. \VerMod_{\LDim} \middle/ \overline{\minsubmod} \right.}{}{\check{\StagMod}}{}{\VerMod_{\RDim}}.
\end{equation}
Moreover, its beta-invariant is obviously the same as that of
$\OthStagMod$, namely $\beta$.  It should now be evident that
$\overline{y} = \overline{X} y - \varpi$ is a \sv{} of $\check{\StagMod}$,
so by \propref{prop: singular vectors in staggered modules}, we may 
construct a module $\StagMod =
\left. \check{\StagMod} \middle/ \uea \overline{y} \right.$ for each
beta-invariant $\beta$ whose exact sequence is
\begin{equation}
\dses{\left. \VerMod_{\LDim} \middle/ \overline{\minsubmod} \right.}{}{\StagMod}{}{\left. \VerMod_{\RDim} \middle/ \uea \overline{X} v_{\RDim} \right.}.
\end{equation}
We can even reduce the left module of $\StagMod$ further by quotienting
by any submodule not containing $\omega_0$.

It remains only to remark upon the differences in the $\overline{X}^\pm$
case.  We may apply the above formalism to consider separately the 
submodules $\func{\overline{\OthStagMod}^\pm}{u} \subset \OthStagMod$ 
which are generated by $x$ and $\overline{X}^\pm y - u$.  Applying the
Projection Lemma to each, we conclude that there exist $\varpi^\pm$ such
that $\vir^+ \bigl( \overline{X}^\pm y - \varpi^\pm \bigr)
\subseteq \overline{\minsubmod}$ for the submodule
$\overline{\minsubmod} \subseteq \VerMod_{\LDim}$ generated by the
rank $\rnkO + \overline{\rnkO}-1$ \svs{} (we emphasise that this is the
same submodule for both ``$-$'' and ``$+$''). 
The vectors $\overline{y}^\pm = \overline{X}^\pm y - \varpi^\pm$
are therefore both singular in the
quotient $\check{\StagMod} = \OthStagMod/\overline{\minsubmod}$, so an appeal
to \propref{prop: singular vectors in staggered modules} then settles
this case.  Putting this all together, we have proven the following result.

\begin{proposition} \label{prop: general non-critical}
Let $\rnkO$ and $\overline{\rnkO}$ denote the ranks of the \svs{}
$\omega_0 = X x \in \VerMod_{\LDim}$ and $\overline{X} v_{\RDim} \in
\VerMod_{\RDim}$.  If there are no (non-zero) \svs{} in $\Hlft$ of
rank $\rnkO+\overline{\rnkO}-1$, then the dimension of the space of
staggered modules $\Stagg$ with exact sequence (\ref{eqnSES7}) matches
the dimension of the space of staggered modules $\check{\Stagg}$ with
exact sequence (\ref{eqnSES6}).
\end{proposition}

\begin{example} \label{ex:-2,13,15 again}
This result allows us to understand why the exact sequence (\ref{eqnSES1})
of \exref{ex:-2,13,15} admits a one-parameter family of staggered modules. 
In \exref{ex: comparison of similar cases}, we proved that this was indeed
the case, but now we see it as a direct consequence of
\propref{prop: general non-critical}, and hence as a corollary of the
Projection Lemma.  To whit, the left module is $\VerMod_0 / \VerMod_3$
and the right module is $\VerMod_1 / \VerMod_6$ (see \figref{figEx12}
in \secref{secStag}).  The ranks of $\omega_0 = L_{-1} x$ and
$\overline{X} v_{\RDim}$ are $1$ and $2$ respectively, so that of
$\overline{\omega}_0$ is $\rnkO + \overline{\rnkO} = 3$.  But, there is
no (non-vanishing) rank $2$ \sv{} of $\LMod$ (it would have
dimension $3$), hence the proposition applies.

We note that the proposition does not apply to the exact sequence
(\ref{eqnSES2}) considered in \exref{ex:0,12,14}.  In this case,
the left module is $\VerMod_0 / \VerMod_2$ and the right module
is $\VerMod_1 / \VerMod_5$, so we find that $\rnkO + \overline{\rnkO} = 2$.
But there is a non-vanishing rank $1$ \sv{} in $\LMod$, namely $\omega_0$.
This failure to meet the hypotheses should be expected as we have already
shown (\exref{ex: comparison of similar cases}) that the dimension of the
space of $\StagMod$ differs from that of the corresponding
$\check{\StagMod}$.  We will therefore have to work harder to get
an intuitive understanding of why this is so (beyond a brute force
computation of \svs{}).
\end{example}

\begin{example} \label{ex:LM(1,p)}
In the study of the so-called $LM \brac{1,q}$ \lcfts{} \cite{GabInd96}, one encounters staggered modules $\StagMod_s$ with $c = 13 - 6 \bigl( q + q^{-1} \bigr)$ ($t = q$) and exact sequence
\begin{equation} \label{eqnSES8}
\dses{\QuotMod_{1,s}}{}{\StagMod_s}{}{\QuotMod_{1,s + 2 \brac{q - \sigma}}},
\end{equation}
where $\QuotMod_{1,s} = \VerMod_{h_{1,s}} / \VerMod_{h_{1,s} + s}$.  Here, $s$ is a positive integer not divisible by $q$, and $0 < \sigma < q$ is the remainder obtained upon dividing $s$ by $q$.  The left and right modules are of chain type, the former being irreducible if $s < q$ and reducible with \svs{} of ranks $0$ and $1$ if $s > q$.  The right module is always reducible with \svs{} of ranks $0$ and $1$.

We then have $\rnkO = 0$ when $s < q$ and $\rnkO = 1$ otherwise, $\ell = \brac{q - \sigma} \brac{s - \sigma} / q$, $\overline{\rnkO} = 2$ and $\overline{\ell} = s + 2 \brac{q - \sigma}$.  Since the left module has no \svs{} of rank $\rnkO + 1$, it follows from \propref{prop: general non-critical} and \thmref{thm: moduli space right Verma} that the exact sequence (\ref{eqnSES8}) describes a one-parameter family of staggered modules.  Identifying the staggered modules appearing in the $LM \brac{1,q}$ models therefore requires computing the corresponding beta-invariants.  Unfortunately, this has only been done for certain small $s$.
\end{example}

\propref{prop: general non-critical} states that if $\LMod$ has no
(non-zero) \svs{} of rank greater than or equal to $\rnkO+\overline{\rnkO}-1$,
then the existence question for staggered modules $\StagMod$ is
equivalent to the same question for the corresponding $\check{\StagMod}$. 
Moreover, \propref{propVanSVs} tells us that the left module $\LMod$ of
$\StagMod$ cannot have a \sv{} of rank $\rnkO+\overline{\rnkO}$ or
$\rnkO+\overline{\rnkO}+1$, according as to whether $\LMod$ is of chain or
braid type, respectively.  We have therefore solved the existence
question for staggered modules in all but a finite number of
outstanding \emph{critical rank} cases.  It is these cases that we now turn to.

\subsection{Existence at the Critical Ranks} \label{secCritRank}

If $\Hlft$ has non-zero singular vectors at the critical rank
$\rnkO+\overline{\rnkO}-1$, we can still follow the strategy of
\secref{secRightModProj} to try to construct $\StagMod$, but we cannot in general quotient away the full submodule $\overline{\minsubmod}$ without ending up with a left module smaller than $\LMod$.  We will therefore have to perform a more detailed analysis to determine when we can quotient by a smaller submodule.

For convenience, we will separate the outstanding cases according to the configurations of the \svs{} of $\LMod$ and $\RMod$ at the critical ranks.  We let $\nbrGamma \in \set{0,1,2}$ denote the number of rank $\rnkO+\overline{\rnkO}-1$ \svs{} of $\LMod$ and $\Rnull \in \set{0,1,2}$ denote the (minimal) number of \svs{} needed to generate $\sJ$, where $\RMod = \VerMod_{\RDim} / \sJ$.  The critical rank cases correspond to neither $\nbrGamma$ nor $\Rnull$ vanishing, so we have four \sv{} configurations which we illustrate in \figref{figCritRank}.  There, $\nbrGamma$ represents the number of black circles in the top row for $\LMod$ and $\Rnull$ represents the number of white circles in the bottom row for $\RMod$.  We label the critical rank cases by this pair of integers $(\nbrGamma,\Rnull)$.

{
\psfrag{1a}[][]{Case $(1,1)$}
\psfrag{1b}[][]{Case $(1,2)$}
\psfrag{2a}[][]{Case $(2,1)$}
\psfrag{2b}[][]{Case $(2,2)$}
\psfrag{B1}[][]{$\overline{\beta}$}
\psfrag{B2}[][]{$\overline{\beta}^\pm$}
\psfrag{B3}[][]{$\overline{\beta}_\pm$}
\psfrag{B4}[][]{$\overline{\beta}^\pm_\pm$}
\begin{figure}
\begin{center}
\includegraphics[width=10cm]{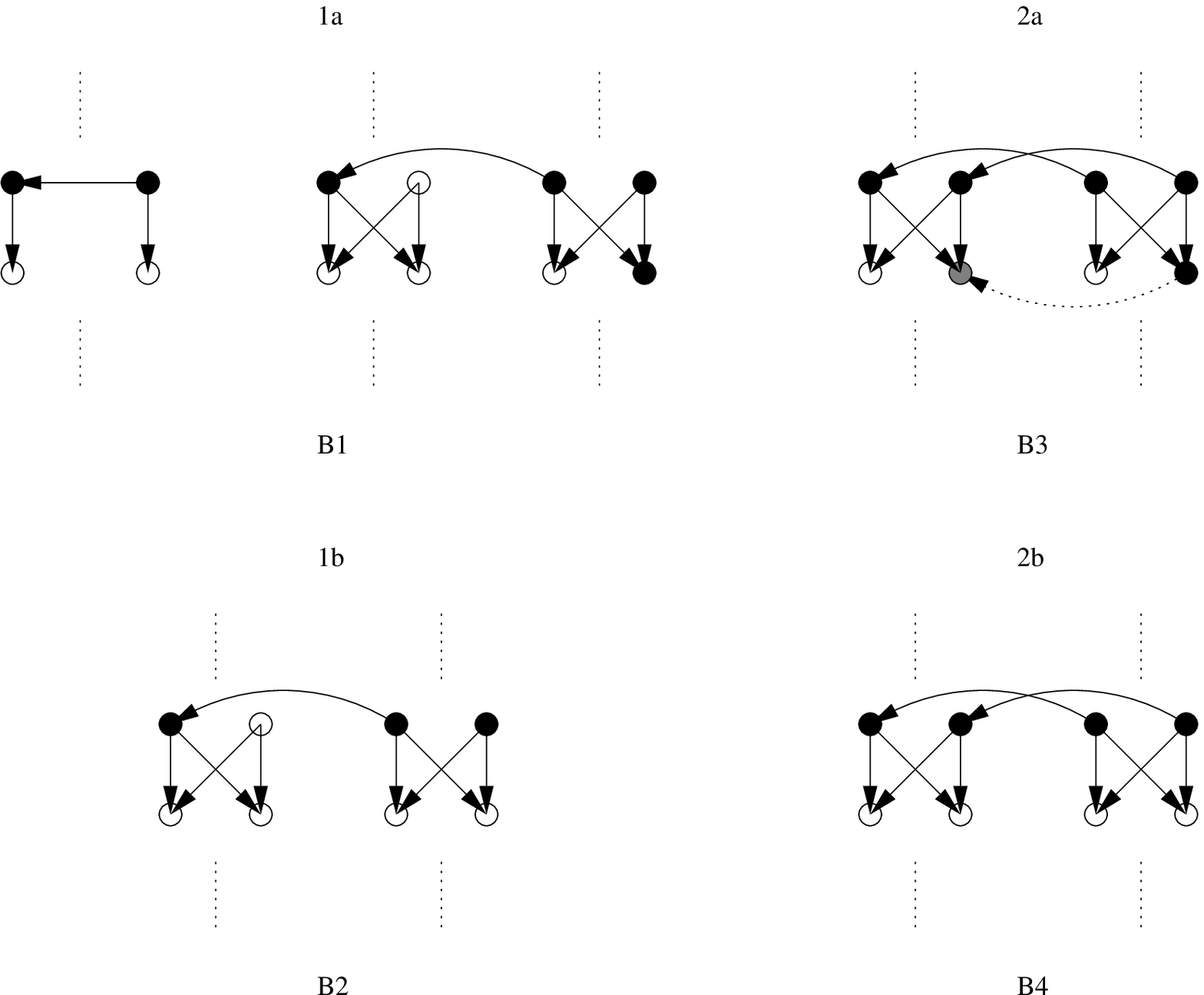}
\caption{The critical rank configurations for which \propDref{prop: general non-critical}{propVanSVs} are not sufficient to settle the existence question.
Pictured are the \svs{} of $\LMod$ of ranks
$\rnkO+\overline{\rnkO}-1$ and $\rnkO+\overline{\rnkO}$,
and their counterparts of ranks $\overline{\rnkO}-1$ and
$\overline{\rnkO}$ in $\RMod$.  Black indicates that the
\sv{} is present, white that it has been set to zero, and
grey that either possibility is admissible.  The (curved)
horizontal arrows indicate the non-diagonal action of $L_0$. 
It is understood that in certain circumstances, some of the
\svs{} pictured may not actually be present in the braid cases
(for example when $\overline{\rnkO} = 1$, $\RMod$ has only one
\sv{} of rank $\overline{\rnkO}-1 = 0$).  We also indicate for
each configuration the beta-invariants of \eqnref{eqnDefBetaBars}
whose vanishing is equivalent to the existence of the associated
staggered module.} \label{figCritRank}
\end{center}
\end{figure}
}

Let us first consider the case $(1,1)$ with modules of chain type for simplicity.  Recall that the
data of the module $\func{\overline{\OthStagMod}}{\varpi}$ was
denoted by $(\overline{\omega}_1, \overline{\omega}_2)$, where
$\varpi$ was chosen so that
$\overline{\omega}_j = L_j \overline{y} =
L_j \brac{\overline{X} y - \varpi} \in \overline{\minsubmod}$. 
Instead of quotienting $\OthStagMod$ by $\overline{\minsubmod}$,
we would now like to
quotient by the smaller submodule
$\uea \overline{\omega}_0 =
\uea \overline{X} X x \subset \overline{\minsubmod}$.
It is clear that $\overline{y}$ will become singular in the
quotient
if and only if $\overline{\omega}_j = 0$ for $j=1,2$. Of course,
we have the freedom of gauge transformations in choosing $\varpi$, so the
question should be whether $(\overline{\omega}_1, \overline{\omega}_2)$
is equivalent to $(0,0)$.  From this, we conclude that
$\overline{y} = \overline{X} y - \varpi$ will be singular in
$\OthStagMod / \uea \overline{\omega}_0$ (for some choice of $\varpi$)
if and only if the beta-invariant $\overline{\beta}$ of
$\func{\overline{\OthStagMod}}{\varpi}$ vanishes.  We remark that
this is equivalent to the vanishing of $\overline{\beta}$ for any
$\func{\overline{\OthStagMod}}{u}$,
$u \in \bigl( \VerMod_{\LDim} \bigr)_{\ell + \overline{\ell}}$,
by gauge invariance.

In general, $\nbrGamma$ and $\Rnull$ may be greater
than $1$ and there are a few possibilities among the submodules of $\msmbar$
that we might want to quotient out. 
We will analyse whether the submodules
$\uea \bigl( \overline{X}^{\eps} y - \varpi^{\eps} \bigr) \subset \OthStagMod$
contain the
\svs{} $X_{\rnkO+\overline{\rnkO}-1}^{\eps'} x \in \VerMod_{\LDim}$,
where $\eps' \in \set{-,+}$ parametrises the non-vanishing \svs{}
$X_{\rnkO+\overline{\rnkO}-1}^{\eps'} x \neq 0$ of $\LMod$.
As in the argument above, we
find that to each generating \sv{} of $\sJ$ and each rank
$\rnkO+\overline{\rnkO}-1$ \sv{} of $\LMod$, there is a
corresponding beta-invariant which must vanish. 
Specifically, given vectors
$\overline{y}^{\eps} = \overline{X}^{\eps} y - \varpi^{\eps}$
such that $\overline{\omega}^{\eps}_j = L_j \overline{y}^{\eps}
\in \overline{\minsubmod}$ and elements $\overline{\chi}^{\eps}_{\eps'}
\in \uea^-$ (singular and prime) such that
$X_{\rnkO+\overline{\rnkO}}^{\eps} =
\overline{\chi}^{\eps}_{\eps'} X_{\rnkO+\overline{\rnkO}-1}^{\eps'}$
and $X_{\rnkO+\overline{\rnkO}-1}^{\eps'} x \neq 0$ in $\LMod$, we
define $\nbrGamma \Rnull \in \set{0,1,2,4}$ beta-invariants by
\begin{equation} \label{eqnDefBetaBars}
\bigl( \overline{\chi}^{\eps}_{-} \bigr)^\dagger \overline{y}^{\eps} = \overline{\beta}^{\eps}_{-} X_{\rnkO+\overline{\rnkO}-1}^{-} x \qquad \text{and} \qquad \bigl( \overline{\chi}^{\eps}_{+} \bigr)^\dagger \overline{y}^{\eps} = \overline{\beta}^{\eps}_{+} X_{\rnkO+\overline{\rnkO}-1}^{+} x \pmod{\uea X^-_{\rnkO+\overline{\rnkO}-1} x}.
\end{equation}
These are the beta-invariants of the
$\func{\overline{\OthStagMod}^{\eps}}{\varpi^{\eps}}$,
and we may quotient $\OthStagMod$ to get a staggered
module $\check{\StagMod}$ with left module $\LMod$ and
singular vectors $\overline{y}^{\eps}$ if and only if
all of the $\overline{\beta}^{\eps}_{\eps'}$ vanish
(the easy proof of this is sketched below).
We have indicated which beta-invariants are relevant
to each critical rank case in \figref{figCritRank} for
convenience.  We further remark that we will suppress
the indices $\eps$ and $\eps'$ in cases where they take
a single value (as in case $(1,1)$ above).

\begin{theorem} \label{thm: general existence in the critical case}
Given $\Hlft$, $\Hrgt = \VerMod_{\RDim} / \sJ$ and
$(\omega_1,\omega_2) \in \data$ such that $\Hlft$ contains
non-zero \svs{} of rank $\rnkO+\overline{\rnkO}-1$, a
staggered module $\Stagg$ with these left and right modules
and data exists if and only if
$\overline{\beta}^{\eps}_{\eps'} = 0$ for all $\eps \in \set{-,+}$
such that $\overline{X}^{\eps} v_{h^\rgt} \in \sJ$ and all
$\eps' \in \set{-,+}$ such that
$X^{\eps'}_{\rnkO+\overline{\rnkO}-1} x \neq 0$.
\end{theorem}

\begin{proof}
In view of \propref{prop: singular vectors in staggered modules} and
the above, all that needs to be proven is that the vanishing of the
appropriate invariants $\overline{\beta}$ occurs precisely when the
$\overline{y}$ become non-vanishing \svs{} in the quotient
$\check{\StagMod} = \OthStagMod / \sK$ (recall that $\sK$ is then a
submodule of $\overline{\minsubmod}$). To lighten the notation, we
will omit superscript indices $\eps$.  It is understood that what
follows must be repeated separately for the $\Rnull$ values that
$\eps$ takes.

It is clear that $\overline{y}$ will be singular if and only if
both $\overline{\omega}_1$ and $\overline{\omega}_2$ belong to
$\sK$.  When $\sK$ is generated by \svs{} of grades $\ell + \ellbar$
or greater, for example when $\VerMod_{\LDim}$
is of chain type, this requires that the data vanishes
(this direction is always easy in fact).  Now, the data
can be \emph{chosen} to vanish using a gauge transformation if and
only if all beta-invariants $\overline{\beta}$ or $\overline{\beta}_\pm$
vanish, because vanishing data is admissible (\propref{propAdmissible}),
gauge transformations connect any two equivalent pieces of data
(\propref{prop: equivalence}) and beta-invariants completely determine
the isomorphism class (\thmref{thm: invariants}).  The proof is then
complete for such $\sK$.

However, $\sK \subset \overline{\minsubmod}$ may be
generated by \svs{} of lower grade than $\ell + \overline{\ell}$. 
To deal with this possibility, note that
\begin{equation}
\overline{\omega}_j \in \uea X^\pm_{\rnkO+\overline{\rnkO}-1} x \qquad \Rightarrow \qquad \overline{\beta}_\mp = 0.
\end{equation}
Indeed, this is just the analogue of (a part of) \eqnref{eqnWrongSubspace}
in the present situation, and it immediately implies
that if $\overline{y}$ becomes singular in the quotient $\OthStagMod/\sK$,
then the invariants $\betabar_{\eps'}$ vanish. Roughly speaking, the converse
is also true:  Split the data as $\omegabar_j = \omegabar^+_j + \omegabar_j^-$,
where $\overline{\omega}^\pm_j \in \uea X^\pm_{\rnkO+\overline{\rnkO}-1} x$.
From the arguments in \secref{sec: dimensions} and \secref{sec: invariants},
we can infer that the admissible $(\omegabar_1^\pm, \omegabar_2^\pm)$
modulo the gauge transformations $g_u$, 
$u \in (\uea X_{\rnkO+\rnkR-1}^\pm)_{\ell + \overline{\ell}}$,
form a one-dimensional vector space parametrised by $\betabar_\pm$.
Then, $\overline{\beta}_\pm = 0$ implies that we can \emph{choose}
$\omegabar_j^\pm = 0$ by a gauge transformation. It follows that the
``full data'' $\omegabar_j$ can be chosen to belong to
$\uea X^\mp_{\rnkO+\overline{\rnkO}-1} x$. When $\sK$ is generated by
$X_{\rnkO+\rnkR-1}^\mp x$, the vanishing of $\betabar_\pm$ therefore implies
that $\overline{y}$ becomes singular in the quotient $\OthStagMod/\sK$.  When $\sK$ is generated by \svs{} of grade $\rnkO+\overline{\rnkO}$, the vanishing of both the $\betabar_\pm$ implies the same.  This completes the proof.
\end{proof}

Determining when the beta-invariants $\overline{\beta}^{\eps}_{\eps'}$ of the staggered modules $\func{\overline{\OthStagMod}^{\eps}}{u}$ vanish is an explicit condition which can be checked in particular examples (see \exref{ex:-2,13,21} below).  To get more insight into this, we revisit the definitions of these beta-invariants using the forms given in \eqnDref{eqnDefBetaChain}{eqnDefBetaBraid}.  This allows us to set $u=0$ and write
\begin{equation}
\overline{\beta}^{\eps}_{\eps'} X_{\rnkO+\overline{\rnkO}-1}^{\eps'} x = \bigl( \overline{\chi}^{\eps}_{\eps'} \bigr)^\dagger \bigl( \wun - \overline{\mathbb{P}} \bigr) \overline{y}^{\eps} = \bigl( \overline{\chi}^{\eps}_{\eps'} \bigr)^\dagger \bigl( \wun - \overline{\mathbb{P}} \bigr) \overline{X}^{\eps} y,
\end{equation}
modulo $\uea X_{\rnkO+\overline{\rnkO}-1}^- x$ if $\eps' = +$, where $\wun - \overline{\mathbb{P}}$ denotes the net effect of the Projection Lemma (as in \secref{sec: invariants}).  Now, $\overline{X}^{\eps}$ is singular, and both $\bigl( \overline{\chi}^{\eps}_{\eps'} \bigr)^\dagger$ and $\overline{\mathbb{P}}$ have positive modes on the right of each of their terms (see \eqnDref{eqnDefPChain}{eqnDefPBraid}).  We may therefore write
\begin{equation}
\bigl( \overline{\chi}^{\eps}_{\eps'} \bigr)^\dagger \bigl( \wun - \overline{\mathbb{P}} \bigr) \overline{X}^{\eps} = U_0^{(\eps,\eps')} \brac{L_0 - \RDim} + U_1^{(\eps,\eps')} L_1 + U_2^{(\eps,\eps')} L_2
\end{equation}
for some $U_0^{(\eps,\eps')}, U_1^{(\eps,\eps')}, U_2^{(\eps,\eps')} \in \uea$, hence
\begin{equation} \label{eqnBetaBarAffLin}
\overline{\beta}^{\eps}_{\eps'} X_{\rnkO+\overline{\rnkO}-1}^{\eps'} x = U_0^{(\eps,\eps')} \omega_0 + U_1^{(\eps,\eps')} \omega_1 + U_2^{(\eps,\eps')} \omega_2.
\end{equation}
This expresses the $\overline{\beta}^{\eps}_{\eps'}$ as
\emph{affine-linear} functionals of the data $(\omega_1 , \omega_2)$
of $\OthStagMod$ (and thus also of $\check{\StagMod}$).  Finally,
applying a gauge transformation $g_u$ to $(\omega_1 , \omega_2)$
results in the left hand side of \eqnref{eqnBetaBarAffLin} changing
by
\begin{equation}
\bigl( \overline{\chi}^{\eps}_{\eps'} \bigr)^\dagger \bigl(
    \wun - \overline{\mathbb{P}} \bigr) \overline{X}^{\eps} u
    - U_0^{(\eps,\eps')} \brac{L_0 - \RDim} u = 0,
\end{equation}
since $u$ has conformal dimension $\RDim$ and
$\bigl( \wun - \overline{\mathbb{P}} \bigr) \overline{X}^{\eps} u
\in \overline{\minsubmod}_{\ell + \overline{\ell}}$.  This gauge
invariance then lets us conclude that the
$\overline{\beta}^{\eps}_{\eps'}$ are affine functions on
the space $\data/\gauge$ of isomorphism classes of staggered modules
$\check{\StagMod}$ with exact sequence (\ref{eqnSES6}).  Assuming that
$\ell > 0$, so that the beta-invariants $\beta$ or $\beta_\pm$ of
$\check{\StagMod}$ are defined, we can therefore consider the
$\overline{\beta}^{\eps}_{\eps'}$ as affine functions of $\beta$ or $\beta_\pm$.

\begin{example} \label{ex:-2,13,21}
We consider the existence of a $c=-2$ ($t=2$) staggered module $\StagMod$ with exact sequence
\begin{equation}
\dses{\VerMod_0 / \VerMod_3}{}{\StagMod}{}{\VerMod_1 / \VerMod_3}.
\end{equation}
We therefore have $X = L_{-1}$, $\rnkO=1$, $\overline{X} = \overline{\chi} = L_{-1}^2 - 2 L_{-2}$ and $\overline{\rnkO}=1$.  Since $\omega_0 = L_{-1} x$ has rank $\rnkO+\overline{\rnkO}-1 = 1$, this is a critical rank example.

By \thmref{thm: moduli space right Verma}, there is a one-dimensional space of staggered modules $\OthStagMod$ with left module $\VerMod_0$ and right module $\VerMod_1$, parametrised by $\beta$.  We must determine the beta-invariant $\overline{\beta}$ of the submodule $\func{\overline{\OthStagMod}}{0}$ generated by $x$ and $\overline{X} y$.  Referring to the calculation of \exref{ex:-2,V0,V3}, we have
\begin{align}
\overline{\beta} \omega_0 &= \overline{X}^\dagger \bigl( \wun - \overline{\mathbb{P}} \bigr) \overline{X} y = \bigl( L_1^2 - 2 L_2 \bigr) \bigl( \wun + \tfrac{1}{4} L_{-3} L_3 \bigr) \bigl( L_{-1}^2 - 2 L_{-2} \bigr) y \notag \\
&= \bigl( 8 L_{-1} L_0 L_1 - 15 L_{-1} L_1 + 4 \brac{2 L_0 + 1} \brac{L_0 - 1} \bigr) y = \brac{-15 \beta + 12} \omega_0.
\end{align}
The conclusion is then that $\StagMod$ exists by
\thmref{thm: general existence in the critical case} if and only if the
affine relation $\overline{\beta} = 12 - 15 \beta = 0$ holds, which requires
$\beta = \tfrac{4}{5}$.  This value is of course reproduced by searching for
an explicit singular vector of the form
$\overline{y} = \overline{X} y - \varpi$ with
$\varpi \in \VerMod_0 / \uea \overline{X} \omega_0 = \LMod$
(as in \secref{secStagModSVs}).
\end{example}

Consider a case $(1,1)$ staggered module $\StagMod$ of chain type (or $\rnkO=1$ braid type).
If $\ell > 0$ (so $\rnkO > 0$), then there is a single
invariant $\beta$ to consider.  By \thmref{thm: general existence in
the critical case}, $\StagMod$ exists if and only if a single invariant
$\overline{\beta}$ vanishes.  We have shown that the latter invariant is
an affine function of the former, so there are three possibilities:
\begin{itemize}
\item $\overline{\beta}$ is constant and zero, so $\StagMod$ exists for all $\beta$.
\item $\overline{\beta}$ is constant and non-zero, so $\StagMod$ does not exist for any $\beta$.
\item $\overline{\beta}$ is not constant, so $\StagMod$ exists for a unique $\beta$.
\end{itemize}
In the absence of any information to the contrary, we should expect
that the last possibility is overwhelmingly more likely to occur. 
And indeed, this is what we observe.  For instance, the staggered
modules of \exTref{ex:0,12,14}{ex:0,14,18}{ex:-2,13,21} all admit
only a single value of $\beta$.  We can now finally understand this
as the generic consequence of imposing one (linear, inhomogeneous)
relation, $\overline{\beta} = 0$, on one unknown, $\beta$.

More generally, we can use \eqnref{eqnBetaBarAffLin} to decompose the beta-invariants of the $\func{\overline{\OthStagMod}^{\eps}}{0}$ as
\begin{equation}
\func{\overline{\beta}^{\eps}_{\eps'}}{\omega_1 , \omega_2} = \func{\gamma^{\eps}_{\eps'}}{\omega_1 , \omega_2} + \func{\overline{\beta}^{\eps}_{\eps'}}{0,0}.
\end{equation}
This defines \emph{linear} functionals $\gamma^{\eps}_{\eps'}$ on the
space of data of the $\OthStagMod$ (and $\check{\StagMod}$).  Let
$\nbrBeta \in \set{0,1,2}$ denote the number of beta-invariants
needed to describe the $\OthStagMod$.  Assuming that the
$\gamma^{\eps}_{\eps'}$ are all linearly
independent,\footnote{Unfortunately, demonstrating this
linear independence (in particular, the non-vanishing)
seems to require a significantly more delicate analysis
than that presented in the proof of \thmref{thm: invariants}.
We hope to return to this issue in a future publication.} we
therefore obtain $\nbrGamma \Rnull$ linear relations in
$\nbrBeta$ unknowns.  Analysing these numbers in each
case then leads to simple expectations for the dimension
of the space of staggered modules $\StagMod$.

More specifically, when the left and right modules are of chain type,
$\nbrBeta$ is $0$ or $1$, depending on whether $\rnkO = 0$ or not. 
In the braid case, $\nbrBeta$ is $0$, $1$ or $2$, depending on whether
$\rnkO = 0$, $\rnkO = 1$ or $\rnkO > 1$ (this is a direct restatement
of \thmref{thm: moduli space right Verma}).  We should therefore
\emph{expect} that the staggered modules $\StagMod$ corresponding
to the critical rank configurations of \figref{figCritRank} will
exist in case $(1,1)$ provided that $\rnkO > 0$ and cases $(1,2)$
and $(2,1)$ provided that $\rnkO > 1$.  We should not expect the
$\StagMod$ to exist otherwise.  Moreover, we expect that when
$\StagMod$ exists, it is unique, except in case $(1,1)$ with
braid type and $\rnkO > 1$, in which case we expect a one-parameter family of
staggered modules.

\begin{example} \label{ex:CritRank}
It is easy to investigate examples of critical rank staggered modules
using the \sv{} result of \propref{prop: singular vectors in staggered
modules}.  For example, we know from \exref{ex: comparison of similar
cases} that a $c=0$ ($t=\tfrac{3}{2}$) staggered module with
$\LMod = \VerMod_0 / \VerMod_2$ and $\RMod = \VerMod_1 / \VerMod_5$ is
unique, admitting only $\beta = -\tfrac{1}{2}$.  Similarly, replacing
the right module by $\VerMod_1 / \VerMod_7$ leads to a unique staggered
module with $\beta = \tfrac{1}{3}$.  These beta-invariants
are the unique solutions of $\betabar = -\frac{1120}{3} (2 \beta + 1) = 0$
and $\betabar = -\frac{17248000}{243} (3 \beta - 1) = 0$ respectively.
Moreover, both of these examples fall into case $(1,1)$, but we may
deduce from their uniqueness (which was anticipated above)
that the case $(1,2)$ staggered module $\StagMod$ corresponding to replacing
the right module by $\VerMod_1 / \brac{\VerMod_5 + \VerMod_7}$
\emph{does not exist}:  The associated $\check{\StagMod}$ would have to
have \svs{} at grades $5$ and $7$, requiring both
$\beta = -\tfrac{1}{2}$ \emph{and} $\beta = \tfrac{1}{3}$.

For case $(2,1)$ examples, we take $\LMod = \VerMod_0 / \brac{\VerMod_5 + \VerMod_7}$ and $\RMod = \VerMod_h / \VerMod_{h'}$ for $h = 1,2$ and $h' = 5,7$ (and $c=0$).  In all these cases, $\rnkO = 1$, so we do not expect that such staggered modules exist.  And one can explicitly check in each case that the appropriate \sv{} does not exist, confirming our expectations.  It is more interesting to consider the $\rnkO = 2$ examples with $\LMod = \VerMod_0 / \brac{\VerMod_{12} + \VerMod_{15}}$ and $\RMod = \VerMod_5 / \VerMod_h$ for $h = 12$ and $15$.  The \svs{} turn out to exist if and only if
\begin{equation}
\beta_- = -\tfrac{11200}{51}, \ \beta_+ = \tfrac{1680}{17} \qquad \text{and} \qquad \beta_- = -\tfrac{5600}{57}, \ \beta_+ = \tfrac{3360}{19},
\end{equation}
respectively, in line with expectations.  Finally, if we replace the right module by $\VerMod_5 / \brac{\VerMod_{12} + \VerMod_{15}}$ to get a case $(2,2)$ example, we see from the different $\beta_\pm$ above that this staggered module cannot exist, again as anticipated.

Our last example illustrates case $(1,1)$ with $\rnkO > 1$.  We search for a $c=0$ staggered module $\StagMod$ with $\LMod = \VerMod_0 / \VerMod_7$ and $\RMod = \VerMod_5 / \VerMod_{12}$, hence $\rnkO = 2$.  The corresponding $\check{\StagMod}$ turns out to have a \sv{} at grade $12$ provided that
\begin{equation}
189 \beta_- + 80 \beta_+ = -3360.
\end{equation}
It follows that there exists a one-parameter family of such staggered modules $\StagMod$, just as we expect.
\end{example}

The above examples completely support our na\"{\i}ve expectations concerning the dimensions of the spaces of critical rank staggered modules.  However, things are never quite as simple as one might like.

\begin{example} \label{ex:Drat}
Let $c=1$ ($t=1$) and $\LMod = \RMod = \VerMod_{1/4} / \VerMod_{9/4}$.  These are chain type modules with $\ell = 0$, so the corresponding staggered module $\StagMod$ would be a case $(1,1)$ critical rank example with $\rnkO = 0$.  With no $\beta$ but one $\overline{\beta}$, we should not expect that such an $\StagMod$ exists.  Nevertheless, it is easy to check that the vector
\begin{equation}
\brac{L_{-1}^2 - L_{-2}} y -\tfrac{4}{3} L_{-2} x \in \check{\StagMod}
\end{equation}
is singular.  By \propref{prop: singular vectors in staggered modules}, a staggered module with this left and right module \emph{does} therefore exist, contrary to our expectations.
\end{example}

\begin{example} \label{ex:GenDrat}
We can readily generalise the realisation of \exref{ex:Drat} for other
$\ell = 0$ examples.  Let $t$ be arbitrary but let $h = h_{r,s}$,
$r,s \in \ZZ_+$, vary with $t$ as in \eqnref{eqnKacDims}.  Then,
$\overline{X} \in \uea^-_{rs}$ also varies with $t$, though it need
not remain prime (that is, $\overline{\chi} = \overline{X}$ for
generic $t$ only).  We may therefore methodically investigate the
existence of staggered modules with
$\LMod = \RMod = \VerMod_{h} / \VerMod_{h+rs}$ by computing
\begin{equation}
\overline{\beta}' x = \overline{X}^\dagger \overline{X} y \in \OthStagMod
\end{equation}
for small $r$ and $s$ (because $\ell = 0$, there is no
$\overline{\mathbb{P}}$).  Clearly $\overline{\beta}'$ need not
coincide with the true invariant $\overline{\beta}$ if $\overline{X}$
is composite.  Some results are
(note that swapping $r$ and $s$ amounts to inverting $t$):

\begin{table}[!ht]
\begin{center}
\setlength{\extrarowheight}{3pt}
\begin{tabular}{C|C|C}
(r,s) & \overline{\beta}' & \overline{\beta} = 0 \\
\hline \hline
(1,1) & 2 & \text{--} \\
\hline
(2,1) & 4 \brac{t^2-1} & t=\pm1 \\
\hline
(3,1) & 24 \brac{t^2-1} \brac{4t^2-1} & t=\pm1 \\
\hline
(4,1) & 288 \brac{t^2-1} \brac{4t^2-1} \brac{9t^2-1} & t=\pm1,\ \pm\tfrac{1}{2} \\
\hline
(5,1) & 5760 \brac{t^2-1} \brac{4t^2-1} \brac{9t^2-1} \brac{16t^2-1} & t=\pm1 \\
\hline
(6,1) & 172800 \brac{t^2-1} \brac{4t^2-1} \brac{9t^2-1} \brac{16t^2-1} \brac{25t^2-1} & t=\pm1,\ \pm\tfrac{1}{2},\ \pm\tfrac{1}{3} \\
\hline \hline
(2,2) & -8t^{-4} \brac{t^2-1}^2 \brac{t^2-4} \brac{4t^2-1} & t=\pm\tfrac{1}{2},\ \pm2 \\
\hline
(3,2) & -192t^{-6} \brac{t^2-1}^3 \brac{t^2-4} \brac{4t^2-1}^2 \brac{9t^2-1} & t=\pm\tfrac{1}{3},\ \pm2 \\
\end{tabular}
\end{center}
\end{table}

\noindent Here, we list those $t$ for which $\overline{\beta}'$ vanishes
\emph{and} for which this vanishing implies the vanishing of
$\overline{\beta}$ (which requires $\overline{X}$ to be prime), hence the
existence of a staggered module with
$\LMod = \RMod = \VerMod_h / \VerMod_{h+rs}$.  This sequence of examples
makes it clear that given $r$ and $s$, staggered modules of this kind can
certainly exist.
\end{example}

In the $\ell=0$ case discussed above, the invariants $\overline{\beta}$ are evidently constants. As we have seen, their vanishing is nevertheless a subtle question. However, continuing the analysis of \exref{ex:GenDrat} leads to a clear pattern for the existence question in this case, and in fact, this question was already solved explicitly (for chain type modules) by Rohsiepe in \cite{RohRed96}. His argument extends to any staggered module for which $\beta=0$ or $(\beta^-,\beta^+) = (0,0)$ and $\overline{X}$ is prime ($\overline{\rnkO}=1$), and we outline it below.  Note that this is always a critical rank case.

\begin{proposition} \label{propRohsiepe}
Suppose that $\check{\Stagg}$ is a staggered module with left module $\Hlft$,
right module $\Verma_{h^\rgt}$ and all beta-invariants vanishing
(if any are defined).
Suppose further that the prime \sv{} $\overline{X} v_{h^\rgt}$
of \emph{smallest}
grade $\overline{\ell}$ is such that $\overline{X} \omega_0 = 0$.
Then there exists a singular vector in $\check{\Stagg}$ at grade
$\ell+\overline{\ell}$ if and only if $h^\rgt = h_{r,s}$ with
$t = \tfrac{q}{p} \in \QQ$ (where $\gcd \set{p,q} = 1$), $p \mid r$, $q \mid s$ and $\abs{p} s \neq \abs{q} r$.
\end{proposition}
\begin{proof}
We will prove the existence of the singular vector by
demonstrating the vanishing of $\betabar$ or $\betabar_\pm$.  
We immediately remark that the
assumption of $\ellbar$ being the smallest grade of a prime singular
means that $\ellbar = rs$ for a pair $(r,s) \in \bZ_+ \times \bZ_+$
that satisfies $h_{r,s}=h^\rgt$ with minimal product $rs$.

Since any invariants of $\check{\Stagg}$ vanish, we may choose
$y \in \check{\Stagg}$ such that $\omega_1 = \omega_2 = 0$,
by \propDref{prop: equivalence}{propAdmissible} and \thmref{thm: invariants}.
Writing
$L_j \overline{X} = V_0 (L_0-h^\rgt) + V_1 L_1 + V_2 L_2$, we notice that
with this choice,
$L_j \overline{X} y = V_0 \omega_0 \in \overline{\minsubmod}$,
so we need no projections to define the
$\overline{\beta}_{\pm}$.
Now, one of these invariants is given (perhaps
modulo $\uea X_\rnkO^- x$) by
\begin{equation}
\overline{\beta} \omega_0 = \overline{X}^\dagger \overline{X} y
= U_0 \brac{L_0 - h^R} y,
\end{equation}
where we have written
$\overline{X}^\dagger \overline{X} = U_0 (L_0-h^\rgt) + U_1 L_1 + U_2 L_2$
as usual.  But by \PBW-ordering appropriately, we may choose
$U_0 = \func{f}{L_0}$ for some polynomial $f$, since
$\overline{X}^\dagger \overline{X} \in \uea_0$.  We therefore obtain
\begin{equation}
\overline{\beta} \omega_0 = \func{f}{L_0} \brac{L_0 - h^R} y = \func{f}{\RDim} \omega_0.
\end{equation}
The vanishing of $\overline{\beta}$ is therefore equivalent to $\RDim$
being a zero of $f$, hence a double zero of $\func{f}{h} \brac{h - h^R}$.

Consider now the \hws{} $v_h \in \VerMod_h$.  We have
\begin{equation}
\overline{X}^\dagger \overline{X} v_h = U_0 (L_0-h^\rgt) v_h
    = \func{f}{h} \brac{h - h^R} v_h.
\end{equation}
By extending $\set{\overline{X} v_h}$ to a basis of
$\brac{\VerMod_h}_{\overline{\ell}}$, it is possible to show
that $\overline{\beta} = 0$ if and only if the Kac determinant
(\eqnref{eqnKac}) of $\Verma_{h}$ at grade $\overline{\ell} = rs$
possesses a double zero at $h=h^\rgt$
(this is an innocent generalisation of the statement of
\cite[Lemma~ 6.2]{RohRed96} --- its proof needs no changes).
No $(r', s')$ with $h_{r',s'}=h^\rgt$ has
$r's'<\ellbar=rs$, so the double zero can only occur if there is another
such pair $(r',s') \neq (r,s)$ with $r's'=rs$. Such a second distinct pair
is easily verified to have the form $(r',s') = (\abs{t}^{-1} s , \abs{t} r)$, and integrality and distinctness yield the conditions given in the statement of the proposition.

These conditions are equivalent to the
vanishing of \emph{this} $\overline{\beta}$.  But, they also imply
that $\RMod$ and $\LMod$ are of chain type.  Hence this
is the only $\betabar$ and its vanishing is actually sufficient for the existence of the \sv{}.  This completes the proof.
\end{proof}

The restriction that $\overline{X} v_{h^\rgt}$ have minimal (positive) grade
is not serious, but Rohsiepe's argument requires some refining in this case.
Essentially, if $\VerMod_{\RDim}$ is of braid type with
$\overline{X} = X_1^+$, we generalise \cite[Lemma~ 6.2]{RohRed96} to
conclude that $\overline{\beta} = 0$ is equivalent to the Kac determinant
of $\Verma_{h}$ at grade $\overline{\ell} = \overline{\ell}_1^+$ having a
zero at $h = \RMod$ of order
$\tpartnum{\overline{\ell}_1^+ - \overline{\ell}_1^-} + 2$ (or greater). 
However, coupling the explicit form of the Kac determinant formula with
the conclusion of \propref{propRohsiepe} for $\overline{X} = X_1^-$, we
can deduce that the order of this zero is precisely
$\tpartnum{\overline{\ell}_1^+ - \overline{\ell}_1^-} + 1$.  Thus,
$\overline{\beta}$ cannot vanish.

This solves the existence question for staggered modules $\StagMod$
with no non-vanishing beta-invariants, $\overline{X} \omega_0 = 0$ and
$\RMod = \VerMod_h / \uea \overline{X} v_h$, where $\overline{X}$ is
prime:  They exist if and only if $h = h_{\lambda \abs{p} , \mu \abs{q}}$
for some $\lambda , \mu \in \ZZ_+$, where $t = \tfrac{q}{p} \in \QQ$
and $\lambda \neq \mu$.  In particular, the left and right modules must
be of chain type.  One can also deduce existence for general $\overline{X}$,
assuming existence when $\overline{X}$ is prime, by inductively applying
\propref{propRohsiepe} to certain submodules of (quotient modules of) the
corresponding $\OthStagMod$.  However, deducing general non-existence
from non-existence when $\overline{X}$ is prime requires far more
intricate extensions of Rohsiepe's argument.  Such arguments could complete the analysis in some further special cases, but the details
are not in the spirit of what we have achieved here, so we
will not elaborate any further upon them.

As mentioned before, the existence of these $\ell = 0$ critical rank
staggered modules is certainly not in line with our na\"{\i}ve
expectations based on counting constraints and unknowns.  However,
viewed in the light of \exref{ex:GenDrat}, we can conclude that these
counterexamples to our expectations are in fact quite rare --- given
$\overline{\ell}$, then in all the continuum of values of $t$ there
are only finitely many for which such staggered modules exist.

Of course, we should contrast this with the critical rank cases not
covered by \propref{propRohsiepe}.  In these cases, whilst we have
not been able to rule out counterexamples to our expectations, we
know of none!  We would like to offer a speculative argument
suggesting why this is so.  Recall that the analysis of the cases
covered by \propref{propRohsiepe} is simplified by not requiring
the $\overline{\mathbb{P}}$ when defining the $\overline{\beta}$. 
Structurally, we only need consider \emph{two} \svs{}, $\omega_0$
(which may as well be $x$ in the analysis) and $\overline{X} v_{\RDim}$
(which is prime), in our calculations.  The key observation which we
exploited in \exref{ex:GenDrat} was that such a configuration of two
\svs{} can be continuously deformed for all $t$.  The result was
(modulo issues of $\overline{X}$ remaining prime) an expression for
$\overline{\beta}$ as a \emph{polynomial} in $t$ and $t^{-1}$.  Given
this, it is no surprise that this polynomial will vanish for some
values of $t$.  In other words, because each $\overline{\beta}$
corresponds to a configuration of only two \svs{}, we should expect
that our na\"{\i}ve counting arguments will fail from time to time.

In contrast, the more general critical rank cases require the
consideration of at least three \svs{}.  Such configurations cannot
be deformed smoothly --- varying $t$ without at least one of \svs{} disappearing is impossible.  There is therefore very little to be gained from trying
to express the $\overline{\beta}$ as polynomials in $t$ and $t^{-1}$
because the result will not correspond to a meaningful beta-invariant
for almost all $t$.  For this reason, we suspect that counterexamples
to our na\"{\i}ve expectations of this more general type must be
\emph{significantly rarer} than those guaranteed
by \propref{propRohsiepe}.  Indeed, one might even be tempted to
conjecture that there are in fact no counterexamples beyond those
which we have described above.  Evidently, more work is necessary
to further understand this important situation.

\section{Summary of Results} \label{sec: conclusions}

In the preceding sections, we have answered our main question --- that of the characterisation and classification of staggered modules --- in an expository yet detailed manner.  We fully expect that the formalism developed throughout the course of this study will be invaluable when faced with further questions concerning these kinds of indecomposable modules and their generalisations.  Moreover, we have tried throughout to illustrate with examples how such questions arise in concrete practical studies \emph{and can be answered}.

The details should nevertheless not prevent us from presenting the reasonably simple answer that we have obtained to the original question.  The results may be presented in purely structural terms, as one would hope, and we are finally in a position to summarise what we have shown.

\begin{theorem} \label{thmTheAnswer}
Given two \hwms{} $\Hlft$ and $\Hrgt$ of central charge $c$ and highest weights $h^\lft$ and $h^\rgt$ respectively, the space $\mathbb{S}$ of isomorphism classes of staggered modules $\StagMod$ with exact sequence
\begin{equation*}
\dses{\LMod}{\iota}{\StagMod}{\pi}{\RMod}
\end{equation*}
is described as follows.  Let:
\begin{itemize}
\item $\ell = h^\rgt - h^\lft$ be the grade of a \sv{} $\omega_0 \in \Hlft$.
\item $\rnkO$ be the rank of $\omega_0$, if $\omega_0 \neq 0$.
\item $\Rnull \in \set{0,1,2}$ be the number of generating \svs{} $\overline{X}^\eps v_{h^\rgt}$ of $\sJ$, where $\Hrgt = \Verma_{h^\rgt}/\sJ$.
\item $\overline{\rnkO}$ be the rank of the $\overline{X}^\eps v_{h^\rgt}$, if $\Rnull > 0$.
\item $\nbrBeta \in \set{0,1,2}$ be the number of (non-zero) rank $\rnkO-1$ \svs{} in $\LMod$.
\item $\nbrGamma \in \set{0,1,2}$ be the number of (non-zero) rank
$\rnkO+\overline{\rnkO}-1$ singular vectors in $\Hlft$, if $\Rnull > 0$.
\end{itemize}
Then:
\begin{itemize}
\item There exists no such $\Stagg$ unless $\omega_0 \neq 0$ (requiring $\ell$ to be a non-negative integer).
\item There exists no such $\Stagg$ unless each $\overline{X}^\eps \omega_0 = 0$.
\end{itemize}
Assuming these necessary conditions are met, we have:
\begin{itemize}
\item When $\Rnull = 0$ or $\Rnull > 0$ but $\nbrGamma = 0$,
$\mathbb{S}$ is a vector space $\data / \gauge$ of dimension $\nbrBeta$.
When non-trivial, this vector space is parametrised by the beta-invariants
of \secref{sec: invariants}.
\item In general, $\mathbb{S}$ is an affine subspace of $\data / \gauge$
characterised by the vanishing of the $\nbrGamma \Rnull$ auxiliary
beta-invariants of \secref{secCritRank}.
\end{itemize}
\end{theorem}

\thmref{thmTheAnswer} gives a complete description of the space
$\mathbb{S}$, hence a complete classification of staggered modules, when
$\Rnull = 0$ or $\Rnull > 0$ and $\nbrGamma = 0$.  In the few remaining cases
in which $\Rnull > 0$ and $\nbrGamma > 0$ (the critical rank cases of
\secref{secCritRank}), our classification is not complete.  For these
cases, pictured in \figref{figCritRank}, we can however say that if
$\nbrBeta = 0$ (or all the beta-invariants of \secref{sec: invariants}
vanish), then the nature of $\mathbb{S}$ is determined by
\propref{propRohsiepe} and its simple consequences.  Otherwise, we expect
(based on some speculative reasoning and an extensive study of examples)
that the dimension of $\mathbb{S}$ is given by
\begin{equation}
\dim \mathbb{S} = \nbrBeta - \nbrGamma \Rnull,
\end{equation}
where negative dimensions indicate that $\mathbb{S}$ is empty.  We hope to report on the validity of this expectation in the future.


\begin{thebibliography}{10}

\bibitem{KacCon79}
V~Kac.
\newblock {Contravariant Form for Infinite Dimensional Algebras and
  Superalgebras}.
\newblock {\em Lect. Notes in Phys.}, 94:441--445, 1979.

\bibitem{FeiSke82}
B~Feigin and D~Fuchs.
\newblock {Skew-Symmetric Differential Operators on the Line and Verma Modules
  over the Virasoro Algebra}.
\newblock {\em Func. Anal. and Appl.}, 16:114--126, 1982.

\bibitem{BelInf84}
A~Belavin, A~Polyakov, and A~Zamolodchikov.
\newblock {Infinite Conformal Symmetry in Two-Dimensional Quantum Field
  Theory}.
\newblock {\em Nucl. Phys.}, B241:333--380, 1984.

\bibitem{RozQua92}
L~Rozansky and H~Saleur.
\newblock {Quantum Field Theory for the Multivariable Alexander-Conway
  Polynomial}.
\newblock {\em Nucl. Phys.}, B376:461--509, 1992.

\bibitem{GurLog93}
V~Gurarie.
\newblock {Logarithmic Operators in Conformal Field Theory}.
\newblock {\em Nucl. Phys.}, B410:535--549, 1993.
\newblock arXiv:hep-th/9303160.

\bibitem{FloBit03}
M~Flohr.
\newblock {Bits and Pieces in Logarithmic Conformal Field Theory}.
\newblock {\em Int. J. Mod. Phys.}, A18:4497--4592, 2003.
\newblock arXiv:hep-th/0111228.

\bibitem{FjeLog02}
J~Fjelstad, J~Fuchs, S~Hwang, A~Semikhatov, and I~Yu Tipunin.
\newblock {Logarithmic Conformal Field Theories via Logarithmic Deformations}.
\newblock {\em Nucl. Phys.}, B633:379--413, 2002.
\newblock arXiv:hep-th/0201091.

\bibitem{FeiLog06}
B~Feigin, A~Gainutdinov, A~Semikhatov, and I~Yu Tipunin.
\newblock {Logarithmic Extensions of Minimal Models: Characters and Modular
  Transformations}.
\newblock {\em Nucl. Phys.}, B757:303--343, 2006.
\newblock arXiv:hep-th/0606196.

\bibitem{PeaLog06}
P~Pearce, J~Rasmussen, and J-B Zuber.
\newblock {Logarithmic Minimal Models}.
\newblock {\em J. Stat. Mech.}, 0611:017, 2006.
\newblock arXiv:hep-th/0607232.

\bibitem{ReaAss07}
N~Read and H~Saleur.
\newblock {Associative-Algebraic Approach to Logarithmic Conformal Field
  Theories}.
\newblock {\em Nucl. Phys.}, B777:316--351, 2007.
\newblock arXiv:hep-th/0701117.

\bibitem{EbeVir06}
H~Eberle and M~Flohr.
\newblock {Virasoro Representations and Fusion for General Augmented Minimal
  Models}.
\newblock {\em J. Phys.}, A39:15245--15286, 2006.
\newblock arXiv:hep-th/0604097.

\bibitem{RidPer07}
P~Mathieu and D~Ridout.
\newblock {From Percolation to Logarithmic Conformal Field Theory}.
\newblock {\em Phys. Lett.}, B657:120--129, 2007.
\newblock arXiv:0708.0802 [hep-th].

\bibitem{SchSca00}
O~Schramm.
\newblock {Scaling Limits of Loop-Erased Random Walks and Uniform Spanning
  Trees}.
\newblock {\em Isr. J. Math.}, 118:221--288, 2000.
\newblock arXiv:math.PR/9904022.

\bibitem{LawCon05}
Gregory~F. Lawler.
\newblock {\em Conformally Invariant Processes in the Plane}, volume 114 of
  {\em Mathematical Surveys and Monographs}.
\newblock American Mathematical Society, Providence, RI, 2005.

\bibitem{BauGro06}
M~Bauer and D~Bernard.
\newblock 2{D} growth processes: {SLE} and {L}oewner chains.
\newblock {\em Phys. Rep.}, 432:115--222, 2006.
\newblock arXiv:math-ph/0602049.

\bibitem{CarSLE05}
J~Cardy.
\newblock S{LE} for theoretical physicists.
\newblock {\em Ann. Phys.}, 318:81--118, 2005.
\newblock arXiv:cond-mat/0503313.

\bibitem{BauSLE03}
M~Bauer and D~Bernard.
\newblock {SLE Martingales and the Virasoro Algebra}.
\newblock {\em Phys. Lett.}, B557:309--316, 2003.
\newblock arXiv:hep-th/0301064.

\bibitem{BauCon04}
M~Bauer and D~Bernard.
\newblock Conformal transformations and the {SLE} partition function
  martingale.
\newblock {\em Ann. Henri Poincar\'e}, 5:289--326, 2004.
\newblock arXiv:math-ph/0305061.

\bibitem{KytVir07}
K~Kyt\"{o}l\"{a}.
\newblock {Virasoro Module Structure of Local Martingales of SLE Variants}.
\newblock {\em Rev. Math. Phys.}, 19:455--509, 2007.
\newblock arXiv:math-ph/0604047.

\bibitem{KytFro08}
K~Kyt\"{o}l\"{a}.
\newblock {SLE Local Martingales in Logarithmic Representations}.
\newblock {\em J. Stat. Mech.}, P08005, 2009.
\newblock arXiv:0804.2612 [math-ph].

\bibitem{RidLog07}
P~Mathieu and D~Ridout.
\newblock {Logarithmic $M \left( 2,p \right)$ Minimal Models, their Logarithmic
  Couplings, and Duality}.
\newblock {\em Nucl. Phys.}, B801:268--295, 2008.
\newblock arXiv:0711.3541 [hep-th].

\bibitem{SaiGeo09}
Y~Saint-Aubin, P~Pearce, and J~Rasmussen.
\newblock {Geometric Exponents, SLE and Logarithmic Minimal Models}.
\newblock {\em J. Stat. Mech.}, 0902:02028, 2009.
\newblock arXiv:0809.4806 [cond-mat.stat-mech].

\bibitem{SimTwi08}
J~Simmons and J~Cardy.
\newblock {Twist Operator Correlation Functions in $O(n)$ Loop Models}.
\newblock {\em J. Phys. A: Math. Theor.}, 42 235001, 2009.
\newblock arXiv:0811.4767 [math-ph].

\bibitem{GabInd96}
M~Gaberdiel and H~Kausch.
\newblock {Indecomposable Fusion Products}.
\newblock {\em Nucl. Phys.}, B477:293--318, 1996.
\newblock arXiv:hep-th/9604026.

\bibitem{RidPer08}
D~Ridout.
\newblock {On the Percolation BCFT and the Crossing Probability of Watts}.
\newblock {\em Nucl. Phys.}, B810:503--526, 2009.
\newblock arXiv:0808.3530 [hep-th].

\bibitem{RohRed96}
F~Rohsiepe.
\newblock {On Reducible but Indecomposable Representations of the Virasoro
  Algebra}.
\newblock arXiv:hep-th/9611160.

\bibitem{RohNic96}
F~Rohsiepe.
\newblock {Nichtunit\"{a}re Darstellungen der Virasoro-Algebra mit
  nichttrivialen Jordanbl\"{o}cken}.
\newblock Diploma Dissertation, Bonn University, 1996.

\bibitem{MooLie95}
R~Moody and A~Pianzola.
\newblock {\em {Lie Algebras with Triangular Decompositions}}.
\newblock Canadian Mathematical Society Series of Monographs and Advanced
  Texts. Wiley, New York, 1995.

\bibitem{AstStr97}
A~Astashkevich.
\newblock {On the Structure of Verma Module over Virasoro and Neveu-Schwarz
  Algebras}.
\newblock {\em Comm. Math. Phys.}, 186:531--562, 1997.
\newblock arXiv:hep-th/9511032.

\bibitem{KacBom88}
V~Kac and A~Raina.
\newblock {\em {Bombay Lectures on Highest Weight Representations of Infinite
  Dimensional Lie Algebras}}, volume~2 of {\em Advanced Series in Mathematical
  Physics}.
\newblock World Scientific, Singapore, 1988.

\bibitem{ItzSta89}
C~Itzykson and J-M Drouffe.
\newblock {\em {Statistical Field Theory, Volume 2: Strong Coupling, Monte
  Carlo Methods, Conformal Field Theory, and Random Systems}}.
\newblock Cambridge Monographs on Mathematical Physics. Cambridge University
  Press, Cambridge, 1989.

\bibitem{FeiVer84}
B~Feigin and D~Fuchs.
\newblock {Verma Modules over the Virasoro Algebra}.
\newblock In {\em Topology}, volume 1060 of {\em Lecture Notes in Mathematics},
  pages 230--245. Springer, Berlin, 1984.

\bibitem{CarHom56}
H~Cartan and S~Eilenberg.
\newblock {\em {Homological Algebra}}.
\newblock Princeton University Press, Princeton, 1956.

\bibitem{GabAlg03}
M~Gaberdiel.
\newblock {An Algebraic Approach to Logarithmic Conformal Field Theory}.
\newblock {\em Int. J. Mod. Phys.}, A18:4593--4638, 2003.
\newblock arXiv:hep-th/0111260.

\bibitem{BogInd74}
J~Bogn\'{a}r.
\newblock {\em {Indefinite Inner Product Spaces}}, volume~78 of {\em Ergebnisse
  der Mathematik und ihrer Grenzgebiete}.
\newblock Springer-Verlag, Berlin, 1974.

\bibitem{GurCon04}
V~Gurarie and A~Ludwig.
\newblock {Conformal Field Theory at Central Charge $c=0$ and Two-Dimensional
  Critical Systems with Quenched Disorder}.
\newblock In M~Shifman, editor, {\em From Fields to Strings: Circumnavigating
  Theoretical Physics. Ian Kogan Memorial Collection}, volume~2, pages
  1384--1440. World Scientific, Singapore, 2005.
\newblock arXiv:hep-th/0409105.

\end{thebibliography}

\end{document}